\documentclass[english,a4paper,cleveref,nosubfigcap, nolineno, thm-restate]{socg-lipics-v2021}
\hideLIPIcs
\usepackage{graphicx}
\graphicspath{{./figures/}}
\usepackage{amsmath}
\usepackage{amssymb}
\usepackage{amsfonts}
\usepackage{bbm}
\usepackage[utf8]{inputenc}
\usepackage[dvipsnames]{xcolor}
\usepackage{tikz-cd}
\tikzcdset{scale cd/.style={every label/.append style={scale=#1},
    cells={nodes={scale=#1}}}}
\usepackage{xcolor}
\usepackage{quiver}
\usepackage{tipa}
\usepackage{float}
\usepackage{pgf}
\usepackage{scalerel}
\usepackage{algorithm}
\usepackage{algpseudocode}
\usepackage{graphbox}
\usepackage{graphicx}
\usepackage{multirow}
\usepackage{thmtools}
\usepackage{mathtools}
\usepackage{booktabs}
\usepackage{caption}
\usepackage{subcaption}
\usepackage{cleveref}

\crefname{figure}{Figure}{Figures}
\crefname{section}{Section}{Sections}
\crefname{subsection}{{Section}}{{Sections}}
\crefname{subsubsection}{{Section}}{{Sections}}
\crefname{appendix}{{Appendix}}{{Appendices}}
\crefname{equation}{{Equation}}{{Equations}}
\crefname{remark}{{Remark}}{{Remarks}}
\crefname{theorem}{{Theorem}}{{Theorems}}
\crefname{definition}{{Definition}}{{Definitions}}
\crefname{lemma}{{Lemma}}{{Lemmas}}
\crefname{proposition}{{Proposition}}{{Propositions}}
\crefname{corollary}{Corollary}{Corollaries}
\crefname{table}{{Table}}{{Tables}}

\hypersetup{
    colorlinks=true,
    linkcolor=blue,
    filecolor=black,      
    urlcolor=blue,
    citecolor = blue,
    pdftitle={Mapper Interleaving Loss Computation},
    }
    
\crefname{figure}{Figure}{Figures}
\crefname{section}{Section}{Sections}
\crefname{equation}{Equation}{Equations}
\crefname{remark}{Remark}{Remarks}
\crefname{theorem}{Theorem}{Theorems}
\crefname{definition}{Definition}{Definitions}
\crefname{lemma}{Lemma}{Lemmas}
\crefname{proposition}{Proposition}{Propositions}
\crefname{table}{Table}{Tables}

\newtheorem{problem}{Problem}
\crefformat{problem}{P#2#1#3} 

\newcommand{\mybf}[1]{\textbf{\sffamily{#1}}}

\newcommand{\R}{\mathbb{R}}

\newcommand{\X}{{\mathbb X}}

\newcommand{\cC}{\mathcal{C}}
\newcommand{\cD}{\mathcal{D}}

\newcommand{\cU}{\mathcal{U}}

\newcommand{\Open}{\mathbf{Open}}
\newcommand{\Set}{\mathbf{Set}}

\newcommand{\Top}{\mathbf{Top}}

\makeatletter
\newcommand{\colim@}[2]{%
  \vtop{\m@th\ialign{##\cr
    \hfil$#1\operator@font colim$\hfil\cr
    \noalign{\nointerlineskip\kern1.5\ex@}#2\cr
    \noalign{\nointerlineskip\kern-\ex@}\cr}}%
}
\newcommand{\colim}{%
  \mathop{\mathpalette\colim@{\rightarrowfill@\textstyle}}\nmlimits@
}
\makeatother

\renewcommand{\phi}{\varphi}
\newcommand{\inv}{^{-1}}

\newcommand{\1}{\mathbbm{1}}
\renewcommand{\Im}{\mathrm{Im}}

\newcommand{\mycomment}[1]{}

\newcommand*{\Parallelogramr}[1][]{
  \pgfpicture\pgfsetroundjoin
    \pgftransformxslant{.6}
    \pgfpathrectangle{\pgfpointorigin}{\pgfpoint{.60em}{.65em}}
    \pgfusepath{stroke,#1}
  \endpgfpicture}
  
\newcommand*{\Parallelograml}[1][]{
  \pgfpicture\pgfsetroundjoin
    \pgftransformxslant{-.6}
    \pgfpathrectangle{\pgfpointorigin}{\pgfpoint{.60em}{.65em}}
    \pgfusepath{stroke,#1}
  \endpgfpicture}
  
\newcommand{\triangled}{\raisebox{\depth}{$\bigtriangledown$}}
\newcommand{\triangleu}{\bigtriangleup}

\newcommand{\Lpl}{L_{\scaleobj{.7}{\Parallelograml}}}
\newcommand{\Lpr}{L_{\scaleobj{.7}{\Parallelogramr}}}
\newcommand{\Ltd}{L_{\bigtriangledown}}
\newcommand{\Ltu}{L_{\bigtriangleup}}

\renewcommand{\phi}{\varphi}

\title{Towards an Optimal Bound for the Interleaving Distance on Mapper Graphs 
\thanks{This research was partially supported by a grant from the Department of Energy (DOE) DE-SC0021015 and grants from the National Science Foundation (NSF) DMS-2301361, CCF-1907591, CCF-2106578, CCF-2142713, CCF-2106672, and CCF-1907612.}
}

\titlerunning{Optimal Bounds for Mapper Graph Interleaving}

\author{Erin Wolf Chambers}{ University of Notre Dame, USA}{echambe2@nd.dot.edu}{https://orcid.org/0000-0001-8333-3676}{}
\author{Ishika Ghosh}{Michigan State University, USA}{ghoshis3@msu.edu}{https://orcid.org/0000-0002-7901-5912}{}
\author{Elizabeth Munch}{Michigan State University, USA}{muncheli@msu.edu}{https://orcid.org/0000-0002-9459-9493}{}
\author{Sarah Percival}{University of New Mexico, USA}{spercival@msu.edu}{https://orcid.org/0000-0003-1024-4618}{}
\author{Bei Wang}{University of Utah, USA}{beiwang@sci.utah.edu}{https://orcid.org/0000-0002-9240-0700}{}

\date{}


\begin{document}

\maketitle

\begin{abstract}
Mapper graphs are widely used tools in topological data analysis and visualization. They can be understood as discrete approximations of Reeb graphs, providing insight into the shape and connectivity of complex data. Given a high-dimensional point cloud together with a real-valued function defined on it, a mapper graph summarizes the induced topological structure: each node represents a local neighborhood, and edges connect nodes whose corresponding neighborhoods overlap. Our focus is the interleaving distance for mapper graphs, arising as a discretized analogue of the interleaving distance for Reeb graphs—a quantity known to be NP-hard to compute. This distance measures how similar two mapper graphs are by quantifying how much they must be ``stretched'' to be made comparable. Recent work introduced a loss function that gives an upper bound on this distance. The loss evaluates how far a given collection of maps, called an assignment, is from being a true interleaving. Importantly, it is computationally tractable, offering a practical way to bound the distance, however the quality of the bound is dependent on the choice of assignment.
In this paper, we develop the first framework for bounding the interleaving distance on mapper graphs. We present the bound in two ways: first, by formulating an integer linear program (ILP) that determines whether an $n$-interleaving exists for a given $n$; and second, by constructing an ILP that identifies an assignment with minimal loss for that $n$.
We also evaluate the method on small examples where the interleaving distance is known, and on benchmark and simulated datasets, demonstrating the utility of the approach for classification tasks based on mapper graphs.

\end{abstract}

\section{Introduction}
\label{sec:introduction}

Mapper graphs \cite{Singh2007} have become an increasingly popular tool in topological data analysis (TDA), generating substantial interest in both theory \cite{alvarado2024, Brown2020, Carriere2017b} and practice \cite{Purvine2023, Rathore2023, Rizvi2017, Saggar2018, zhou2023comparing}. 
Their construction yields a graphical representation of the data as follows. Let $\chi$ be a dataset equipped with a distance function $d \colon \chi \times \chi \to \R$ and a filter function $f \colon \chi \to \mathbb{R}$. Choose a cover $\cU$ of the image $f(\chi)$. 
The mapper graph is defined by assigning a vertex to each cluster obtained from the restriction of the dataset to the preimage $f^{-1}(U)$ for each $U \in \cU$. An edge is drawn between two vertices whenever the corresponding clusters have nonempty intersection.
In the continuous setting, one begins with a topological space $\X$ and a function $f \colon \X \to \R$, and considers the inverse images of the sets in $\cU$. The mapper graph is then defined as the nerve of the connected components of the inverse cover $\{f^{-1}(U) \mid U \in \cU\}$; see \cref{fig:mapper} for an illustration. From this perspective, the Mapper graph may be regarded as a discretization of a Reeb graph \cite{Reeb}, although this broader generalization is not needed in the present work.

\begin{figure}
\centering
\includegraphics[width=0.7\textwidth]{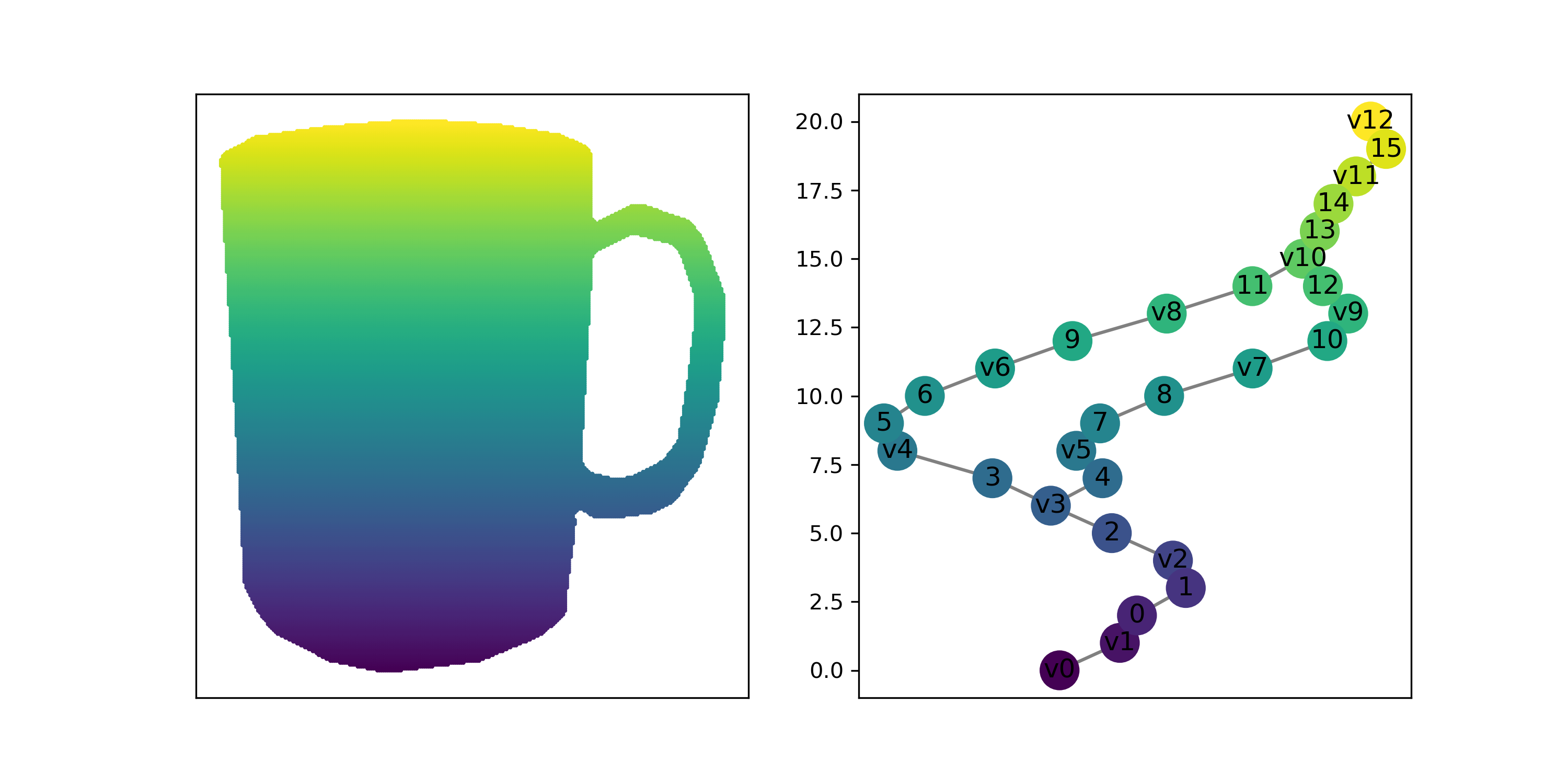}
\vspace{-2mm}
\caption{An example of a 2D input space (a \textit{cup} from the MPEG-7 dataset) together with the corresponding mapper graph, computed using a filtration function given by the pixel $y$-coordinate, is shown on the right. Vertices are labeled for reference only and do not affect the function values.}
\label{fig:mapper}
\vspace{-4mm}
\end{figure} 

A key limitation of mapper graphs is that, in practice, their use has been largely confined to qualitative assessment and exploratory data analysis. This is primarily because quantitative evaluation of mapper graph quality has so far relied mostly on theoretical guarantees. In particular, comparing an output mapper graph with a ground-truth mapper graph requires a computable similarity measure. Although several distances between mapper and Reeb graphs have been proposed (see \cref{sec:RelatedWork}), most are difficult to apply in practice due to prohibitive computational complexity. Here, we focus on the interleaving distance for mapper graphs. This construction arises from a discretization of the interleaving distance for Reeb graphs \cite{deSilva2016}, itself adapted from translating interleavings for algebraic and topological representations \cite{bubenik2015metrics, Chazal2009, deSilva2018} to graph-based objects. However, like most available metrics, it is NP-hard to compute in general \cite{bjerkevik2017computational, bubenik2015metrics}. 
 
In this paper, we use the framework given in \cite{chambers2023bounding} to bound the interleaving distance.
Formally, we encode our mapper inputs as cosheaves (see e.g.~\cite{curry2013sheaves, riehl2017category}) of the form $F\colon \Open(\cU) \rightarrow \Set$ for a given cover $\cU$.
We can define a thickening operation on the elements of $\Open(\cU)$,  which amounts to a discretized version of thickening of open intervals in $\R$, and then define the smoothing of a functor $F^n$ by post-composing with this thickening operation. 
Then an $n$-interleaving is a pair of natural transformations $\phi\colon F \to G^n$ and $\psi\colon G \to F^n$ that commute with each other up to this thickening operation. 
The interleaving distance $d_I(F,G)$ is given by the minimum $n$ for which such an interleaving can be found. 

We build on the theoretical work of \cite{chambers2023bounding}, which provides a bound on the interleaving distance for mapper graphs given a construction called an \emph{unnatural transformation}; i.e.~a collection of maps with the structure of a natural transformation but without promised commutativity.
Specifically, we call a pair $(\phi,\psi)$ an $n$-assignment if it satisfies the setup of an $n$-interleaving but potentially does not commute.
In this paper, we encode the information of these assignments in binary matrices. 
Our first major contribution is to show that verifying whether a given assignment constitutes an interleaving, as well as computing the terms appearing in the bound of \cite{chambers2023bounding}, can be reduced to a collection of matrix multiplication problems. These formulations are summarized in \cref{tab:theList} and established in \cref{thm:interleaving,thm:LossEquivMatrices}.
Then, treating the entries in the unnatural transformation matrices as variables, we set up an integer linear program (ILP) which we can pass to publicly available solvers. 
The result is the first publicly available code which provides bounds for the interleaving distance between mapper graphs, in the package \texttt{ceREEBerus} \cite{cereeberus}. 
Further, we give experiments, with code available on Github \cite{losspapergithubexperiments}, using \texttt{ceREEBerus} to compute the distance between mapper graphs in both synthetic settings and for larger benchmark datasets (See \cref{sec:experiments}).
While the ILP formulation does not guarantee that the computed bound is optimal (and hence equal to the true interleaving distance), in all experiments where the true distance is known, it recovers the correct interleaving distance.
In summary, our contributions include: 
\begin{itemize}
\item \mybf{Matrix formulation of interleavings.} We encode $n$-assignments as binary matrices and show that both interleaving verification and the computation of bounds from \cite{chambers2023bounding} reduce to a collection of matrix multiplication problems. See \cref{sec:AllMatrices}.

\item \mybf{Algorithmic characterization.} We provide explicit matrix-based criteria for checking interleavings and evaluating the associated bounds, with formal guarantees established in \cref{thm:interleaving,thm:LossEquivMatrices}.

\item \mybf{Optimization framework.} By treating the entries of the assignment matrices as variables, we formulate the problem as an ILP, enabling the use of off-the-shelf solvers. See \cref{sec:ILP}.

\item \mybf{Software contribution.} We include an implementation in the Python package \texttt{ceREEBerus} \cite{cereeberus}, the first publicly available implementation for computing bounds on the interleaving distance between mapper graphs. 

\item \mybf{Empirical validation.} We evaluate our approach on both synthetic and benchmark datasets, and in all experiments where the true interleaving distance is known, our method recovers the correct value despite the lack of optimality guarantees.
See \cref{sec:experiments}.
\end{itemize}

\section{Related Work}
\label{sec:RelatedWork} 

The mapper graph may be regarded as a discretization of the Reeb graph \cite{Reeb}, which captures the evolution of connected components of level sets of a topological space. Since a Reeb graph is, in most settings, a graph equipped with a function, one can compare Reeb graphs using standard graph distances (e.g., graph edit distance \cite{Gao2010}) by ignoring the function. However, given the topological nature of the construction, it is more desirable to employ a distance that accounts for the function.

There is now an extensive family of distances between Reeb graphs, including 
the interleaving distance \cite{Chambers2021,deSilva2016}, 
the functional distortion distance \cite{Bauer2013}, 
the universal distance \cite{Bauer2022,Bauer2020}, 
the functional contortion distance \cite{Bauer2022}, 
the bottleneck distance between the relevant persistence diagrams \cite{Carriere2017}, 
and graph edit distances \cite{Bauer2016, Bauer2020}. 
See \cite{Bollen2021,Yan2021a} for surveys. 
Work employing the interleaving distance---especially for the analysis of mapper graphs---includes convergence results showing that mapper graph converges to the Reeb graph \cite{Brown2020,Munch2016,Oulhaj2025}.
However, those works convert a mapper graph into a Reeb graph and use the continuous distance, whereas we directly construct an interleaving distance between the mapper graphs themselves.
There is also extensive work understanding how parameter choices affect the resulting mapper graph \cite{Carriere2018,Carriere2017b} and automatic parameter tuning~\cite{alvarado2025gmapper,chalapathi2021adaptive}. 

Beyond Reeb and mapper graphs, the interleaving distance appears broadly throughout the TDA literature across a wide range of input objects, originating with its introduction for persistence modules \cite{Chazal2009}. 
It has since expanded to include 
multiparameter persistence modules \cite{Berkouk2021,Lesnick2015}, 
filtered spaces \cite{Blumberg2023}, 
merge trees \cite{Beurskens2023,Morozov2013}, 
labeled Reeb graphs \cite{Lan2024},
formigrams \cite{Kim2023a}, 
labeled merge trees \cite{Curry2022,Gasparovic2024,Munch2019, Yan2019}, 
and more general categorical constructions \cite{Botnan2020,Cruz2019,deSilva2018,Kim2023,Meehan2017,Scoccola2020}.
Yet in many such settings, the interleaving distance is NP-hard to compute \cite{bjerkevik2020computing,bjerkevik2017computational}, resulting in a metric that is well-studied in theory but rarely used in practice. 

Among these works, most do not provide implementations, though there are a few notable exceptions. The interleaving distance for labeled merge trees \cite{Munch2019} lies in P, and implementations are available (e.g., \cite{Yan2019}). In \cite{Curry2022}, labelings derived from Gromov--Wasserstein couplings are used to bound the interleaving distance between unlabeled merge trees. 
In \cite{Chambers2025}, the addition metric measure space decoration on the Reeb graph is used to define a distance using the Gromov-Hausdorff distance.
Optimal transport--type distances have also been explored for mapper graphs by modeling them as hypergraphs~\cite{zhou2023comparing}, providing a practical alternative. Finally, \cite{Pegoraro2021} employs an ILP, based on the merge-tree-specific formulation of \cite{FarahbakhshTouli2019}, to obtain bounds on the merge tree interleaving distance. As with mapper graphs in this paper, existing algorithms for merge tree interleavings yield bounds rather than guaranteed exact computations.

\section{Background}
\label{sec:background}

We now present the necessary background for this paper, largely following the terminology introduced in~\cite{chambers2023bounding}.

\subsection{Functors, Cosheaves and Mapper Graphs}
\label{sec:CategoryTheory}

We work in a category-theoretic setting and provide a brief overview of the necessary definitions, with more details provided in \cite{curry2013sheaves,riehl2017category}. 

A \emph{category} $\cC$ is a collection of objects $X,Y,Z,\dots$ along with morphisms $f,g,h,\dots$ such that 
\begin{enumerate}
    \item each morphism $f$ has a designated domain $X$ and codomain $Y$,
    \item every object has an identity morphism $\1_X\colon X \to X$,
    \item any pair of morphisms $f\colon X \to Y$ and $g\colon Y \to Z$ has a composite morphism $g \circ  f\colon X \to Z$ which satisfies an identity axiom, where $f\colon X \to Y = \1_Y \circ  f = f \circ \1_X$, and an associativity axiom, where $h\circ (g\circ f) = (h\circ g)\circ f$.
  \end{enumerate}
Many settings in mathematics can be viewed as categories: $\Set$ is the category whose objects are sets and morphisms are set maps; 
$\Top$ is the category whose objects are topological spaces and morphisms are continuous functions;
and, for a given topological space $X$, $\Open(X)$ is the category whose objects are the open sets and morphisms are given by inclusion. A category is \emph{small} if the collections of objects and morphisms are both sets.

Given two categories $\cC$ and $\cD$, a \emph{functor} $F\colon\cC \to \cD$ maps each object $x \in \cC$ to an object $F(x) \in \cD$ and each morphism $f\colon X \to Y$ to a morphism $F[f]\colon F(X) \to F(Y)$ that satisfies the following properties: for any $X \in \cC$, $F[\1_X] = \1_{F(X)}$; for any $f,g \in \cC$ for which the composition $gf$ is defined, we have $F[gf] = F[g] F[f]$. 
Given two functors $F,G\colon \cC \to \cD$, a \emph{natural transformation} $\eta\colon F \Rightarrow G$ is a collection of maps $\eta_X\colon F(X) \to G(X)$ such that for any morphism $f\colon X \to Y$ in $\cC$, the diagram 
\begin{equation*}
\begin{tikzcd}
X\ar[d,  "f"] 
& F(X) 
    \ar[r, "\eta_X"] 
    \ar[d, "{F[f]}"']
& G(X) 
    \ar[d, "{G[f]}"]
\\
Y 
& F(Y) \ar[r, "\eta_Y"] 
& G(Y)
\end{tikzcd}
\end{equation*}
commutes. 
One example of a functor is $\pi_0\colon\Top \to \Set$, where $\pi_0(X)$ denotes the set of path-connected components of a topological space $X$, and the morphisms are set maps $\pi_0[f]\colon \pi_0(X) \to \pi_0(Y)$ sending a path-connected component $A$ in $X$ to the element $f(A)$ in $Y$.

In the special case where $J$ is a small category, a diagram is simply a functor $F\colon J \to \cC$.
Let $c \in \cC$. By a slight abuse of notation, we define the constant functor $c \colon J \to \cC$, $c(j) \mapsto c$ for all $j \in J$ and all morphisms are sent to the identity morphism. 
In this setting, a \emph{cocone} is a natural transformation $\lambda\colon F \to c$.
We denote $\lambda_j\colon F(j) \to c$ 
and note that for any $f\colon j \to k$ in $J$, $\lambda_k \circ F[f] = \lambda_j$. 
A cocone $\lambda\colon F \to c$ is a \emph{colimit} if, for any other cocone $\lambda'\colon F \to c'$, there is a unique  morphism $u\colon c \to c'$ such that the diagram
\begin{equation*}
\begin{tikzcd}
F(j) 
    \ar[rr, "{F[f]}"]
    \ar[dr, "\lambda_j"]
    \ar[ddr, "\lambda_j'"']
&& F(k)
    \ar[dl, "\lambda_k"']
    \ar[ddl, "\lambda_k'"]
\\
& c \ar[d, dashed, "u"] \\
& c'
\end{tikzcd}
\end{equation*}
commutes for all $f\colon j \to k$ in $J$. 
We denote this by $\colim_I(F)$, and note that this always exists and is unique in our setting where $\cD =\Set$.

For the rest of this paper, we turn our attention to functors of the form $F\colon\Open(X) \to \Set$, called \emph{pre-cosheaves}. 
We are particularly interested in the case when a pre-cosheaf is a \emph{cosheaf}. In short, a cosheaf is a presheaf that is entirely and uniquely defined by any open cover of ${U \in \Open(X)}$.
More rigorously, given an object $U \in \Open(X)$ and a cover $\{U_\alpha\}_{\alpha \in A}$ of $U$, define a category $\cU = \{U_\alpha \cap U_{\alpha'} \mid \alpha,\alpha' \in A \}$ with morphisms given by inclusion. As these open sets $U_\alpha$ are themselves objects in $\Open(X)$, the functor $F\colon\cU \to \Set$ is well-defined, so we can consider its colimit $\lambda\colon F \to L$. A \emph{cosheaf} is a functor $F$ whose unique map $L \to F(U)$ given by the colimit definition is an isomorphism for every $U \in \Open(X)$.

\subsection{Mapper Cosheaves and the Interleaving Distance}
\label{ssec:functorGraphs}

In this section, we show how we can represent the data of a mapper graph as a cosheaf. 
We assume that we are given data in the form of a compact topological space with a function $f\colon\X \to \R$. 
Fix a resolution $\delta>0$. 
We assume $\Im(f) \subseteq[-\delta L, \delta L]$.
This choice of resolution induces a discretization of $[-\delta L, \delta L]$ into a cell complex $K$ where the 0-cells are given by the points $\{ \sigma_i \coloneqq i \delta \mid i \in [-L, \cdots,L]\}$, and edges are given by the intervals $\{ \tau_i \coloneqq (i \delta, (i+1)\delta) \mid i  \in [-L, L-1] \}$. Further, $K$ induces a cover of the image with cover elements of the form 
    \[\{ U_{\sigma_i} = ((i-1)\delta, (i+1)\delta) \mid i \in \{-L+1,\cdots, L-1\}\} \]
and intersections of the form
    \[\{ U_{\tau_i} = (i\delta, (i+1)\delta) \mid i \in [-L,\cdots, L-1]\}.\]
Together, we write $\cU= \{U_{\sigma_i}\} \cup \{ U_{\tau_i}\}$. 
The set $\cU$ forms a poset $(\cU, \subseteq)$ under inclusion. 
We use $\rho \in K$ to denote either a $\sigma_i$ or a $\tau_i$, and similarly write $U_\rho \in \cU$ when referring to cells of either dimension.

We define open sets in this poset using the Alexandrov topology, following \cite{Barmak2011} and \cite{chambers2023bounding}.
Given this poset $(\cU,\subseteq)$, for any set $S \subseteq \cU$, the upset is 
$S^{\uparrow} = \{U \in \cU \mid \exists \,V \in S,\, V \subseteq U  \,  \}$ 
and the downset is
$S^{\downarrow} = \{ U \mid \exists\,  V \in S,\, U \subseteq V \}$. 
A set $S \subset \cU$ is said to be open iff $S^{\downarrow} = S$ or equivalently, if for any $U \in S$, and any $V \subset U$, $V \in S$. 
It can be checked that this topology has a basis given by the set of downsets: 
write $S_\rho = \{ U_{\rho}\}^{\downarrow}$, then the basis for $\Open(\cU)$ is given by $\{ S_\rho^{\downarrow} \mid \rho \in K\}$ and we call $S_\rho$ a basic open set. 
See the left of \cref{fig:PosetNotation} for a visualization of this notation. 

We will often use these combinatorial subsets of $\cU$ interchangeably with their geometric realizations, given by $|S| \coloneqq \bigcup_{U \in S} U \subset \R$. 
We can check that this notation is then reasonable on our  set. 
First, for $\sigma_i = i \delta$ and $U_{\sigma_i} = ((i-1)\delta, (i+1)\delta) \subset \R$, we have 
${\{U_{\sigma_i}\}^{\downarrow} =\{U_{\tau_i},U_{\sigma_i},  U_{\tau_{i+1}}\}}$ 
and $|\{U_{\sigma_i}\}^{\downarrow}| = U_{\sigma_i}$. 
Similarly, for $U_{\tau_i} = (i\delta, (i+1)\delta)$, 
$\{U_{\tau_i}\}^{\downarrow} =\{U_{\tau_i}\} $, and of course 
$|\{U_{\tau_i}\}^{\downarrow}| = U_{\tau_i}$.
Again, see \cref{fig:PosetNotation} for an illustration of the notation. 

\begin{figure}
    \centering
    \includegraphics[width=\textwidth]{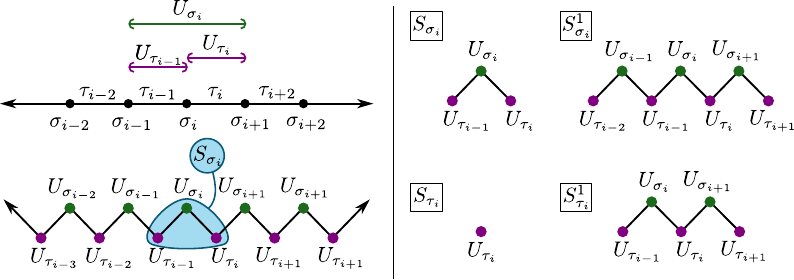}
    \caption{The notation for the cover of $\R$. Left: the distinction between the cover element $U_{\sigma}$ and the open set in the Alexandrov topology, $S_{\sigma}$, which is a discrete set of three objects. Right: 1-thickenings $S_{\sigma_i}$ and $S_{\tau_i}$ for the two types of basis open sets. (Left side of figure modified from \cite{chambers2023bounding}). } 
    \label{fig:PosetNotation}
\end{figure}

Let $\pi_0\colon \Top \to \Set$ denote the functor returning the set of path connected components of a given topological space. 
Then the mapper cosheaf of $(\X,f)$ at resolution $\delta$ is given by the functor 
\begin{equation*}
\begin{matrix}
    F\colon & \Open(\cU) & \rightarrow & \Set \\
       & S & \mapsto & \pi_0(f\inv(|S|))
\end{matrix}
\end{equation*}
By functoriality of $\pi_0$, for open sets $S \subseteq T$, $F$ induces a map
\[F[S \subseteq T] \colon \pi_0f\inv(|S|) \to \pi_0f\inv(|T|),\]
so that $F$ satisfies the requirements of being a functor.
Furthermore, this functor satisfies the definition of a cosheaf. 
Thus, we call any cosheaf of the form $F\colon \Open(\cU) \to \Set$ a \textit{mapper cosheaf}. 
For notational simplicity, we write the induced map as $F[\subseteq]\colon F(S) \to F(T)$ when $S \subseteq T$ is clear from context. 

We next show that we can define an interleaving on these mapper cosheaves.
Given any open set $S \in \Open(\cU)$, we construct the \emph{1-thickening} by taking the downset of the upset of $S$, denoted by
$S^1 = (S^{\uparrow})^\downarrow$.
The \emph{$n$-thickening} is defined to be $n$ repetitions of this process, given recursively as
\begin{equation*}
    S^n = 
    \begin{cases}
    S & n = 0;\\
    (S^{n-1})^{\uparrow\downarrow} & n \geq 1.
    \end{cases}
\end{equation*}
Each $S^n$ is an element of $\Open(\cU)$, if $S \subseteq T$, then $S^n \subseteq T^n$. 
Thus we can view the thickening operation as a functor $(-)^n\colon \Open(\cU)   \to   \Open(\cU)$.
It is easy to check  that for a basic open $S_{\sigma_i}$, $|S_{\sigma_i}^n| = ((i-n)\delta, (i+n)\delta)$ and for $S_{\tau_i}$, $|S_{\tau_i}^n| = ((i-n)\delta, (i+1+n)\delta)$. 
See the right of \cref{fig:PosetNotation} for an example of 1-thickenings of basic opens.

As shown in previous work \cite{chambers2023bounding}, this thickening can be used to define an interleaving distance that compares mapper cosheaf functors  
${F,G\colon\Open(\cU) \to \Set}$.
We use $F^n$ to denote the composition of functors $F \circ ( - )^n\colon \Open(\cU) \to \Set$, so that $F^n(S) = F(S^n)$, and similarly for $G^n$. 
We define an \emph{$n$-interleaving} to be a pair of natural transformations $\phi\colon F \Rightarrow G^n$ and $\psi\colon G \Rightarrow F^n$ given by set maps $\phi_S\colon F(S) \to G(S^n)$. 
We note the map at $S^n$, $\phi_{S^n}\colon F(S^n) \to G(S^{2n})$, can also be viewed as a component map of a natural transformation $\phi^n\colon F^n \Rightarrow G^{2n}$. Therefore, when $\phi$ is indeed a natural transformation, we use the notation $\phi_{S^n}$ and $\phi_S^n$ interchangeably.

\begin{definition}[\cite{chambers2023bounding}]
\label{def:interleavingDistance}
Let $F,G\colon\Open(\cU) \to \Set$ be cosheaves and $n \in \mathbb{Z}_{\geq 0}$. An \emph{$n$-interleaving} is a pair of natural transformations 
$\phi\colon F \Rightarrow G^n$
and 
$\psi\colon G \Rightarrow F^n$
such that the diagrams 
\begin{equation*}
    \begin{tikzcd}[ampersand replacement=\&]
        F(S) 
            \ar[rr, "{F[S \subseteq S^{2n}]}"]   
            \ar[dr, "\phi_S"',violet] 
            \& \& F(S^{2n}) \& 
        \& F(S^n) \ar[dr]
            \ar[dr, "\phi_{S^{n}}",violet]
        \& \\
        \& G(S^n)\ar[ur, "\psi_{S^n}"', orange]  \& \& 
        G(S) 
            \ar[rr, "{G[S \subseteq S^{2n}]}"']
            \ar[ur, "\psi_{S}", orange]  
        \&\& G(S^{2n})
    \end{tikzcd}
\end{equation*} 
commute for all $S \in \Open(\cU)$.
The \emph{interleaving distance} is given by 
\begin{equation*}
    d_I(F,G) = \inf\{ n \geq 0 \mid \text{there exists an $n$-interleaving} \}, 
\end{equation*}
and is set to be $d(F,G) = \infty$ if there is no interleaving for any $n$.
\end{definition}

For the sake of simplicity, we will denote these diagrams by $\triangled_{\phi,\psi}(S)$ and $\triangleu_{\phi,\psi}(S)$, respectively. Similarly, we denote the diagrams representing natural transformations
\begin{equation*}
    \begin{tikzcd}
        F(S)  
            \ar[r, "{F[\subseteq ]}"] 
            \ar[dr, "\phi_S"', very near start, violet]
        & F(T)
            \ar[dr, "\phi_T"', very near start, violet]
        & \\
        & G(S^n) 
            \ar[r, "{G[\subseteq ]}"'] 
        & G (T^n)
    \end{tikzcd}
    \hspace{1cm}
    \begin{tikzcd}
        & F(S^n)
            \ar[r, "{F[\subseteq ]}"] 
        & F (T^n)\\
        G(S)
            \ar[r, "{G[\subseteq ]}"'] 
            \ar[ur, "\psi_S", very near start, orange]
        & G(T) 
            \ar[ur, "\psi_T", very near start, orange]
        & 
    \end{tikzcd}
\end{equation*}
by $\Parallelograml_\phi(S,T)$ and $\Parallelogramr_\psi(S,T)$, respectively. 
As shown in \cite{chambers2023bounding}, \cref{def:interleavingDistance} fits in the interleaving framework built by Bubenik et al.~\cite{bubenik2015metrics} and thus it is an extended pseudometric.

\subsection{Loss Function Definition}
\label{ssec:LossFunction}

In this section, we outline the loss function, largely following \cite{chambers2023bounding}, with some necessary technical modifications.
First, we introduce the notion of an assignment, inspired by~\cite{Robinson2020} and~\cite{nlab:unnatural_transformation}, which we use to denote collections of morphisms used to define an interleaving, without requiring that these morphisms form a natural transformation.

\begin{definition}[Modified from \cite{chambers2023bounding}]
\label{def:assignment}
Given functors $H,H'\colon\Open(\cU) \to \Set$, an \emph{unnatural transformation}\footnote{A  natural transformation is an unnatural transformation which obeys commutativity properties, so the two sets are not mutually exclusive.}   $\eta\colon H \rightarrow H'$ is a collection of maps $\eta_S\colon H(S) \to H'(S)$ with no additional promise of commutativity. 
For a fixed $n \geq 0$ and cosheaves $F$ and $G$, an \emph{$n$-assignment} is a pair of unnatural transformations $\phi\colon F \Rightarrow G^n$ and $\psi\colon G \Rightarrow F^n$.

An \emph{(extended) basis unnatural transformation} 
$\eta\colon H \leadsto (H')^n$ for mapper cosheaves $H$ and $H'$ is a collection of maps
\begin{equation*}
    \eta = \{ \eta_{S_\rho}\colon H(S_\rho) \to H'(S_\rho^n) \mid \rho \in K\} 
        \cup \{\eta_{S_\rho^n}\colon H(S_\rho^n) \to H'(S_\rho^{2n}) \mid \rho \in K\}
\end{equation*}
for all basis elements $S_\rho$ from $\rho \in K$. 
A \emph{basis $n$-assignment} (or simply a \emph{basis assignment}) is a pair of extended basis unnatural transformations
$
\phi\colon F \leadsto G^n
$
and 
$
\psi\colon G \leadsto F^n.
$
\end{definition}
The modification from \cite{chambers2023bounding} presented here is to include the maps $\eta_{S^n_\sigma}$ along with $\eta_{S_\sigma}$ as part of the basis $n$-assignment. 
In \cite{chambers2023bounding}, we used the fact that the maps $\eta_{S_\sigma}$ could be used to determine $\eta_{S^n_\sigma}$; however the matrix representations introduced in \cref{sec:DataStructure} do not allow for easy computation of this extension. 
Thus, we have modified the definition so that the maps $H(S_\sigma^n) \to H(S_\sigma^{2n})$ are also given as input. 
For the remainder of the paper, we use the term ``basis unnatural transformation'' to imply the extended definition given here. 

We take these basis unnatural transformations and extend them to natural transformations after a shift using a value $k$.
We adopt the following notation:
when an assignment might not commute, we  denote its maps by lower case $\phi \colon F \leadsto G^n$ and $\psi\colon G \leadsto F^n$ with squiggly arrows; for assignments which are constructed to be natural transformations, we  denote them by upper case $\Phi\colon F \Rightarrow G^n$ and $\Psi\colon G \Rightarrow F^n$ with double arrows. 
\begin{definition}[\cite{chambers2023bounding}]
\label{def:extend_unnat_trans}
Given a basis unnatural transformation $\phi\colon F \leadsto G^n$ and a fixed $k$, define the (non-extended) basis unnatural transformation $\Phi\colon F \Rightarrow G^{n+k}$ by 
$\Phi_S =  G[S^n \subseteq S^{n+k}] \circ \phi_S$
for all $S \in \{S_\sigma \mid \sigma \in K\}$. %
\end{definition}

We measure the quality of an  $n$-assignment $\phi, \psi$ using  collections of distances for all open sets, $\{d_S^F \mid S \in \Open(\cU)\}$ and $\{d_S^G \mid S \in \Open(\cU)\}$, defined as follows.  
\begin{definition}[\cite{chambers2023bounding}]
The distance $d_S^{F}(A,B)$ for $A,B \in F(S)$ is defined to be 
\begin{equation}
\label{eqn:DistanceDefinition}
    d_S^{F}(A,B) = 
    \min\{ k\geq0 \mid F[S \subset S^k](A) = F[S \subset S^k](B)\}.
\end{equation}
If no such $k$ exists, then we  set $d_S^{F}(A,B) = \infty$.
\end{definition}
Geometrically, this can be thought of as the smallest integer $n$ such that the connected components represented by $A$ and $B$ in $f\inv(|S|)$ are in the same connected component when included into $f\inv(|S^n|)$. 
Further, we can see that this distance is an ultrametric since if for ${A,B,C \in F(S)}$, $k = \max\{d_S^{F}(A,B), d_S^{F}(B,C)\}$, then $A$, $B$, and $C$ all map to the same object in $F(S^k)$, and thus $d_S^F(A,C) \leq k$. 
\begin{definition}[\cite{chambers2023bounding}]
\label{df:Loss_by_diagram}
Fix an $n$-assignment
$(\phi,\psi)$. 
We define four \emph{diagram loss functions}: 
\begin{align*}
\Lpl^{S,T}(\phi)
    &= \max\limits_{\alpha \in F(S)} 
       d_{T^n}^{G}(
        \varphi_T \circ F[S \subseteq T](\alpha),
        G[S^n \subseteq T^n] \circ \varphi_S(\alpha)
        );\\
\Lpr^{S,T} (\psi)
    &= \max\limits_{\alpha \in G(S)} 
       d_{T^n}^{F}(
        \psi_T \circ G[S \subseteq T](\alpha), 
        F[S^n \subseteq T^n] \circ \psi_S(\alpha)
    );\\
\Ltd^S (\phi,\psi)
    &= \max\limits_{\alpha \in F(S)}  \Big \lceil \tfrac{1}{2} \cdot d_{S^{2n}}^{F}(
    F[S \subseteq S^{2n}] (\alpha),
    \psi_{S^n} \circ \varphi_S(\alpha)
    ) \Big \rceil;\\
\Ltu^S (\phi,\psi)
    &= \max\limits_{\alpha \in G(S)}\Big \lceil \tfrac{1}{2} \cdot d_{S^{2n}}^{G}(
    G[S \subseteq S^{2n}](\alpha),
    \varphi_{S^n} \circ \psi_S(\alpha)
    )\Big \rceil.
\end{align*}
\end{definition}
Combining these ideas, we define the loss function between given $n$-assignments.

\begin{definition}[Modified from \cite{chambers2023bounding}]
\label{shortdef}
Given an $n$-assignment $(\phi,\psi)$, the \emph{(extended basis) loss function} is 
\begin{equation}
\label{eq:loss}
L_B(\phi,\psi) = \max_{ \substack{\sigma < \tau \in K \\ \rho \in K}}
\left\{
\Lpl^{S_\tau, S_\sigma}, \Lpr^{S_\tau, S_\sigma},
\Lpl^{S_\rho, S_\rho^n}, \Lpr^{S_\rho, S_\rho^n},
\Ltu^{S_\rho}, \Ltd^{S_\rho}
\right\}
\end{equation}
where the max is over all cells $\rho \in K$ of any dimension, or adjacent edge-vertex pairs $(\tau, \sigma)$. 
\end{definition}

Again, note the slight modification from the definition given in  \cite{chambers2023bounding} where two additional terms, $\Lpl^{S_\rho, S_\rho^n}$ and $\Lpr^{S_\rho, S_\rho^n}$, are included in this new loss function. 
Because the extended basis unnatural transformation provides the maps 
$\phi_{S_\sigma^n}\colon F(S_\sigma^n) \to G(S_\sigma^{2n})$ as inputs rather than being determined from the maps $\phi_{S_\sigma}$, we are taking a maximum over more inputs so it is an immediate corollary that this function gives a bound on the true interleaving distance. 

\begin{theorem}[{Corollary of \cite[Theorem 3.16]{chambers2023bounding}}]
\label{maintheorem}
Given an (extended) basis $n$-assignment  
$\phi$
and 
$\psi$, 
we have 
$d_I(F,G) \leq n + L_B(\phi,\psi).$
\end{theorem}

\section{Problems} 
\label{sec:Problems}

Ultimately, the goal of this paper is to solve the following problem. 

\begin{problem}
\label{prob:ComputeInterleavingDist}
\emph{Compute the interleaving distance.}
\begin{itemize}
    \item \emph{Input:}  A pair of mapper graphs represented by cosheaves  $F$ and $G$. 
    \item \emph{Question:} Compute the interleaving distance, $d_I(F,G)$. 
\end{itemize}
\end{problem}

We will give further details on the input structure in \cref{sec:DataStructure}. 
For the moment, however, following standard practice in complexity theory, we focus on the binary decision version of the problem.

\begin{problem} 
\label{prob:BinaryDecision} 
\emph{Binary decision problem.}
\begin{itemize}
    \item \emph{Input:} Fixed $n\geq 0$ and a pair of mapper graphs represented by cosheaves  $F$ and $G$. 
    \item \emph{Question:} Is there an $n$-interleaving between $F$ and $G$? 
\end{itemize}
\end{problem}

Of course, searching for a polynomial time algorithm to answer \cref{prob:ComputeInterleavingDist} is a fool's errand. 
By \cite{bjerkevik2017computational}, \cref{prob:BinaryDecision} is graph isomorphism complete when $n=0$; and is NP-complete when $n >0$. 
Still, if we assume we can answer the binary decision problem \cref{prob:BinaryDecision}, then we can use a standard binary search algorithm to answer \cref{prob:ComputeInterleavingDist} as follows. 
Initially we find a plausible interval  within which to search for $n$ by doubling $n \in \{0, 1, 2, 4,\cdots\}$ iteratively until \cref{prob:BinaryDecision} returns \textit{Yes}.
That is, for $A = 2^{j-1}$, P\ref{prob:BinaryDecision} returns \textit{No} and for $B = 2^j$, P\ref{prob:BinaryDecision} returns \textit{Yes}; then we will use $[A,B]$ as our initial interval.
Then, we run standard binary search on the interval $[A,B]$ using the response to P\ref{prob:BinaryDecision} to determine whether to replace the interval with the top or bottom half. 
This is given in \cref{alg:BinarySearch}, where we use \texttt{is\_interleaved}$(F,G, n)$ for the subroutine that answers P\ref{prob:BinaryDecision}. 

\begin{algorithm}
\caption{Binary Search For $N$}
\label{alg:BinarySearch}
\begin{algorithmic}[1]
\Procedure{BinarySearch}{$F,G$}
    \State $i=0$
    \While {\texttt{is\_interleaved}$(F,G, 2^i)$ is False} 
        \State {$i \gets i+1$}
    \EndWhile
    \State $ A \gets max(0, 2^{i-1} + 1)$, $B \gets 2^i$
    \While{$A \neq B$}
        \State $n \gets \lfloor (A + B) / 2 \rfloor$
        \If{\texttt{is\_interleaved}$(F,G, n)$ is False}
            \State $A \gets n+1$ \Comment{Next interval is $[n+1,B]$}
        \Else
            \State  $B \gets n$ \Comment{Next interval is $[A,n]$}
        \EndIf
    \EndWhile
    \State \Return $A$ \Comment{Result is $d_I(F,G) \leq A$}
\EndProcedure
\end{algorithmic}
\end{algorithm}

Thus, we focus this paper on building an optimization framework to provide a solution to \cref{prob:BinaryDecision}. 
We can proceed by searching, given a fixed $n$, for an $n$-assignment $(\phi,\psi)$ until we find a pair that forms an $n$-interleaving. 
However, the framework built up in the previous section gives an additional viewpoint for this search. 
We can instead look for an $n$-assignment $(\phi,\psi)$ that minimizes the loss function of \cref{shortdef} used in the bound in \cref{maintheorem}, thereby making the bound on the interleaving distance as small as possible. 
This is exemplified in the following problem. 

\begin{problem} 
\label{prob:LossComputation} 
\emph{Loss computation problem.}
\begin{itemize}
    \item \emph{Input:} Fixed $n\geq 0$ and a pair of mapper graphs represented by cosheaves  $F$ and $G$. 
    \item \emph{Question:} Find the value $L = \inf_{\phi,\psi}L_B(\phi,\psi)$ over all $n$-assignments $\phi$, $\psi$. This implies that $d_I(F,G) \leq n+L$.  
\end{itemize}
\end{problem}

First, note that there is a polynomial time reduction from  \cref{prob:LossComputation}  to \cref{prob:BinaryDecision}. 
Indeed, there is an $n$-interleaving iff $\inf_{\phi,\psi}L_B(\phi,\psi) = 0$. 
This implies that since \cref{prob:BinaryDecision} is NP-hard, \cref{prob:LossComputation} is NP-hard as well (specifically, GI-complete for $n=0$, NP-complete for $n >0$). 
Thus, looking for a polynomial time algorithm for \cref{prob:LossComputation} is again unlikely to result in success (unless $P=NP$) so we shall not try. 

What we can do in polynomial time is the following. 
If we fix the input $n$-assignment, we have a modified version of \cref{prob:BinaryDecision} as follows.

\begin{problem} 
\label{prob:BinaryDecisionFixedAssignment} 
\emph{Binary decision problem for fixed $\phi$ and $\psi$.}
\begin{itemize}
    \item \emph{Input:} Fixed $n\geq 0$, a pair of mapper graphs represented by cosheaves  $F$ and $G$, and an $n$-assignment $\phi$, $\psi$. 
    \item \emph{Question:} Is $(\phi,\psi)$ an $n$-interleaving between $F$ and $G$? 
\end{itemize}
\end{problem}

Similarly, for a fixed $\phi$ and $\psi$, we have a modified version of \cref{prob:LossComputation} as follows. 

\begin{problem} 
\label{prob:LossComputationFixedAssignment} 
\emph{Loss computation problem for fixed $\phi$ and $\psi$.}
\begin{itemize}
    \item \emph{Input:} Fixed $n\geq 0$, a pair of mapper graphs represented by cosheaves  $F$ and $G$, and an $n$-assignment $\phi$, $\psi$. 
    \item \emph{Question:} Compute $L = L_B(\phi,\psi)$, and thus $d_I(F,G) \leq n+L$.  
\end{itemize}
\end{problem}

It was shown in \cite{chambers2023bounding} that \cref{prob:LossComputationFixedAssignment} can be solved in polynomial time. Thus, the primary source of computational complexity lies in searching for an optimal assignment, and hence the lowest bound. To address this, we develop an optimization formulation that searches for a best assignment to answer \cref{prob:BinaryDecision} or \cref{prob:LossComputation}; when combined with binary search, this also yields a solution to \cref{prob:ComputeInterleavingDist}.

The resulting loss value can, in a single run, provide a na\"ive bound on the distance between the input mapper graphs. However, as we show below, there are examples in which the minimum achievable loss for a fixed $n$ does not yield a tight bound on the interleaving distance.

\begin{example}
\label{ex:binarySearchExample}

\begin{figure}
    \centering
    \includegraphics[width=0.3\linewidth, align = c]{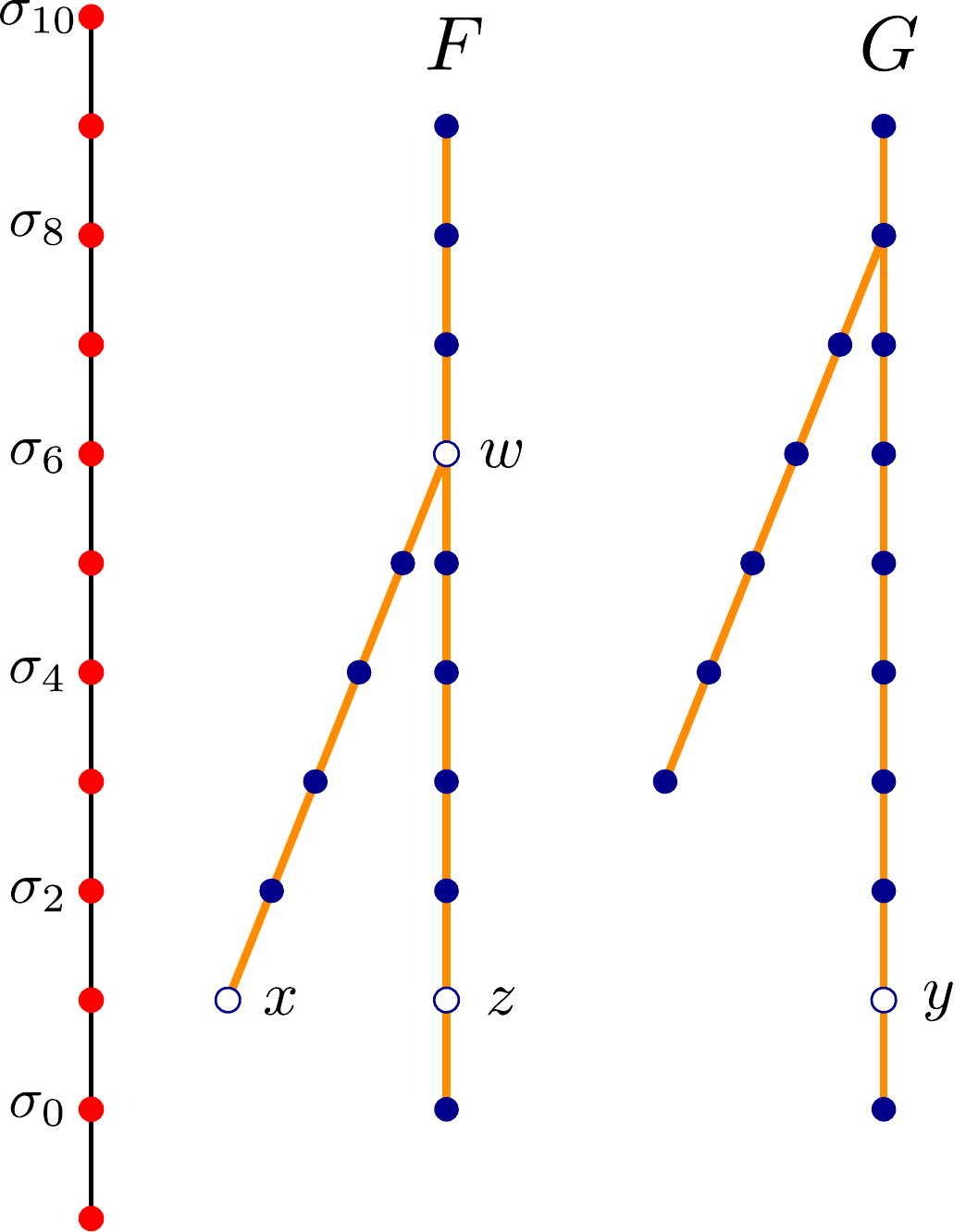}
    \qquad
    \includegraphics[width=0.4
    \linewidth, align = c]{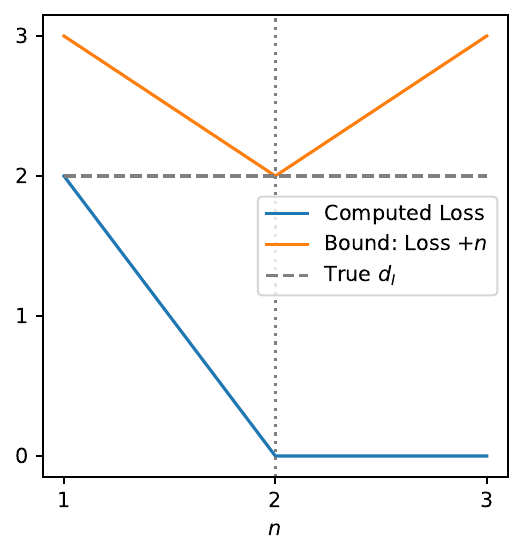}
    \caption{(Left) An example of a pair of mapper graphs where the binary search on $n$ is required. (Right) The computed loss and interleaving bound for this example. }
    \label{fig:BinarySearchExample}
\end{figure}

In this example, we present a pair of graphs and a choice of $n$ for which every assignment satisfies $d_I \lneq n + L_B(\phi,\psi)$, demonstrating that a loss computation for a single choice of $n$ is not enough and therefore a binary search over $n$ is indeed necessary. 
Consider the two mapper graphs shown on the left of \cref{fig:BinarySearchExample} and fix $n = 1$.
Focusing on the vertices $x$ and $z$, note that the only feasible assignment is $\phi(x) = \phi(z) = [y]$, where $[y]$ denotes the connected component containing $y$ in the 1-smoothing of $G$, i.e., $[y] \in G(S_{\sigma_1}^1)$. In the opposite direction, the vertex $y$ has two possible assignments, so $\psi(y)$ may be either $[x]$ or $[z]$ in $F(S_{\sigma_1}^1)$.
Then in the diagram 
\begin{equation*}
\begin{tikzcd}
    F(S_{\sigma_1}) \ar[rr] \ar[dr] 
    && F(S_{\sigma_1}^2) 
    & \substack{x\\z} \ar[rr] \ar[dr] && \substack{{[x]}\\{[z]}}
    \\
    & G(S_{\sigma_1}^1) \ar[ur] 
    &&& {[y]}\ar[ur, dashed, "?"]
\end{tikzcd}
\end{equation*}
either choice for $y$ will result in the other object not commuting. 
Then the loss is given by 
\[
k = \lceil d_F^{S_{\sigma_i}^2}([x],[y]) / 2 \rceil = \lceil 3/2 \rceil = 2,
\]
so the interleaving bound is $n+k = 1+2 = 3$. 
However, we can build an interleaving for $n =2$ by sending the left tail of $F$ to the left tail of $G$, showing that the true interleaving $d_I(F,G) = 2$ is strictly less than any possible bound found by the solver for a single choice of $n$.
This can be remedied, however, by running a binary search over $n$, which we implement.

\end{example}

Because the above example shows we might not get a tight bound by only optimizing the assignments for a given $n$, we take the following modification. 
At the expense of higher computation time per run, we have a second version of binary search where we use the output of P\ref{prob:LossComputation} to its greatest extent. 
As before, we use an exponential search to find an initial interval $[A,B]$. 
However, we modify the binary search as follows. 
Given the midpoint $n \in [A,B]$, we use P\ref{prob:LossComputation} to compute the loss $L$. 
If $L =0$, this is equivalent to finding an $n$-interleaving, and so the next interval to search will be $[A,n]$. 
If on the other hand, $L>0$, then algorithm has found an $n+L$ interleaving, and the next interval to search will be $[n+1, \min\{n+L, B\}]$. 
This modified binary search is given in \cref{alg:ConstrainedBinarySearch}, where we use \texttt{loss}$(F,G,n)$ to represent the output of P\ref{prob:LossComputation}. 

\begin{algorithm}
\caption{Constrained Binary Search For $N$}
\label{alg:ConstrainedBinarySearch}
\begin{algorithmic}[1]
\Procedure{ConstrainedBinarySearch}{$F,G$}
    \State $i=0$
    \While {\texttt{is\_interleaved}$(F,G, 2^i)$ is False} 
        \State {$i \gets i+1$}
    \EndWhile
    \State $A \gets max(0,2^{i-1}+1)$, $B \gets 2^i$
    \While{$A \neq B$}
        \State $n \gets \lfloor (A + B) / 2 \rfloor$
        \State $L \gets \texttt{loss}(F,G,n)$
        \If{$L >0$}
            \State $A \gets n+1$, $B \gets \min\{n+L, B\}$ \Comment{Next interval is $[n+1,\min\{n+L, B\}]$}
        \Else \Comment{$L = 0$}
            \State  $B \gets n$ \Comment{Next interval is $[A,n]$}
        \EndIf
    \EndWhile
    \State \Return $A$ \Comment{Result is $d_I(F,G) \leq A$}
\EndProcedure
\end{algorithmic}
\end{algorithm}

In \cref{sec:DataStructure}, we give explicit constructions of the representations of $F$ and $G$ as graphs, and use these to develop optimization algorithms for P\ref{prob:BinaryDecision} and P\ref{prob:LossComputation}. We then combine these solutions with binary search, as described above, to obtain the smallest available upper bound $d_I(F,G) \leq N$.

While we cannot prove that this $N$ equals the interleaving distance---unsurprising given that the problem is NP-hard and we do not resolve whether $\mathrm{P}=\mathrm{NP}$---our experiments (\cref{sec:experiments}) show that this approach recovers the interleaving distance in several cases where the true value is known.

\section{Data Structures}
\label{sec:DataStructure}

In this section, we show how a mapper cosheaf $F:\Open(\cU)\to\Set$ can be represented as a graph and how the unnatural transformations $\phi$ and $\psi$ can be encoded as binary block matrices. We then demonstrate that the entries of the basis loss function arise from matrix multiplication.  

\begin{figure}%
    \centering
    \includegraphics[width=0.5\linewidth]{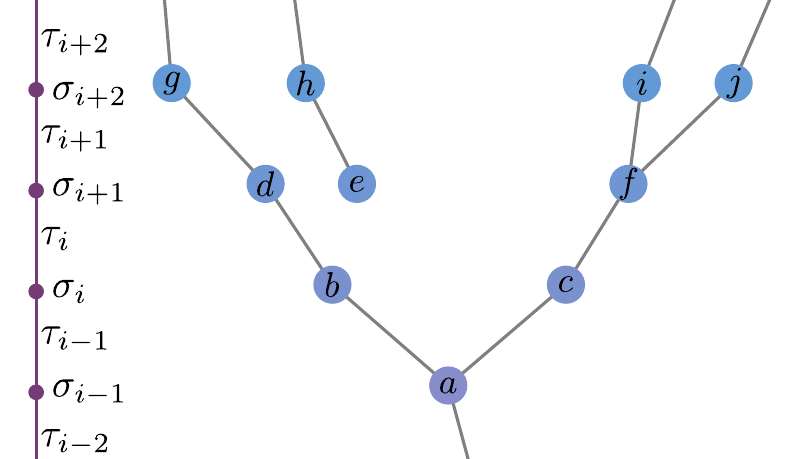}
    \caption{Example of a mapper cosheaf graph representation.
    For example, the vertices from $F(S_{\sigma_i})$ are $b$ and $c$; the edges from $F(S_{\tau_i})$ are $bd$ and $cf$.}
    \label{fig:example_slice}
\end{figure}

\subsection{Representing Mapper Cosheaves as Graphs}
\label{ssec:cosheafgraphs}

Given a mapper cosheaf $F\colon\Open(\cU) \to \Set$, we build a graph $(V_F,E_F)$ as follows. 
Recall that $K$, the one-dimensional cubical complex discretizing $\R$, consists of vertices  $\sigma_{-L},\cdots,\sigma_L$ with heights contained in the bounding box $[-L\delta,L\delta]$, along with edges ${\tau_j = (\sigma_j, \sigma_{j+1})}$. 
The graph for the mapper cosheaf $F:\Open(\cU) \to \Set$ is constructed by generating a vertex for every object in every $F(S_{\sigma_i})$, an edge for every object in every $F(S_{\tau_i})$, and connecting them using the constituent morphisms of $F$.
The resulting vertex and edge sets are
$V_F = \coprod_{i =1}^B F(S_{\sigma_i})$ and $E_F = \coprod_{i =1}^{B-1} F(S_{\tau_i})$, where the difference in index sets is a result of having an edge for every object in every $F(S_{\tau_i})$, of which there is one fewer by construction. 
The endpoints of any edge $e \in F(S_{\tau_i}) \subseteq E_F$ are found via the attaching maps: 
\[F[S_{\tau_i} \subseteq S_{\sigma_{i}}](e) \in F(S_{\sigma_i}) 
    \text{ and } 
F[S_{\tau_i} \subseteq S_{\sigma_{i+1}}](e) \in F(S_{\sigma_{i+1}}). \]
For example,  $ (c,f) \in F(S_{\tau_{i}})$ in \cref{fig:example_slice} has endpoints $c \in F(U_{\sigma_i})$ and $f \in F(U_{\sigma_{i+1}})$. 
This data is stored in a standard adjacency list, where each vertex $v \in F(U_{\sigma_i})$ stores the associated value $i$ as a representation of its height.

We next use $F$ to construct graph representations for both the functors $F^n$ and $F^{2n}$; see \cref{fig:example_all_graphs} for reference.
Since building $F^{2n}$ from $F^n$ is the same as building $F^{n}$ from $F$ up to indexing, we only describe the construction of $F^n$. 
For any $\rho \in K$, $F^n(S_{\rho}) = F(S^n_{\rho})$, and because of the cosheaf properties, this means 
\begin{equation}
\label{eq:colimit}
F^n(S_{\rho}) = \colim _{S \in S_{\rho}^n}F(S) = \colim_{S_{\rho'} \in S_{\rho}^n}F(S_{\rho'}). 
\end{equation}
If $\rho$ is a vertex $\sigma_i$, then checking indices shows that the colimit is taken over the set
\begin{equation}
\{\rho' \in K \mid S_{\rho'} \in S_{\sigma_i}^n \}
=
\{ \sigma_j \mid j \in [i-n,i+n]\} \cup \{ \tau_j \mid j \in [i-n-1,i+n]\}.
\end{equation}
The result is that $F^n(S_{\sigma_i})$ is the set of connected components of the graph induced by the vertex set 
$\coprod_{j \in [i-n,i+n]} F(S_{\sigma_j}) \subseteq V_F.$
Each connected component results in a vertex in the graph representation of $F^n$. 
This induced graph does not include any edges representing elements of $F(S_{\tau_{i-n-1}})$ or $F(S_{\tau_{i+n}})$, although these are both sets involved in the colimit of \cref{eq:colimit}.
However, since these ``half edge'' elements have maps only to $F(S_{\sigma_{i-n}})$ and $F(S_{\sigma_{i+n}})$ respectively, dropping them does not change the resulting colimit.

\begin{figure}
    \centering
    \includegraphics[width=.8\linewidth]{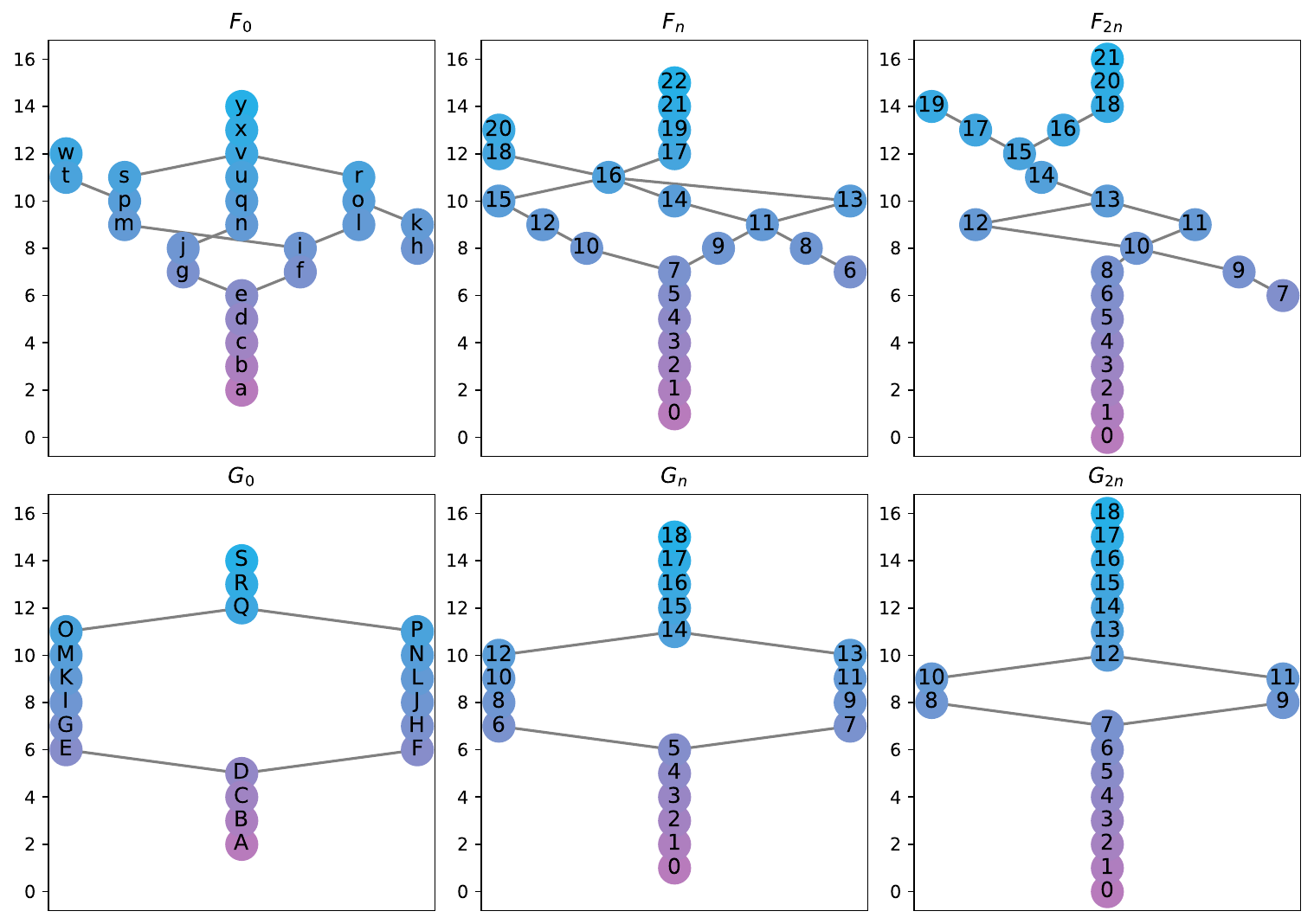}
    \caption{Input mapper cosheaf graphs $F$ and $G$ (left) with $n$- (middle) and $2n$-smoothed versions (right).
    Vertex labels do not compare across graphs; instead, they reflect the construction methods.
    These examples correspond to the matrix examples shown in \cref{fig:boundary_matrix,fig:inclusion_matrix,fig:distance_matrix,fig:assignment_matrix}.}
    \label{fig:example_all_graphs}
\end{figure}

Similarly, for an edge $\tau_i \in K$, the colimit is taken over the set 
\begin{equation*}
\{\rho' \in K \mid S_{\rho'} \in S_{\tau_i}^n \}
=
\{ \sigma_j \mid j \in [i-n+1,i+n]\} \cup \{ \tau_j \mid j \in [i-n,i+n]\}.
\end{equation*}
and thus is the connected components of the graph induced by vertex set
\begin{equation}
\label{eq:induced_graph_for_edge_thickening}
\coprod_{j \in [i-n+1,i+n]} F(S_{\sigma_j}) \subseteq V_F.
\end{equation}
Each connected component of this graph gives an edge in the graph representation of $F^n$, 
with endpoints determined by the induced map on colimits since $S_{\tau_i}^n \subseteq S_{\sigma_i}^n$ and ${S_{\tau_i}^n \subseteq S_{\sigma_{i+1}}^n}$. 

Since connected components of a graph $(V,E)$ can be determined in $O(|V|+|E|)$ time using breadth first search, $F^n$ can be built from $F$ in $O(L \cdot (|V_F| + |E_F|))$ time.
Further, during this construction the natural transformation $F \rightarrow F^n$ can be stored at the same time; see \cref{ssec:InclusionMatrices} for more details. 
Similarly, $F^{2n}$ can be built from $F$, so getting all four additional graphs from $F$ and $G$ can be done in $O(L \cdot (|V_F| + |E_F|))$ time.

Given a graph representation of one of the functors $F$, $F^n$, or $F^{2n}$, we must determine the distance $d_{S_{\sigma_i}}(v,v')$ between any two vertices, as defined in \cref{eqn:DistanceDefinition}.
Let the mapper cosheaf $H \in \{F, F^n, F^{2n}\}$ be graphically represented by a vertex set $V$ and an edge set $E$, where each vertex $v \in V$ corresponds to a $\sigma_i$ and each edge $e \in E$ corresponds to a $\tau_i$.
For $v, v' \in H(S_{\sigma_i})$, the condition $H[S_{\sigma_i} \subseteq S_{\sigma_i}^k](v) = H[S_{\sigma_i} \subseteq S_{\sigma_i}^k](v')$ holds precisely when $v$ and $v'$ lie in the same connected component of the subgraph induced by vertices in $S_{\sigma_i}^k$.
This reduces to finding a minimum-height path between the two vertices, which can be computed by a breadth-first search (BFS) in $O(|V_F| + |E_F|)$ time.

\subsection{Representing Maps and Relationships as Block Matrices}
\label{sec:AllMatrices}

Our next task is to represent the information carried by the relevant (un)natural transformations. 
Recall that an (un)natural transformation $\eta\colon H \to H'$ consists of set maps $\eta_{S}\colon H(S) \to H'(S)$ for all $S \in \{S_\rho \mid \rho \in K\} \cup \{S^n_\rho \mid \rho \in K\}$. Let $(V_H, E_H)$ and $(V_{H'}, E_{H'})$ denote the vertex and edge sets of the graph representations of $H$ and $H'$, respectively. 
Each set map sends vertices (resp.~edges) representing elements of $H(S_{\sigma_i})$ (resp. $H(S_{\tau_i})$) to vertices (resp. edges) representing elements of $H'(S_{\sigma_i})$ (resp. $H'(S_{\tau_i})$). To encode this consistently, we impose a total order on vertices (and similarly on edges) so that if $v$ corresponds to $\sigma_i$ and $v'$ to $\sigma_j$ with $i < j$, then $v < v'$.

We store $\eta$ as a pair of block matrices $M_V \in \{0,1\}^{|V_{H'}| \times |V_H|}$ and $M_E \in \{0,1\}^{|E_{H'}| \times |E_H|}$, where an entry at $(a,a')$ is 1 iff $\eta_{S_\rho}(a) = a'$ for the appropriate $\rho \in K$. All nonzero entries therefore lie within blocks indexed by $H'(S_\rho) \times H(S_\rho)$, and each column contains exactly one $1$, with all remaining entries equal to $0$.

In this section, we give specifics on the matrices representing the various maps needed in the loss function computation. 

\subsubsection{Boundary Matrices}
\label{ssec:BdyMatrices}
    \begin{figure}[]
        \centering
        \includegraphics[width = .8\textwidth, align = c]{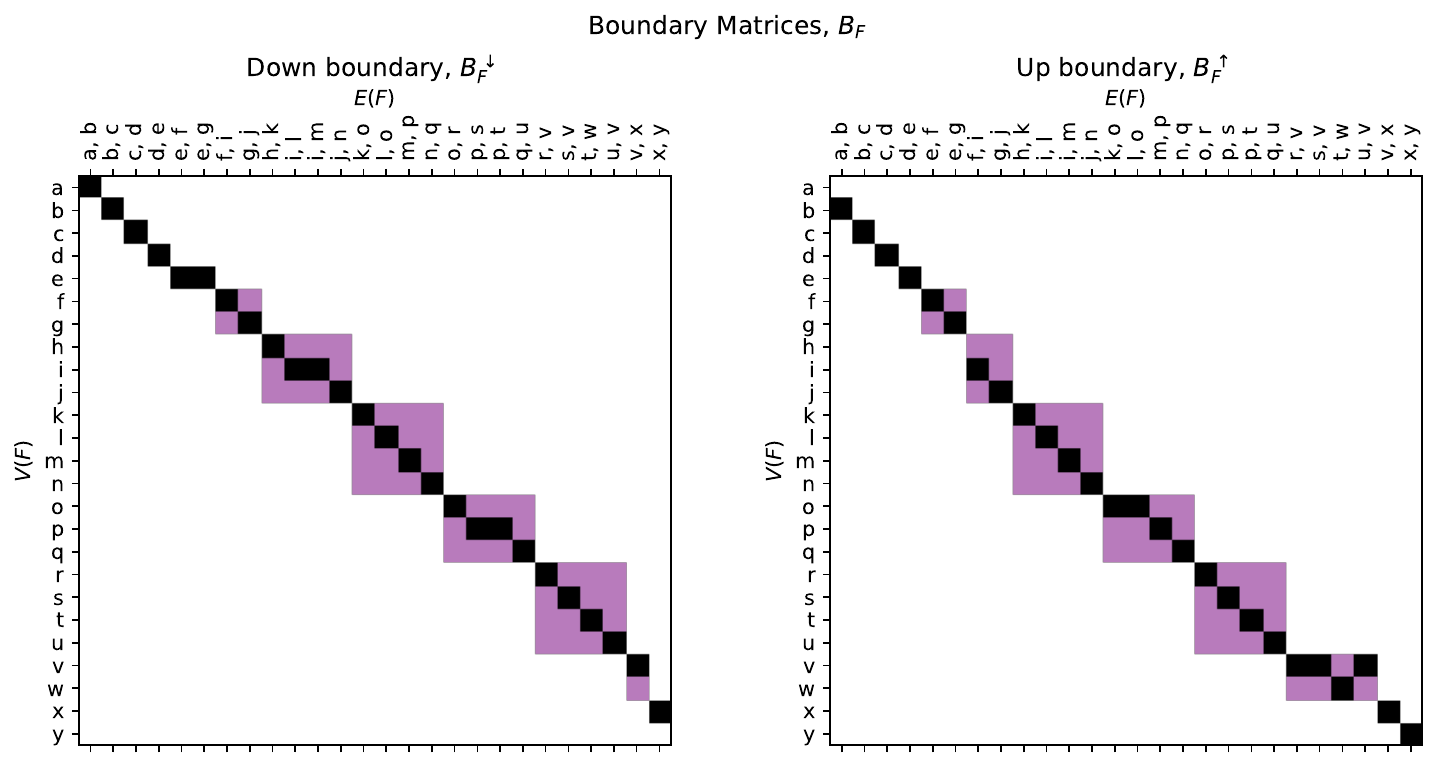}
        \includegraphics[width = .07\textwidth, align = c]{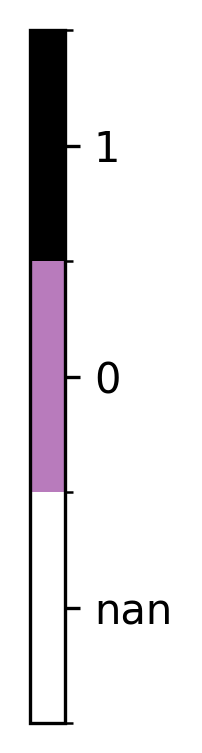}
        \caption{Example boundary matrices of the graph of $F$ shown in \cref{fig:example_all_graphs}.}
        \label{fig:boundary_matrix}
    \end{figure}
The boundary matrices describe the relationship between the edges and vertices of a graph representation of a given mapper cosheaf. 
For a mapper cosheaf $H \in \{F, F^n,  G, G^n \}$, this is a matrix $B_H \in \{0,1\}^{|E_H| \times |V_H|}$, with a 1 in entry $B_H[e,v]$ iff $v$ is an endpoint of $e$. 
Unlike the other matrices we have used, this matrix has two nonzero entries in each column. However, it can be decomposed into a sum of two block matrices, each with a single nonzero entry per column, as follows. 
Let $B_H^{\uparrow} \in \{0,1\}^{|E_H| \times |V_H|}$ have a 1 in entry $B[e,v]$ iff $v$ is an endpoint of $e$ for $v \in F(\sigma_i)$ and $e \in F(\tau_{i})$; let $B_H^{\downarrow} \in \{0,1\}^{|E_H| \times |V_H|}$ have a 1 in entry $B[e,v]$ iff $v$ is an endpoint of $e$ for $v \in F(\sigma_i)$ and $e \in F(\tau_{i-1})$. 
Then $B_H = B_H^{\uparrow} + B_H^{\downarrow}$; see \cref{fig:boundary_matrix} for an example. 
In this case, we will have eight matrices to work with, which we denote by 
$B_F^{\uparrow}$, $B_F^{\downarrow}$, 
$B_{F^n}^{\uparrow}$,  $B_{F^n}^{\downarrow}$,  
$B_G^{\uparrow}$, $B_G^{\downarrow}$, 
$B_{G^n}^{\uparrow}$, and $B_{G^n}^{\downarrow}$.
We do not compute the boundary matrices for either $F^{2n}$ or ${G}^{2n}$ since, as can be seen in \cref{tab:theList}, these matrices are never used.
Because these are simply matrix encodings of the adjacency information of the relevant graphs, they can be constructed in $O(V_{\max} \cdot E_{\max})$ time, where $V_{\max} = \max\{|V_H| \mid H \in F, F^n,  G, G^n  \}$ and $E_{\max} = \max\{|E_H| \mid H \in F, F^n,  G, G^n  \}$.

\subsubsection{Inclusion Matrices} 
\label{ssec:InclusionMatrices}
    \begin{figure}[!htb]
        \centering
        \includegraphics[width = .8\linewidth, align = c]{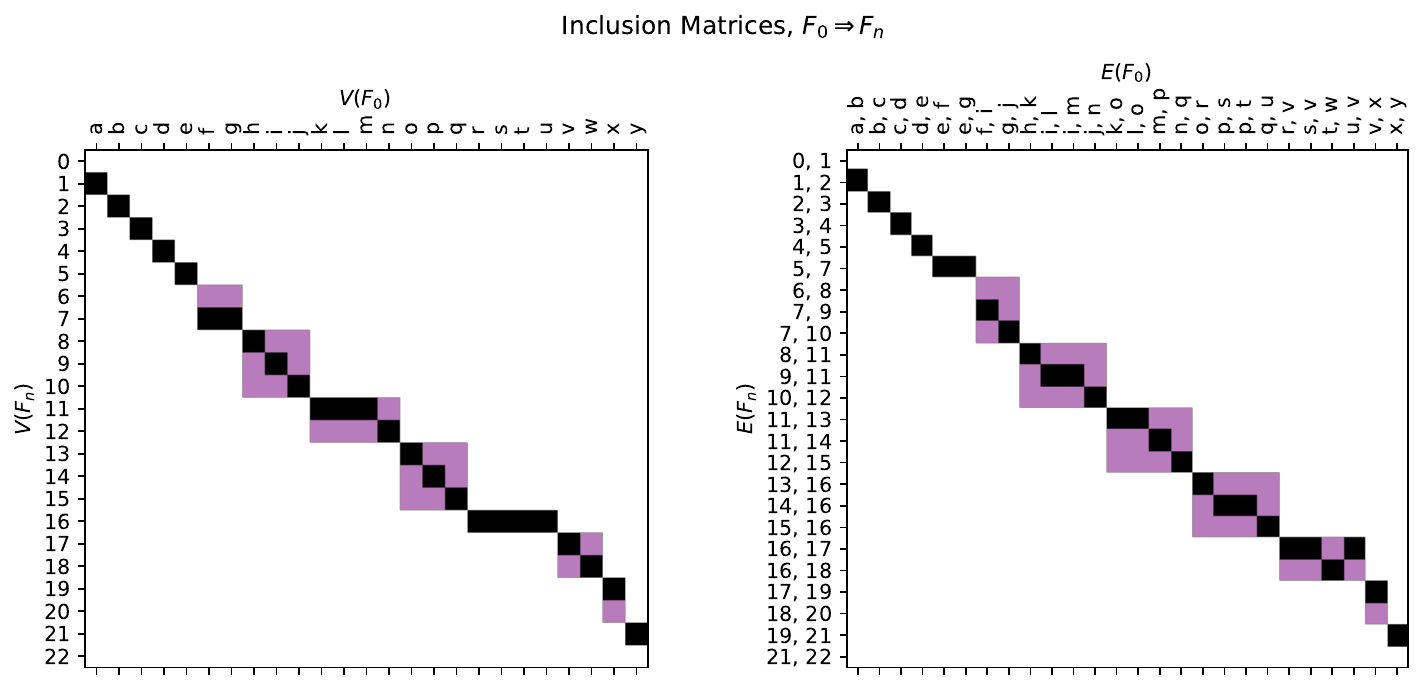}
        \includegraphics[width = .07\textwidth, align = c]{figures/colorbar-PurBl-binary.png}
        \caption{Inclusion matrices $I_F^V$ (left) and $I_F^E$ (right) representing $F \Rightarrow F_n$ from \cref{fig:example_all_graphs}.}
        \label{fig:inclusion_matrix}
    \end{figure}

The inclusion matrices represent the natural transformations $F \Rightarrow F_n \Rightarrow F_{2n}$ and $G  \Rightarrow G_n  \Rightarrow G_{2n}$.
Focusing on the first of these, $F \Rightarrow F^n$, the component morphisms take a vertex $v \in F(S_{\sigma_i})$ to the vertex $v' \in F^n(S_{\sigma_i})$ representing the connected component containing $v'$. 
The edge maps are determined similarly. 
    Thus the vertex inclusion matrix is written as $I_F^V \in \{0, 1 \}^{|V_{F^n}| \times |V_{F}|}$ with a 1 in entry $I_F[v',v]$ if $v$ is in the connected component represented by $v'$.
    A similar matrix is built for the edges, $I_F^E \in \{0, 1 \}^{|E_{F^n}| \times |E_{F}|}$. 
    The result is that each inclusion natural transformation is given by a pair of matrices, so from this construction we have the matrices $I_F^V$, $I_F^E$, $I_{F^n}^V$, $I_{F^n}^E$, $I_G^V$, $I_G^E$, $I_{G^n}^V$, $I_{G^n}^E$; see \cref{fig:inclusion_matrix} for an example. 
These matrices satisfy the block decomposition and have exactly one nonzero entry in each column. They also have the useful property that the inclusion functor from a graph to its $2n$-smoothing can be computed via matrix multiplication; for example, the induced map on vertices for $F \Rightarrow F_{2n}$ is given by $I_{F^n}^V I_F^V$. Moreover, since they track representatives of objects in the $n$-smoothed connected components, these matrices can be constructed alongside the smoothed graphs with only constant overhead.

\subsubsection{Distance Matrices}
\label{ssec:DistanceMatrices}
     \begin{figure}[!htb]
       \centering
        \includegraphics[width = .4\linewidth, align = c]{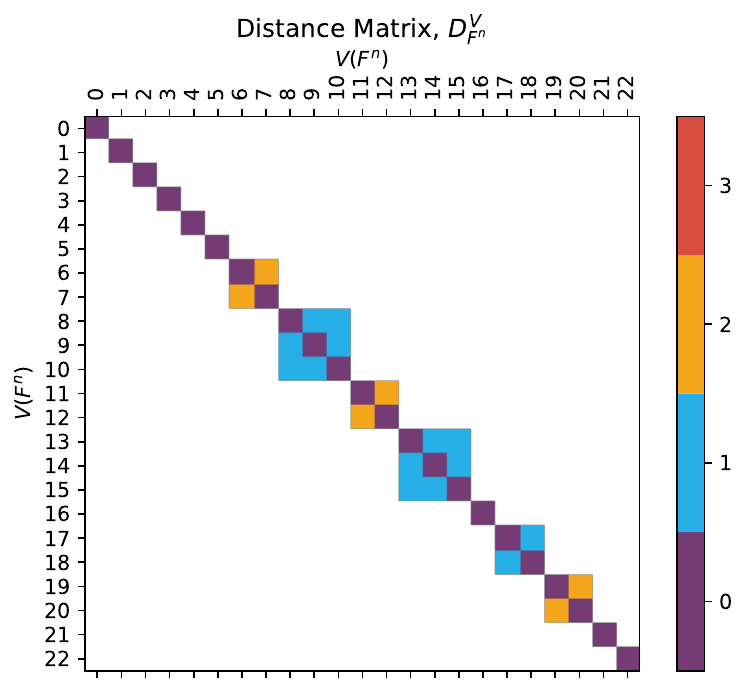}
        \includegraphics[width = .4\linewidth, align = c]{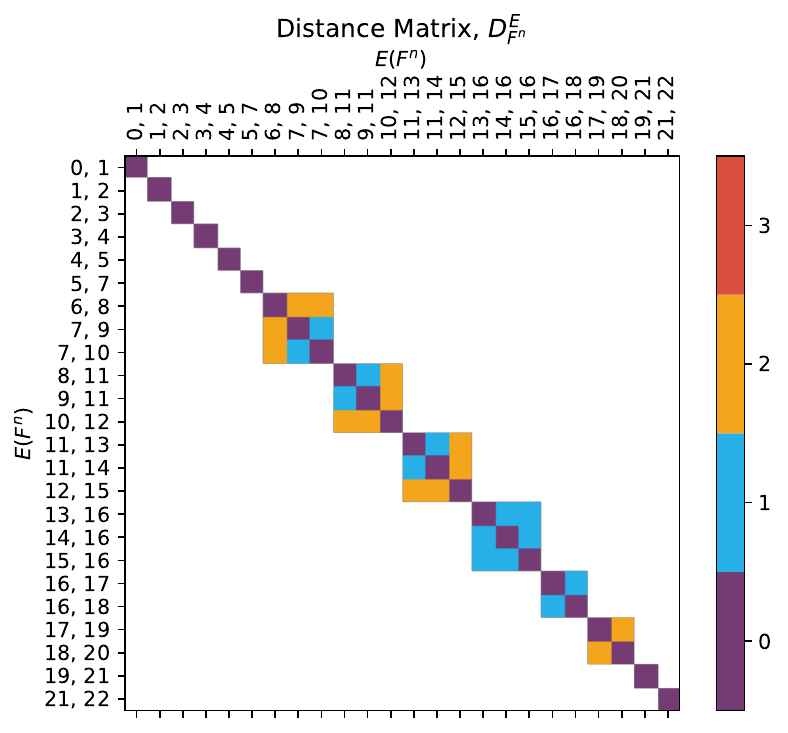}
        \caption{Distance matrix example for $F^n$ from \cref{fig:example_all_graphs}.}
        \label{fig:distance_matrix}
    \end{figure}
The distance matrices encode the distance from \cref{eqn:DistanceDefinition}. 
In particular, for vertices $v,w \in F(S_\rho)$, we seek the minimum $k$ such that they map to the same object in $F(S_\rho^k)$. 
Equivalently, this is the minimum $k$ where $v$ and $w$ are in the same connected component of the subgraph induced by the vertex set $\coprod_{S_{\sigma_i} \subseteq S_\rho^k} F(S_{\sigma_i})$.
For example, in the graph of $F^n$ from \cref{fig:example_all_graphs}, vertices $11$ and $12$, both at level $\sigma_9$, are not in the same connected component until $k=2$ using the path through vertex $7$ below. 
Similarly, edges $(12,15)$ and $(11,14)$ have distance 2 via the path through vertex $7$.

This distance is stored as a pair of matrices $D_F^V \in \{0,1\}^{|V(F)|^2}$ and $D_F^E\in \{0,1\}^{|E(F)|^2}$. 
For the vertices, $D_F^V[u,v] = d_{S_{\sigma_i}}(u,v)$ if $u,v \in F(S_{\sigma_i}) \subset V_F$, and 0 otherwise; $D_F^E$ is defined similarly for the elements of $F(S_{\tau_i})$.
These matrices inherit the block structure determined by the mapper cosheaf heights.
See \cref{fig:distance_matrix} for examples corresponding to $F^n$ from \cref{fig:example_all_graphs}. 
In total, we have 12 distance matrices, two each for the six input matrices. 
However, as noted in \cref{tab:theList}, only the eight for $F^n$, $F^{2n}$, $G^n$, and $G^{2n}$ need to be calculated. 

We now compute these distance matrices from the graph representations.
For each level $i \in [-L,L]$, we construct the block for $\sigma_i$ and $\tau_i$ in $D_F^V$  and $D_F^E$ respectively.
In the vertex case, we build a union-find data structure by adding the edges in order of increasing thickening. 
Recall (and see \cref{fig:PosetNotation} for notation reminder) that the elements of $F(S^k_{\sigma_i})$ correspond to the subgraph induced by the vertices in $\coprod _{j \in [i-n,i+n]} F(S_{\sigma_j})$, which includes the edges in $\coprod_{j \in [i-k, i+k-1]}F(S_{\tau_j})$. 
Therefore, for $k=1$, we add all edges corresponding to $\tau_{i-1}$ and $\tau_i$. Note that the edges for $\tau_{i-2}$ and $\tau_{i+1}$ are irrelevant here, as they do not involve any vertices in $S_{\sigma_j}^1$ and thus cannot affect connectivity. We then check whether any vertices in $F(S_{\sigma_i}^1)$ lie in the same connected component. For each pair $u$ and $v$ that do, we set $D_F^V[u,v] = 1$.
Then for $k=2$, we add the edges for $\tau_{i-2}$ and $\tau_{i+2}$, and repeat the same check for pairs that are not already in the same connected component at $k=1$, setting these distances to 2 if they are now in the same component. 
We continue until all pair distances are determined. 

A similar sweep for each level is used to construct $D_F^E$, with the difference being that the edges added at step $k$ are those from $F(\tau_{i-k+1}, \tau_{i+k-1})$. 
We have $V_{\max}$ entries in our union-find data structure, so amortized time per operation is $O(\alpha(V_{\max}))$, with a total of $E_{\max}$ operations.
Thus, since the union-find structure is built at each level, the pair of matrices for a given graph takes time $O(L\cdot E_{\max} \cdot \alpha(V_{\max}))$, where $\alpha(\cdot)$ is the inverse Ackermann function.

\subsubsection{Assignment Matrices}
\begin{figure}[!ht]
        \centering
        \includegraphics[width=.8\linewidth, align = c]{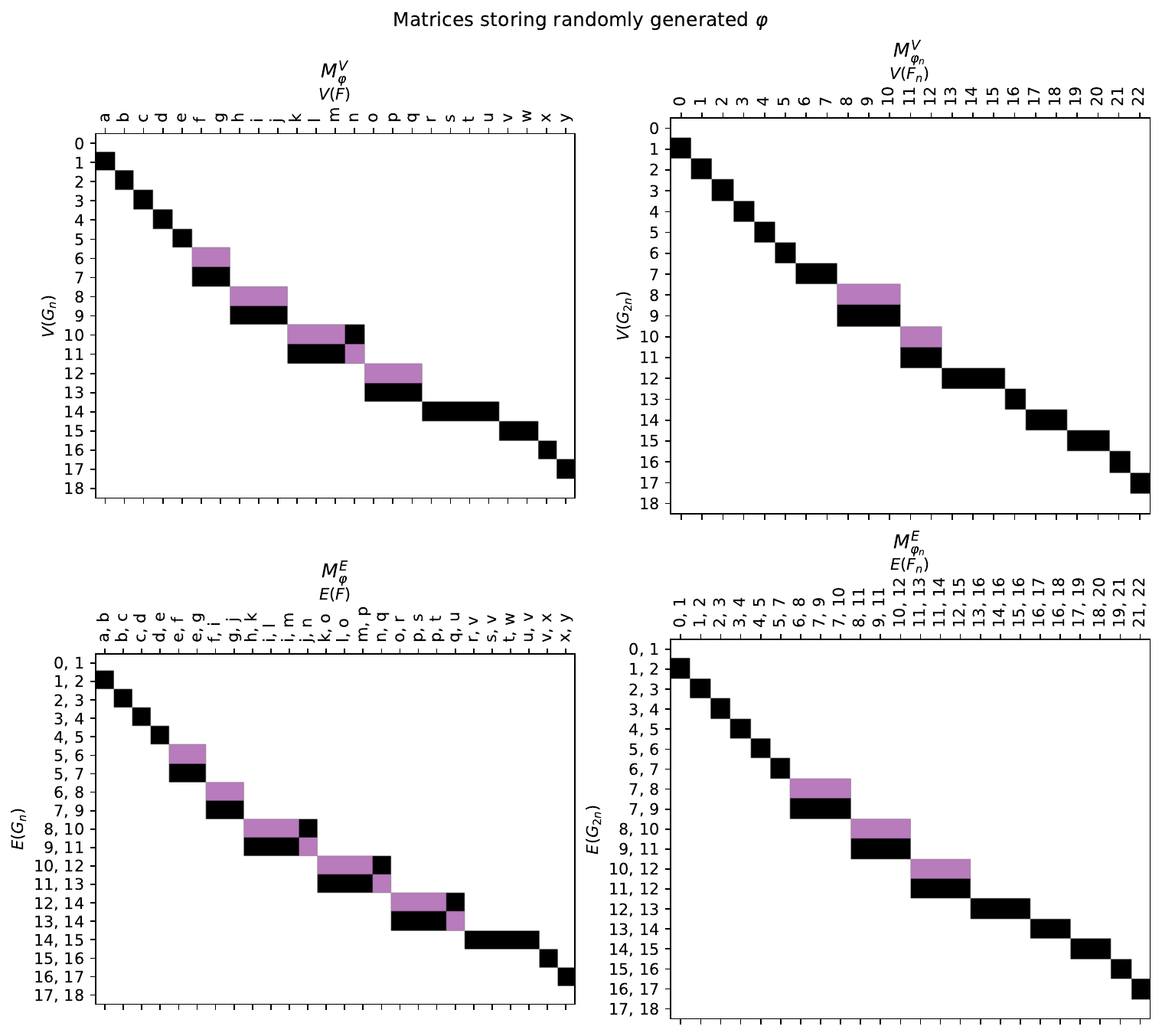}
        \includegraphics[width = .07\textwidth, align = c]{colorbar-PurBl-binary.png}
        \vspace{-4mm}
        \caption{Matrices for an input unnatural transformation $\phi$ of the example graphs in \cref{fig:example_all_graphs}. }
        \label{fig:assignment_matrix}
\end{figure}

To solve for assignment choices, we record them as matrices. Given an $n$-assignment $\phi$ and $\psi$, we first restrict each map to $S_\rho$ and $S_\rho^n$, and then separate these restrictions according to whether $\rho$ is a vertex ($\rho = \sigma_i$) or an edge ($\rho = \tau_i$). The four resulting matrices for $\phi$ (with analogous matrices for $\psi$) are denoted as 
\begin{equation*}
    \begin{matrix}
    M_\phi^V \in \{0,1\}^{|V(G^n)| \times |V(F)|}, & 
    M_{\phi_n}^V \in \{0,1\}^{|V(G^{2n})| \times |V(F^n)|}, \\ 
    M_\phi^E \in \{0,1\}^{|E(G^n)| \times |E(F)|}, &
    M_{\phi_n}^E \in \{0,1\}^{|E(G^{2n})| \times |E(F^n)|}. 
    \end{matrix}
    \end{equation*}
As with the other matrices, these are block matrices aligned with the $\sigma_i$ and $\tau_i$ indices, each containing exactly one 1 per column. The matrix $M_\phi^V$ satisfies $M_\phi^V[u,v] = 1$ iff $\phi_{S_{\sigma_i}}(u) = v$ for $u \in F(S_{\sigma_i})$ and $v \in G(S_{\sigma_i}^n)$. An example is shown in \cref{fig:assignment_matrix}, where $\phi$ is randomly generated so that each column contains a single 1; in this case, the associated parallelogram diagram fails to commute, indicating that $\phi$ is not a natural transformation. In our ILP formulation, we search over such assignment matrices to identify choices of $\phi$ and $\psi$ that minimize the loss. We need only initialize empty matrices in the setup.

\subsection[]{Binary Decision and Loss Computation for Fixed $\phi$ and $\psi$}

\begin{table}[]
    \centering
\begin{tabular}{|r|p{1cm}|c|c|c|}
\hline 
&   \centering{\textbf{Loss Term}}   & \textbf{Diagram} & \textbf{Matrix Multiplication} & \textbf{Eval.}\\ \hline \hline
\multirow{2}{*}{\rotatebox[origin=r]{90}{$\substack{\text{Edge-Vertex \hfill}\\\text{Parallelogram \hfill}}$}}
& $\Lpl^{S_\tau, S_\sigma}$
& 
    $
    \begin{tikzcd}[scale cd = .7]
    F(S_\tau)  
        \ar[r, "{F[\subseteq ]}"] 
        \ar[dr, "\varphi_{S_\tau}"']
    & F(S_\sigma)
        \ar[dr, "\varphi_{ S_\sigma}"]
    & \\
    & G^n(S_\tau) 
        \ar[r, "{G[\subseteq ]}"'] 
    & G^n(S_\sigma)  
    \end{tikzcd}
    $
&
    \parbox{4.5cm}{
\centering
        $D_{G^n}^V \left(M_{\phi}^V \cdot B_F^{\uparrow} - B_{G^n}^{\uparrow} \cdot M_{\phi}^E\right)$\\
        $D_{G^n}^V \left(M_{\phi}^V \cdot B_F^{\downarrow} - B_{G^n}^{\downarrow} \cdot M_{\phi}^E\right)$
   }
& \multirow{4}{*}[-11.5ex]{$\max(A)$}
\\
\cline{2-4}
& $\Lpr^{S_\tau, S_\sigma}$
&
    $
    \begin{tikzcd}[scale cd = .7]
    & F^n(S_\tau) 
        \ar[r, "{F[\subseteq ]}"] 
    & F^n(S_\sigma)  
    \\
    G(S_\tau)  
        \ar[r, "{G[\subseteq ]}"'] 
        \ar[ur, "\psi_{S_\tau}"]
    & G(S_\sigma)
        \ar[ur, "\psi_{ S_\sigma}"']
    \end{tikzcd}
    $
&
\parbox{4.5cm}{
\centering
    $D_{F^n}^V \left(M_{\psi}^V \cdot B_G^{\uparrow} - B_{F^n}^{\uparrow} \cdot M_{\psi}^E\right)$\\
    $D_{F^n}^V \left(M_{\psi}^V \cdot B_G^{\downarrow} - B_{F^n}^{\downarrow} \cdot M_{\psi}^E\right)$
}
& 
\\
\cline{2-4}
\multirow{2}{*}{\rotatebox[origin=r]{90}{$\substack{\text{Thickening \hfill}\\\text{Parallelogram \hfill}}$}}
& $\Lpl^{S_\rho, S_\rho^n}$
& 
    $
    \begin{tikzcd}[scale cd = .7]
    F(S_\rho)  
        \ar[r, "{F[\subseteq ]}"] 
        \ar[dr, "\varphi_{S_\rho}"']
    & F^n(S_\rho)
        \ar[dr, "\varphi_{ S_\rho^n}"]
    & \\
    & G^n(S_\rho) 
        \ar[r, "{G[\subseteq ]}"'] 
    & G^{2n}(S_\rho)  
    \end{tikzcd}
    $
&
\parbox{4.5cm}{
\centering
    $D_{G^{2n}}^V \left(M_{\phi^n}^V \cdot I_F^V - I_{G^n}^V \cdot M_{\phi}^V\right)$\\
    $D_{G^{2n}}^E \left(M_{\phi^n}^E \cdot I_F^E - I_{G^n}^E \cdot M_{\phi}^E\right)$
}
&
\\
\cline{2-4}
& $\Lpr^{S_\rho, S_\rho^n}$
&
    $
    \begin{tikzcd}[scale cd = .7]
    & F^n(S_\rho) 
        \ar[r, "{F[\subseteq ]}"] 
    & F^{2n}(S_\rho)  
    \\
    G(S_\rho)  
        \ar[r, "{G[\subseteq ]}"'] 
        \ar[ur, "\psi_{S_\rho}"]
    & G(S_\rho^n)
        \ar[ur, "\psi_{ S_\rho^n}"']
    \end{tikzcd}
    $
&
\parbox{4.5cm}{
\centering
    $D_{F^{2n}}^V \left(M_{\psi^n}^V \cdot I_G^V - I_{F^n}^V \cdot M_{\psi}^V\right)$\\
    $D_{F^{2n}}^E \left(M_{\psi^n}^E \cdot I_G^E - I_{F^n}^E \cdot M_{\psi}^E\right)$
}
& \parbox{3.5cm}{
\centering}
\\
\hline
\multirow{2}{*}[-0.4ex]{\rotatebox[origin=r]{90}{Triangle \hfill}}
& $\Ltd^{S_\rho}$
&
    $
    \begin{tikzcd}[scale cd = .7]
        F(S_\rho) 
            \ar[rr, "{F[ \subseteq ]}"]   
            \ar[dr, "\phi_{S_\rho}"'] 
            & & F^{2n}(S_\rho)
        \\
        & G^n(S_\rho)\ar[ur, "\psi_{S^n_\rho}"']  & 
    \end{tikzcd}
    $
&
    \parbox{4.5cm}{
\centering
    $D_{F^{2n}}^V \left(I_{F^n}^V \cdot I_F^V - M_{\psi^n}^V\cdot M_{\phi}^V \right)$\\
    $D_{F^{2n}}^E \left(I_{F^n}^E \cdot I_F^E - M_{\psi^n}^E\cdot M_{\phi}^E \right)$
    }
    & \multirow{2}{*}[-3.5ex]{\vspace{\fill}$\left\lceil \tfrac{1}{2}\max(A)\right\rceil$}
\\
\cline{2-4}
& $\Ltu^{S_\rho}$
&
    $
    \begin{tikzcd}[scale cd = .7]
        & F^n(S_\rho)\ar[dr, "\phi_{S^n_\rho}"]  & 
        \\
        G(S_\rho) 
            \ar[rr, "{G[ \subseteq ]}"']   
            \ar[ur, "\psi_{S_\rho}"] 
            & & G^{2n}(S_\rho)
    \end{tikzcd}
    $
&
   \parbox{4.5cm}{
\centering
    $D_{G^{2n}}^V \left(I_{G^n}^V \cdot I_G^V - M_{\phi^n}^V\cdot M_{\psi}^V \right)$\\
    $D_{G^{2n}}^E \left(I_{G^n}^E \cdot I_G^E - M_{\phi^n}^E\cdot M_{\psi}^E \right)$
    }
& 
\\
\hline
\end{tabular}
\caption{The list of loss terms, relevant diagram, and matrix multiplications. The final column shows the result from the final matrix, where the multiplied matrix is denoted $A$: either the maximum entry in the multiplied matrix, or the ceiling of half the maximum.} 
\label{tab:theList}
\end{table}

In this section, we use the matrices introduced above to begin addressing \cref{prob:BinaryDecision} and \cref{prob:LossComputation}. 
Throughout, we use the notation $\max(A)$ to denote the maximum entry in the matrix $A$.
For the binary decision problem, we focus on the version of the problem for a fixed assignment as follows.
We give the answers for these two problems by showing the matrix multiplications needed to check the commutativity of all diagrams (for \cref{prob:BinaryDecisionFixedAssignment}); or for computing each term in the loss function (for \cref{prob:LossComputationFixedAssignment}). 
In order to determine if $\phi,\psi$ is an interleaving and thus answer \cref{prob:BinaryDecisionFixedAssignment}, we need to check commutativity of all the diagrams in the second column of \cref{tab:theList}. 
Due to symmetry, we limit our discussion to checking commutativity of the diagrams  $\Parallelograml_\phi(S_\tau,S_\sigma)$, $\Parallelograml_\phi(S_\rho,S_\rho^n)$, and $\triangled(S_\rho)$. 
For \cref{prob:LossComputationFixedAssignment}, we need to show how  each term in   
\begin{equation*}
L_B(\phi,\psi) = \max_{ \substack{\sigma < \tau \in K \\ \rho \in K}}
\left\{
\Lpl^{S_\tau, S_\sigma}, \Lpr^{S_\tau, S_\sigma},
\Lpl^{S_\rho, S_\rho^n}, \Lpr^{S_\rho, S_\rho^n},
\Ltu^{S_\rho}, \Ltd^{S_\rho}
\right\}
\end{equation*}
corresponds to matrix multiplication in this setting.
Again, due to symmetry, we limit our discussion to $\Lpl^{S_\tau, S_\sigma}, \Lpl^{S_\rho, S_\rho^n}$, and $\Ltd^{S_\rho}$; the full list can be found in \cref{tab:theList}.
For a matrix $A$, we write $\max(A) \coloneqq \max_{i,j} A_{ij}$  to denote the maximum entry of $A$.

\subsubsection{Parallelograms}
 First, consider the diagram $\Parallelograml_\phi(S_\tau,S_\sigma)$ given by 
  \begin{equation*}
    \begin{tikzcd}
    F(S_{\tau})  
        \ar[r, "{F[\subseteq ]}"] 
        \ar[dr, "\varphi_{S_\tau}"']
    & F(S_\sigma)
        \ar[dr, "\varphi_{ S_\sigma}"]
    & \\
    & G^n(S_\tau) 
        \ar[r, "{G[\subseteq ]}"'] 
    & G^n(S_\sigma)  
    \end{tikzcd}
 \end{equation*}
 In our previous notation, we have that $\sigma_i$ is the lower vertex of $\tau_i$, and $\sigma_{i+1}$ is the upper vertex. 
 
\begin{figure}
    \centering
    \includegraphics[width=\linewidth]{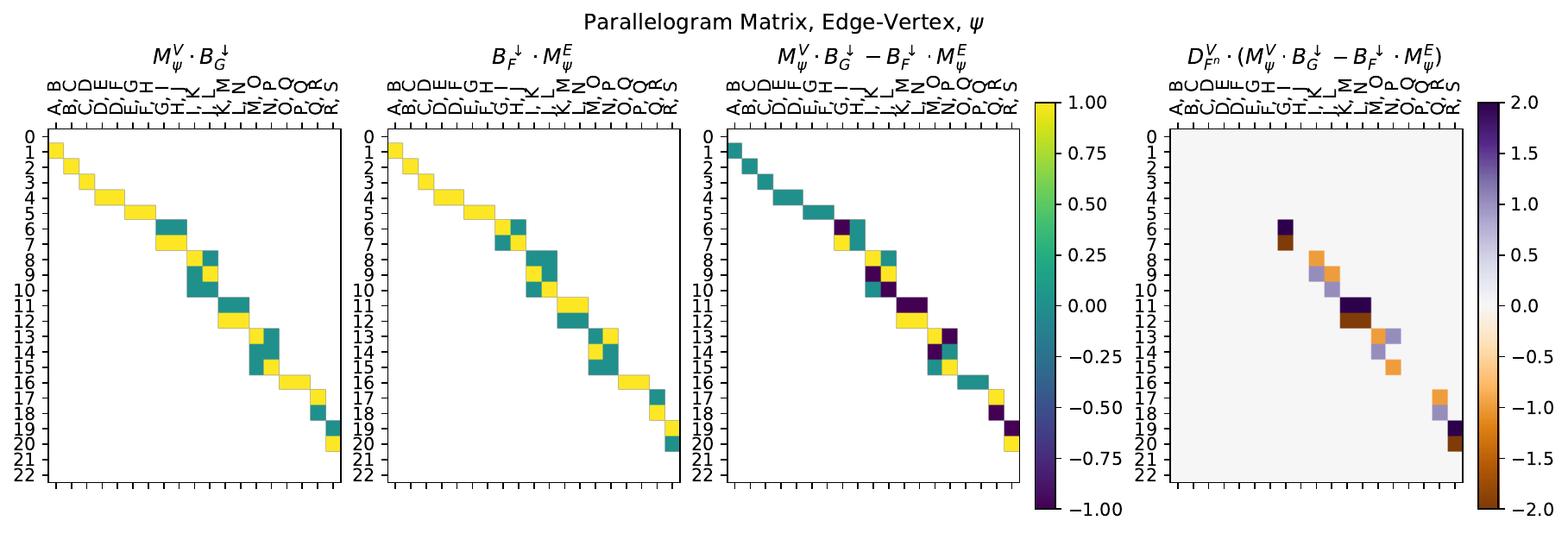}
    \caption{Example of the matrix multiplications used to to compute the loss function for the $\Lpl(S_\tau,S_\sigma)$ case (\cref{lem:matrix_entries_edge_vert_down_binary,lem:matrix_entries_edge_vert_down}). 
    White entries in the first three matrices are outside of the block structure, and can be assumed to be 0 although they are stored as \texttt{NaN}. 
    The loss in this case is the maximum absolute entry in the last matrix.
    } 
    \label{fig:multiplicationExamples_Parallelogram}
\end{figure}

\begin{lemma}
\label{lem:matrix_entries_edge_vert_down_binary}
The diagrams $\{\Parallelograml_\phi(S_{\tau_i},S_{\sigma_i})\}_i$ all commute if and only if 
\begin{equation*}
M_{\phi}^V \cdot B_F^{\downarrow} - B_{G^n}^{\downarrow} \cdot M_{\phi}^E = 0.
\end{equation*}
The diagrams $\{\Parallelograml_\phi(S_{\tau_{i+1}},S_{\sigma_i})\}_i$ all commute if and only if 
\begin{equation*}
M_{\phi}^V \cdot B_F^{\uparrow} - B_{G^n}^{\uparrow} \cdot M_{\phi}^E = 0.
\end{equation*}
\end{lemma} 

See \cref{fig:multiplicationExamples_Parallelogram} for examples of the matrices from the statement.

\begin{proof}
We will focus on proving the first statement since the methodology is the same for the second. 
Fix an $i$.  
Throughout the proof we write $\sigma = \sigma_i$ and $\tau = \tau_i$. 
We are checking whether the images landing in the bottom right of the diagram
 \begin{equation*}
    \begin{tikzcd}
    F(S_{\tau})  
        \ar[r, "{F[\subseteq ]}"] 
        \ar[dr, "\varphi_{S_\tau}"']
    & F(S_\sigma)
        \ar[dr, "\varphi_{ S_\sigma}"]
    & \\
    & G^n(S_\tau) 
        \ar[r, "{G[\subseteq ]}"'] 
    & G^n(S_\sigma)  
    \end{tikzcd}
 \end{equation*}
 are the same. 
 Fix any element $e \in F(S_{\tau})$. 
 If it has lower vertex $v \in F(S_\sigma)$ with $\phi_{S_\sigma}(v) = v'$, then the composition along the top of the diagram sends $e \mapsto v'$. 
 This means that the column for $e$ in $M_{\phi}^V \cdot B_F^{\downarrow}$ has exactly one 1 in the row corresponding to $v'$. 
If $\phi_{S_\tau}(e) = e' \in G^n(S_\tau)$ which has lower vertex $v''$, then the composition along the bottom of the diagram sends $e \mapsto v''$. 
Similarly, the column for $e$ in $B_{G^n}^{\downarrow} \cdot M_{\phi}^E$ has exactly one 1 in the row corresponding to $v''$. 
Writing $A\coloneqq M_{\phi}^V \cdot B_F^{\downarrow} - B_{G^n}^{\downarrow} \cdot M_{\phi}^E$, the column for $e$ in $A$ is entirely 0 if and only if $v' = v''$. 
This means that the matrix $A$ is entirely 0 iff $G[\subseteq]\circ \phi_{S_\tau}(e) =  \phi_{S_\sigma}\circ F[\subseteq](e)$ for all $e \in F(S_\tau)$ which is the definition of the diagram commuting. 
\end{proof}

We next use the above construction to determine  $\Lpl^{S_\tau, S_\sigma}$.

 \begin{lemma}
\label{lem:matrix_entries_edge_vert_down}
    $\max_{i} \Lpl^{S_{\tau_i}, S_{\sigma_i}} = \max \left( D_{G^n}^V \left(M_{\phi}^V \cdot B_F^{\downarrow} - B_{G^n}^{\downarrow} \cdot M_{\phi}^E\right) \right)$
\end{lemma}

\begin{proof} 
Fix an $i$ and again write $\sigma = \sigma_i$ and $\tau = \tau_i$. 
Following the notation of the proof of \cref{lem:matrix_entries_edge_vert_down_binary}, given an $e \in F(S_{\tau})$ we write $v' = \phi_{S_\sigma}\circ F[\subseteq](e)$ and $v'' = G[\subseteq]\circ \phi_{S_\tau}(e) $. 
From the discussion above, this means that the column for $e$ in $M_{\phi}^V \cdot B_F^{\downarrow}$ has exactly one 1 in the row corresponding to $v'$ and the column for $e$ in $B_{G^n}^{\downarrow} \cdot M_{\phi}^E$ has exactly one 1 in the row corresponding to $v''$. 
This means that, writing $A\coloneqq M_{\phi}^V \cdot B_F^{\downarrow} - B_{G^n}^{\downarrow} \cdot M_{\phi}^E$, the column for $e$ in $A$ is either entirely 0 if $v' = v''$ or its only non-zero entries are $A[v',e] = 1$ and $A[v'',e] = -1$. 

We can  left-multiply the matrix $A$ by the distance matrix $D\coloneqq D_{G^{n}}^V$.
Clearly, if $v' = v''$ and thus the column in $A$ for $e$ is 0, the column for $e$ in $D \cdot A$ is 0. 
If $v' \neq v''$, the entries for the column $e$ are 
\begin{equation*}
    (D \cdot A)[u,e] = 
    \begin{cases}
        -d(v',v'') & \text{if } u = v'\\
        d(v',v'') & \text{if } u = v'' \\
        d(u,v') - d(u,v'') & \text{else}
    \end{cases}
\end{equation*}
where $d$ denotes the distance function $d^G_{S^n_\sigma}$. 
However, since $d(u,v') - d(u,v'') \leq d(v',v'')$ by the triangle inequality, the maximum entry in the column occurs at entry $[v'', e]$. 
This means that $\Lpl^{S_{\tau_i}, S_{\sigma_i}}$ is the maximum in the columns corresponding to $e \in F(S_{\tau_i})$. 
Further, the maximum of $\Lpl^{S_{\tau_i}, S_{\sigma_i}}$ over all $i$ is the maximum over all columns, so the lemma follows. 
\end{proof} 

A similar argument gives the following lemma for the upper vertex rather than the lower vertex of $\tau_i$. 
\begin{lemma}
\label{lem:matrix_entries_edge_vert_up}
    $\max_{i} \Lpl^{S_{\tau_i}, S_{\sigma_{i+1}}} = \max \left( D_{G^n}^V \left(M_{\phi}^V \cdot B_F^{\uparrow} - B_{G^n}^{\uparrow} \cdot M_{\phi}^E\right) \right)$
\end{lemma}

Next we consider the thickening matrix type, $\Lpl^{S_\rho, S_\rho^n}$. In this notation, $\rho$ can be either an edge $\tau_i$ or a vertex $\sigma_i$; we start with the vertex case $\rho = \sigma_i$. 

\begin{lemma}
\label{lem:matrix_entries_parallel_thickening}
The diagrams $\{\Parallelograml_\phi(S_{\sigma_i},S_{\sigma_i}^n)\}_i$ all commute if and only if 
\begin{equation*}
M_{\phi^n}^V \cdot I_F^V - I_{G^n}^V \cdot M_{\phi}^V= 0.
\end{equation*}
The loss is determined by
\begin{equation*}
\max_{i} \Lpl^{S_{\sigma_{i}},S_{\sigma_{i}}^n} = \max \left(D_{G^n}^V \left(M_{\phi^n}^V \cdot I_F^V - I_{G^n}^V \cdot M_{\phi}^V\right) \right).
\end{equation*}
\end{lemma}

\begin{proof} %
We are looking for the distance between images in the bottom right of 
\begin{equation*}
\begin{tikzcd}
F(S_{\sigma_i})  
    \ar[r, "{F[\subseteq ]}"] 
    \ar[dr, "{\varphi_{S_{\sigma_i}}}"']
& F^n(S_{\sigma_i})
    \ar[dr, "{\varphi_{ S_{\sigma_i}^n}}"]
& \\
& G^n(S_{\sigma_i}) 
    \ar[r, "{G[\subseteq ]}"'] 
& G^{2n}(S_{\sigma_i})  .
\end{tikzcd}
\end{equation*} 
By a similar argument to \cref{lem:matrix_entries_edge_vert_down}, if we have elements in the diagram of the form
\begin{equation*}
\begin{tikzcd}
v \ar[r, mapsto] \ar[dr, mapsto] & v' \ar[dr, mapsto]\\ 
& \substack{\\w} \ar[r, mapsto, shift right] & \substack{w'\\w''}
\end{tikzcd}
\end{equation*}
then column corresponding to $v$ in  $A = M_{\phi}^V \cdot B_F^{\downarrow} - B_{G^n}^{\downarrow} \cdot M_{\phi}^E$ will be 0 iff $w'=w''$, so the diagrams all commute iff $A = 0$. 
Likewise the column for $v$ in  
$D_{G^n}^V \left(M_{\phi}^V \cdot B_F^{\downarrow} - B_{G^n}^{\downarrow} \cdot M_{\phi}^E\right)$ 
will be entirely 0 if $w' = w''$ and will have maximum entry $d^{G^{2n}}_{S_{\sigma_i}}(w', w'')$ in the entry $[w'',v]$ otherwise. 
The lemma follows.
\end{proof}

Again, a parallel argument shows the edge thickening diagram contribution. 

\begin{lemma}
\label{lem:parallelogramcommute}
The diagrams $\{\Parallelograml_\phi(S_{\tau},S_{\tau}^n)\}_i$ all commute if and only if 
\begin{equation*}
M_{\phi^n}^E \cdot I_F^V - I_{G^n}^E \cdot M_{\phi}^E= 0.
\end{equation*}
The loss is determined by
\begin{equation*}
\max_{i} \Lpl^{S_{\tau_{i}},S_{\tau_{i}}^n} = \max \left(D_{G^n}^E \left(M_{\phi^n}^E \cdot I_F^V - I_{G^n}^E \cdot M_{\phi}^E\right) \right).
\end{equation*}
\end{lemma}

\subsubsection{Triangles}
The last type of loss contribution is that of the triangle diagrams, and the representation of the commutativity of the diagram is similar to the two previous examples. 
However, the one difference in this case is that the resulting loss function is the ceiling of half of the maximum, so in this case we have the following. 

\begin{lemma}
\label{lem:trianglecommute}    
The diagrams $\{\triangled_\phi(S_{\sigma_i})\}_i$ all commute if and only if 
\begin{equation*}
I_{F^n}^V \cdot I_F^V - M_{\psi^n}^V\cdot M_{\phi}^V= 0.
\end{equation*}
The diagrams  $\{\triangleu_\phi(S_{\tau_i})\}_i$ all commute if and only if 
\begin{equation*}
I_{F^n}^E \cdot I_F^E - M_{\psi^n}^E\cdot M_{\phi}^E = 0.
\end{equation*}
The loss in each case is determined by
\begin{align*}
\max_{i} \Ltd^{S_{\sigma_{i}},} &= \left\lceil\frac{1}{2}\max\left(D_{F^{2n}}^V \left(I_{F^n}^V \cdot I_F^V - M_{\psi^n}^V\cdot M_{\phi}^V \right) \right) \right\rceil,\\
\max_{i} \Ltd^{S_{\tau_{i}},} &= \left\lceil\frac{1}{2}\max\left(D_{F^{2n}}^V \left(I_{F^n}^V \cdot I_F^V - M_{\psi^n}^V\cdot M_{\phi}^V \right) \right)\right\rceil.
\end{align*}
\end{lemma}

\begin{proof}
We start with the vertex case for $S_{\sigma_i}$. 
We are checking commutativity of diagrams of the form 
\begin{equation*}
    \begin{tikzcd}
        F(S_{\sigma_i}) 
            \ar[rr, "{F[ \subseteq ]}"]   
            \ar[dr, "\phi_{S_{\sigma_i}}"'] 
            & & F^{2n}(S_{\sigma_i})
        \\
        & G^n(S_{\sigma_i})\ar[ur, "\psi_{S^n_{\sigma_i}}"'].  & 
    \end{tikzcd}
\end{equation*}
Fix a $v \in F(S_{\sigma_i})$, and let $v' \in F^{2n}(S_{\sigma_i})$ be the vertex image along the top of the diagram. 
Let $w = \phi_{S_{\sigma_i}}(v) \in G^n(S_{\sigma_i})$ and $v'' = \psi_{S_{\sigma_i}^n}(w')$ so that checking commutativity of the diagram means checking if $v' = v''$. 

The top map is given in matrix form by $I_{F^n}^V \cdot I_F^V$, so by definition the column for $v$ has exactly one entry 1 in row $v''$. 
For the matrix $M_{\psi^n}^V\cdot M_{\phi}^V$, note that 
the column for $v$ in $M_{\phi}^V$ has exactly one 1 at row $w$, and 
the column for $w$ in $M_{\psi^n}^V$ has exactly one 1 at row $v''$, so the column for $v$ in $M_{\psi^n}^V\cdot M_{\phi}^V$ has exactly one 1 in row $v''$. 
Thus, if $v' = v''$, the column for $v$ in $A = I_{F^n}^V \cdot I_F^V - M_{\psi^n}^V\cdot M_{\phi}^V $ is entirely zero.
Therefore the diagrams $\{\triangled_\phi(S_{\sigma_i})\}_i$ commute if and only if $A = 0$. 
The argument for $\{\triangled_\phi(S_{\tau_i})\}_i$ is the same.

To check the loss function, we again start by focusing on the vertex case. 
As before, we left-multiply the matrix $A$ by the distance matrix $D\coloneqq D_{F^{2n}}^V$.
If $v' = v''$, and thus the column in $A$ for $v$ is 0, the column for $v$ in $D \cdot A$ is also 0. 
If $v' \neq v''$, the entries for the column $v$ are 
\begin{equation*}
(D \cdot A)[u,v] = 
\begin{cases}
    -d(v',v'') & \text{if } u = v'\\
    d(v',v'') & \text{if } u = v'' \\
    d(u,v') - d(u,v'') & \text{else}
\end{cases}
\end{equation*}
where $d$ denotes the distance function $d^{F}_{S^{2n}_\sigma}$, and as before, this column has maximum entry $d(v',v'')$. 
By \cref{df:Loss_by_diagram}, we have that 
\begin{equation*}
\Ltd^S (\phi,\psi)
    = \max\limits_{\alpha \in F(S)}  \Big \lceil \tfrac{1}{2} \cdot d_{S^{2n}}^{F}\left(
    F[S \subseteq S^{2n}] (\alpha),
    \psi_{S^n} \circ \varphi_S(\alpha)
    \right) \Big \rceil
\end{equation*}
so the maximum is taken over $\lceil \tfrac{1}{2} d(v',v'')\rceil$ for all potential inputs $v$. 
The result is that the loss function is $\left\lceil\tfrac{x}{2}\right\rceil$ over all entries in $D \cdot A$ as required. 
As the computation for $S_{\tau_i}$ is the same, we omit it. 
\end{proof}
\cref{lem:matrix_entries_edge_vert_down_binary,lem:matrix_entries_edge_vert_down,lem:matrix_entries_edge_vert_up,lem:parallelogramcommute,lem:trianglecommute,lem:matrix_entries_parallel_thickening} together establish that for fixed $n$, the loss is zero if and only if all diagrams commute, giving the following corollary.
\begin{corollary}
\label{cor:p2_p3_equiv} The problems \cref{prob:BinaryDecision} and \cref{prob:LossComputation} are equivalent. 
\end{corollary}
In practice, this means that either problem may be solved to compute an optimized upper bound on the interleaving distance $d_I(F,G)$. 

\subsubsection{Computing \cref{prob:BinaryDecisionFixedAssignment} and \cref{prob:LossComputationFixedAssignment}}
The result of the lemmata in the previous subsections is that \cref{prob:BinaryDecisionFixedAssignment} and \cref{prob:LossComputationFixedAssignment} can be explicitly answered using the matrices listed in \cref{tab:theList}, as we will see in the two theorems of this section. 
To address \cref{prob:BinaryDecisionFixedAssignment}, observe that if the matrix products (excluding the distance matrix) are identically zero, then the given assignment defines a natural transformation. 

\begin{theorem}
\label{thm:interleaving}
An input $n$-assignment, $\phi$ and $\psi$, constitutes an $n$-interleaving if and only if
\begin{align*}
    \max \Big\{
    & M_{\phi}^V \cdot B_F^{\uparrow} - B_{G^n}^{\uparrow} \cdot M_{\phi}^E, 
    \quad M_{\phi}^V \cdot B_F^{\downarrow} - B_{G^n}^{\downarrow} \cdot M_{\phi}^E,\\
    &M_{\psi}^V \cdot B_G^{\uparrow} - B_{F^n}^{\uparrow} \cdot M_{\psi}^E, 
    \quad M_{\psi}^V \cdot B_G^{\downarrow} - B_{F^n}^{\downarrow} \cdot M_{\psi}^E,\\
    & M_{\phi^n}^V \cdot I_F^V - I_{G^n}^V \cdot M_{\phi}^V, 
    \quad M_{\phi^n}^E \cdot I_F^E - I_{G^n}^E \cdot M_{\phi}^E,\\
    & M_{\psi^n}^V \cdot I_G^V - I_{F^n}^V \cdot M_{\psi}^V, 
    \quad  M_{\psi^n}^E \cdot I_G^E - I_{F^n}^E \cdot M_{\psi}^E,
    \Big\} = 0
\end{align*}
\end{theorem}

\begin{proof}
The proof follows immediately from 
\cref{lem:matrix_entries_edge_vert_down_binary,lem:matrix_entries_parallel_thickening,lem:parallelogramcommute,lem:trianglecommute} and their symmetric counterparts.
\end{proof}

The loss function \cref{eq:loss} can be computed by taking the maximum over the entries of the matrices listed in \cref{tab:theList}, thus giving an answer to \cref{prob:LossComputationFixedAssignment}. 
Explicitly, we have the following theorem.

\begin{theorem}
\label{thm:LossEquivMatrices}
For a given $n$-assignment $(\phi, \psi)$,
\begin{align*}
    L_B(\phi,\psi) &= \\
    \max \Big\{  
    & D_{G^n}^V \left(M_{\phi}^V \cdot B_F^{\uparrow} - B_{G^n}^{\uparrow} \cdot M_{\phi}^E\right), \quad D_{G^n}^V \left(M_{\phi}^V \cdot B_F^{\downarrow} - B_{G^n}^{\downarrow} \cdot M_{\phi}^E\right),\\
    &D_{F^n}^V \left(M_{\psi}^V \cdot B_G^{\uparrow} - B_{F^n}^{\uparrow} \cdot M_{\psi}^E\right), \quad D_{F^n}^V \left(M_{\psi}^V \cdot B_G^{\downarrow} - B_{F^n}^{\downarrow} \cdot M_{\psi}^E\right),\\
    & D_{G^{2n}}^V \left(M_{\phi^n}^V \cdot I_F^V - I_{G^n}^V \cdot M_{\phi}^V\right), \quad D_{G^{2n}}^E \left(M_{\phi^n}^E \cdot I_F^E - I_{G^n}^E \cdot M_{\phi}^E\right),\\
    & D_{F^{2n}}^V \left(M_{\psi^n}^V \cdot I_G^V - I_{F^n}^V \cdot M_{\psi}^V\right), \quad D_{F^{2n}}^E \left(M_{\psi^n}^E \cdot I_G^E - I_{F^n}^E \cdot M_{\psi}^E\right),\\
    & 
     \left\lceil \tfrac{1}{2} D_{F^{2n}}^V \left(I_{F^n}^V \cdot I_F^V - M_{\psi^n}^V\cdot M_{\phi}^V \right) \right\rceil, 
    \quad 
     \left\lceil \tfrac{1}{2} D_{F^{2n}}^E \left(I_{F^n}^E \cdot I_F^E - M_{\psi^n}^E\cdot M_{\phi}^E \right)\right\rceil,\\
    & 
     \left\lceil \tfrac{1}{2} D_{G^{2n}}^V \left(I_{G^n}^V \cdot I_G^V - M_{\phi^n}^V\cdot M_{\psi}^V \right)\right\rceil, 
    \quad 
     \left\lceil \tfrac{1}{2} D_{G^{2n}}^E \left(I_{G^n}^E \cdot I_G^E - M_{\phi^n}^E\cdot M_{\psi}^E \right)\right\rceil
    \Big\}
\end{align*}
where the $\max$ is taken over all entries in all matrices, and $\lceil \cdot \rceil$ denotes the entry-wise ceiling of the matrix inside.
\end{theorem}

\begin{proof}
    The proof follows from \cref{lem:matrix_entries_edge_vert_down,lem:matrix_entries_edge_vert_up,lem:parallelogramcommute,lem:trianglecommute,lem:matrix_entries_parallel_thickening} and their symmetric counterparts.
\end{proof}

\section{Implementation: Integer Linear Programming}
\label{sec:ILP}

Our matrix formulation above leads to the natural question: can we formulate loss function computation as a linear program?
Indeed, we will show that both the objective function and constraints
of the loss function optimization problem can be expressed as linear relationships.
By \cref{thm:LossEquivMatrices}, evaluating the loss function reduces to finding the maximum entries across a collection of matrix products listed in \cref{tab:theList}. The key idea is to treat the entries of all $M_\phi$ and $M_\psi$ matrices as variables and then solve a linear program that selects these entries so as to minimize the resulting loss.
However, due to the discrete nature of the problem, the resulting formulation takes the form of an integer linear program (ILP). 
This comes with a downside in complexity, as solving an ILP is NP-complete~\cite[Pg.~245]{Garey1979}.
That said, computing the interleaving distance is known to be NP-hard (see \cite{bjerkevik2017computational,bjerkevik2020computing}, and the discussion in \cref{sec:Problems}), so we have no reason to attempt to design an algorithm that runs in polynomial time. 
Nevertheless, developing efficient ILP solvers is an active area of research, resulting in many solvers that are extremely efficient in practice, so this translation allows us to take advantage of the current best known solvers. 

\subsection{ILP Model Description}
\label{ssec:ILP_model_description}
The goal of this optimization is either to find an $n$-interleaving (\cref{prob:BinaryDecision}) or search over possible $n$-assignments   to minimize the loss function (\cref{prob:LossComputation}).
This ends up being essentially the same setup, depending on whether we use the multiplication by distance matrices $D_\bullet^\bullet$ in \cref{tab:theList}. 
To that end, we describe the ILP here for the loss function version, since determining if there is an interleaving in (\cref{prob:BinaryDecision}) is the same setup but without inclusion of the $D$ matrices. 

In the case of finding the loss, the objective function for the ILP is to minimize $L_B(\varphi, \psi)$ over all possible $n$-assignments $(\phi,\psi)$. 
We denote the objective function as $\ell$ in the ILP.
As the loss is calculated on assignments, the entries in the blocks of the eight assignment matrices 
$M_\eta^A$ for $\eta \in \{ \phi, \phi^n, \psi, \psi^n\}$ and $A \in \{V, E\}$ serve as decision variables.
Using \cref{thm:LossEquivMatrices}, we know that the loss for each individual diagram can be computed by finding the maximum absolute value or the maximum of the ceiling of half the value among the entries of matrix multiplication terms listed in \cref{tab:theList}. 
Thus, we impose constraints in the ILP to connect the matrix representation of each term in $L_B(\varphi, \psi)$. 
We model the loss $\ell$ to be greater than or equal to the relevant entry in these matrices. 
Since it is a minimization problem, the ILP seeks the smallest $\ell$ that satisfies this condition, ensuring the equality.
Below we describe how the three diagram types (see \cref{tab:theList}) contribute to the ILP formulation. 

Reiterating above, our goal is 
\begin{equation*}
    \text{Minimize} \quad \ell 
\end{equation*}
with respect to the variables $ \{ x_{ij}^{\eta, A} \in M_\eta^A \mid  \eta \in \{ \phi, \phi^n, \psi, \psi^n\}, A \in \{V, E\} \}$ restricted to those entries in the blocks of the relevant matrices (we can implicitly assume $ x_{ij}^{\eta, A} = 0$ outside of these blocks). 
To ensure that the resulting matrices have a single one in each column, we include the first collection of constraints,
   \begin{align*}
        \text{Subject to}
        \qquad 
        \sum_i x_{ij}^{\eta,A} = 1 \qquad \forall x_{ij}^{\eta,A} \in M_\eta^A, \, \eta \in \{ \phi, \phi^n, \psi, \psi^n\}, \, A \in \{V,E\},
    \end{align*}
which ensure that all columns sum to 1.

The next set of constraints are associated to  the edge-vertex parallelograms, where the corresponding diagrams are $\Lpl^{S_\tau, S_\sigma}$ and $\Lpr^{S_\tau, S_\sigma}$.  
For example, in the diagram $\Lpl^{S_\tau, S_\sigma}$, the loss is given by 
    $
       \Lpl^{S_\tau, S_\sigma} =  \max \left( D_{G^n}^V \left(M_{\phi}^V \cdot B_F^{\uparrow} - B_{G^n}^{\uparrow} \cdot M_{\phi}^E\right)\right).
       $
Therefore, we set the constraint that the objective function $\ell$ must be larger than any entry in the multiplied matrix. 
Thus, we include the constraints 
\begin{align*}
\text{Subject to} \quad 
 & \ell \geq  \max\left({D_{G^n}^V \left(M_{\phi}^V \cdot B_F^{\uparrow} - B_{G^n}^{\uparrow} \cdot M_{\phi}^E\right)}\right) \\
& \ell \geq \max\left({D_{G^n}^V \left(M_{\phi}^V \cdot B_F^{\downarrow} - B_{G^n}^{\downarrow} \cdot M_{\phi}^E\right)}\right)  \\
& \ell \geq \max\left({D_{F^n}^V \left(M_{\psi}^V \cdot B_G^{\uparrow} - B_{F^n}^{\uparrow} \cdot M_{\psi}^E\right)}\right) \\
& \ell \geq \max\left({ D_{F^n}^V \left(M_{\psi}^V \cdot B_G^{\downarrow} - B_{F^n}^{\downarrow} \cdot M_{\psi}^E\right)}\right).
\end{align*}
    We note that the $D$ and $B$ matrices  are constants in the ILP.
    Therefore, the matrix multiplication expression is linear and suitable for ILP.

For the thickening parallelograms, the corresponding losses are $\Lpl^{S_\rho, S_\rho^n}$ and $\Lpr^{S_\rho, S_\rho^n}$.  
As in the vertex edge case, we need to ensure that all entries in the relevant multiplied matrices from \cref{tab:theList} are bounded by $\ell$. 
Thus we add the following constraints 
\begin{align*}
\text{Subject to} \quad 
& \ell \geq \max\left( D_{G^n}^V \left(M_{\phi^n}^V \cdot I_F^V - I_{G^n}^V \cdot M_{\phi}^V\right)\right) \\
&\ell \geq \max\left( D_{G^n}^E \left(M_{\phi^n}^E \cdot I_F^E - I_{G^n}^E \cdot M_{\phi}^E\right)\right)\\ 
&\ell \geq \max\left( D_{F^n}^V \left(M_{\psi^n}^V \cdot I_G^V - I_{F^n}^V \cdot M_{\psi}^V\right)\right)\\
&\ell \geq \max\left( D_{F^n}^E \left(M_{\psi^n}^E \cdot I_G^E - I_{F^n}^E \cdot M_{\psi}^E\right)\right)
\end{align*}
    Similar to the edge-vertex parallelograms, the $D$ and $I$ matrices are constants in ILP, which makes these constraints linear.

Finally, we include constraints to control the entries on the triangle diagram matrices with corresponding loss functions  $\Ltd^{S_\rho}$ and $\Ltu^{S_\rho}$.   
We first  discuss the diagram $\Ltd^{S_{\sigma_i}}$ which operates on vertices. 
The loss contribution from this type of diagram can be obtained by finding 
    $
    \max_{i} \Ltd^{S_{\sigma_{i}}} = \max \left\{ \left\lceil\tfrac{x}{2}\right\rceil  \, \middle|\,  x \in D_{F^{2n}}^V \left(I_{F^n}^V \cdot I_F^V - M_{\psi^n}^V\cdot M_{\phi}^V \right) \right\}.
    $
Unlike the previous two types of diagrams, the expression $D_{F^{2n}}^V \left(I_{F^n}^V \cdot I_F^V - M_{\psi^n}^V\cdot M_{\phi}^V \right)$ is quadratic rather than linear, as it includes the term $M_{\psi^n}^V\cdot M_{\phi}^V$, which is a multiplication of two variable matrices. 
However, we can linearize the problem by introducing a collection of integer and binary reparametrization variables and adding them to the constraints.
    Furthermore, the ceiling function that we need to compute here is not a linear function either.
    This is taken care of with another reparametrization integer variable.
Specifically, we create the variables $\{z_{ijk}  \}$, one to represent each entry $x_{ij}^{\psi_n,V} x_{jk}^{\phi,V}$, so that if we construct the matrix $M$ where $M_{ik} = \sum_j z_{ijk} $, then 
${M = (M_{\psi^n}^V)_{ij}(M_{\phi^n}^V)_{jk}}$. 
Then, we get the following constraints for $\Ltd^{S_{\sigma_{i}}}$:
    \begin{align*}
        \text{Subject to} 
        \begin{alignedat}[t]{3}
        \quad & 2c_{ij} \geq k_{ij} \quad  
            &\forall \quad k_{ij} \in D_{F^{2n}}^V \left(I_{F^n}^V \cdot I_F^V - M \right)\\
        \quad & z_{ijk} \leq x_{ij}^{\psi^n,V} 
            &\forall \quad x_{ij}^{\psi^n,V} \in M_{\psi^n}^V\\
        \quad & z_{ijk} \leq x_{jk}^{\phi,V} \quad 
            &\forall \quad x_{jk}^{\phi,V} \in M_{\phi}^V\\
        \quad & z_{ijk} \geq x_{ij}^{\psi^n,V} + x_{jk}^{\phi,V} - 1\\
        \quad & \ell \geq c_{ij}.
        \end{alignedat}
    \end{align*}
    These constraints are only for one of the four relevant triangle diagrams, however the other setups are similar. 

Our full linear program is formulated to minimize $\ell$ subject to the constraints described above. Specifically, let
\[
V_{\max} = \max\{|V_H| \mid H \in F, F^n, G, G^n\}, \quad 
E_{\max} = \max\{|E_H| \mid H \in F, F^n, G, G^n\}.
\]
The total number of variables and constraints in the ILP can be bounded in terms of these graph sizes.

The primary binary decision variables are the entries of the eight assignment matrices $M_\eta^A$, where $A \in \{\varphi, \varphi^n, \psi, \psi^n\}$ and $A \in \{V, E\}$. Each matrix has a block structure, with one block for each function value $i \in [-L, L]$. Each block has size at most $V_{\max}^2$ (respectively $E_{\max}^2$) for vertex-type (respectively edge-type) matrices, yielding a total of $O\left(L \cdot (V_{\max}^2 + E_{\max}^2)\right)$ primary binary variables.

The remaining variables are auxiliary reparameterization variables $z_{ijk}$, introduced to linearize products such as $M_{\psi^n}^V M_\varphi^V$ appearing in the triangle-diagram constraints, along with integer variables $c_{ij}$ to handle ceiling operations. Each triangle product at level $i$ introduces one binary variable $z_{ijk}$ for each triple $(i,j,k)$, contributing up to $V_{\max}^3$ variables per function value. Across all four triangle diagrams, this yields $O\left(L \cdot (V_{\max}^3 + E_{\max}^3)\right)$ auxiliary variables.
Consequently, the total number of variables and constraints in the ILP is $O\left(L \cdot (V_{\max}^3 + E_{\max}^3)\right)$, which is polynomial in the input size. We note that the cubic dependence on $V_{\max}$ and $E_{\max}$ may be reducible with a more refined parameterization, which we leave for future work.

\subsection{Python Implementation Details}
\label{ssec:PythonImplementationILP}
    We implement the ILP using PuLP library in Python which is an open-source software package in Python specializing in modeling and solving mixed-integer linear programming problems \cite{mitchell2011pulp}. 
    Specifically, PuLP is equipped to solve ILP problems involving integer and binary variables, as in our case. 
    PuLP also integrates seamlessly with popular ILP solvers, both commercial (e.g. Gurobi \cite{gurobi}) and open source (e.g. CBC \cite{cbc}, GLPK \cite{glpk}), without having to alter the model formulation. 
    It is important to note that PuLP is exclusively designed to model \textit{linear} optimization problems, hence the extra work done in \cref{ssec:ILP_model_description} to linearize our ILP.    
    We choose PuLP over available non-linear modeling frameworks like Pyomo~\cite{hart2009python} because it is lightweight and has a simpler syntax, making it easier to interpret and update.
    Further, the design of PuLP is streamlined for ILP computation, reducing the computational overhead.
    We note that given an ILP, PuLP will return one of five options\footnote{This is discussed in the PuLP documentation at \href{https://coin-or.github.io/pulp/technical/constants.html}{coin-or.github.io/pulp/technical/constants.html}}: 
\textbf{(1) or (-1):} it terminates with the optimal solution (feasible or infeasible); 
\textbf{(0):} it cannot solve the problem given the time and memory constraints; 
\textbf{(-2):} the problem is unbounded; or 
\textbf{(-3):} the problem is undefined.
However, in our experiments, the solver always terminates successfully (output \textbf{(1)} or \textbf{(-1)}). 

    \begin{figure}
        \centering
        \includegraphics[width=0.7\linewidth]{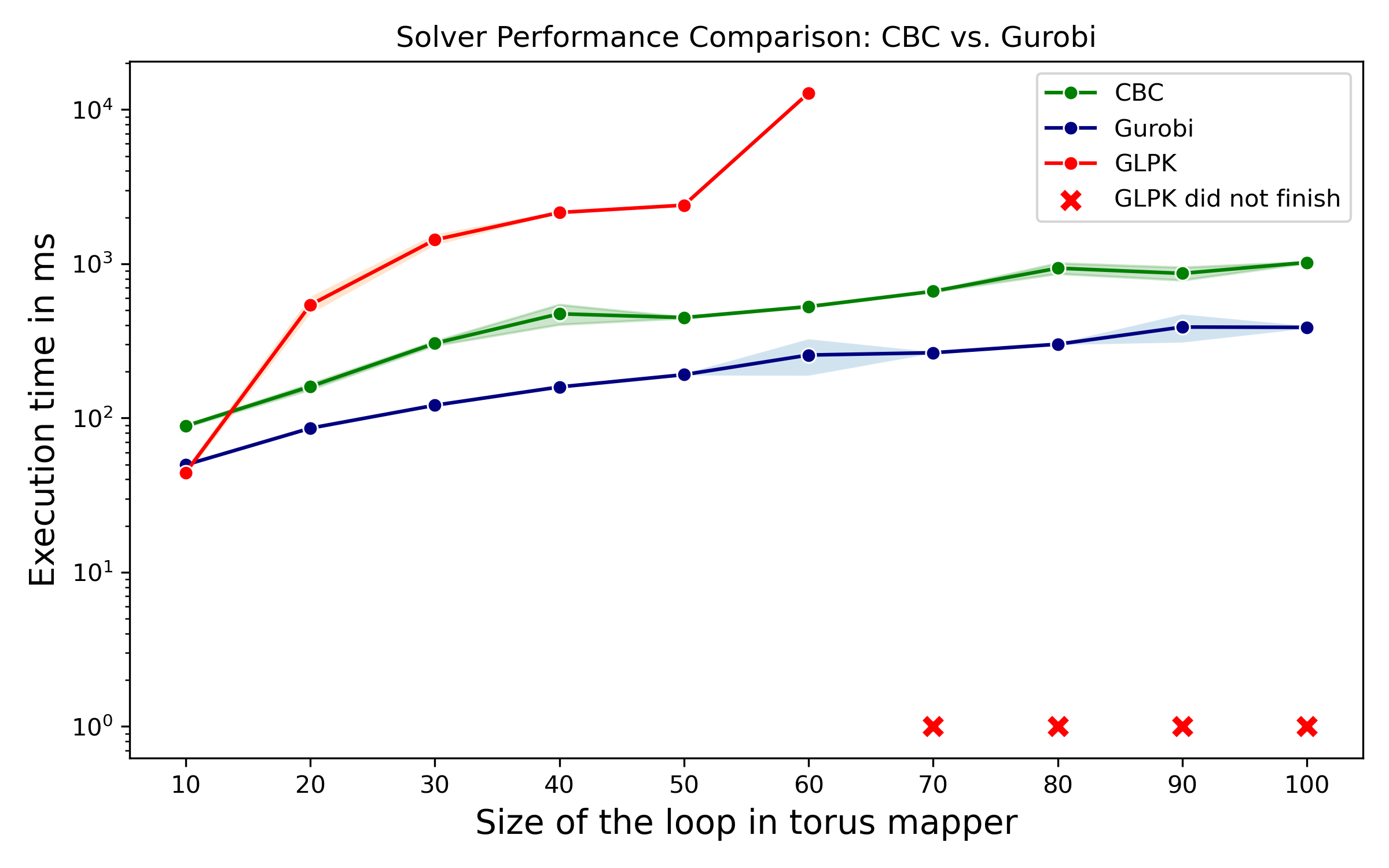}
        \caption{Time taken by different solvers to optimize the loss for the torus and line mappers.}
        \label{fig:solver_timing}
    \end{figure}
    
    We test our ILP implementation with simple input examples to make a choice of solver for use in PuLP; 
    the specifics of this particular example pair will be given in \cref{ssec:line_torus_sec}.
    In this example, increasing loop size increases the size of the input graphs, thus the size of the ILP increases.
   We test two open-source solvers, CBC (COIN-OR Branch and Cut) and GLPK (GNU Linear Programming Kit), as well as the commercial solver Gurobi. 
    The results are shown in \cref{fig:solver_timing}.
    We see that for smaller problem size (up to a loop size of 20 in the torus mapper), all of the solvers take less than one second to optimize the loss.
    However, as the size of the loop increases, the GLPK solver becomes significantly slower than CBC or Gurobi. 
    In fact, for loop sizes of 70 and above, GLPK fails to complete computation even after 15 minutes of runtime.
    As expected, the commercial solver Gurobi is highly optimized, resulting in superior performance. 
    Both CBC and Gurobi's computation times increased approximately linearly with the loop size.    
    Among the open-source solvers we tested, CBC is slower than Gurobi but is still able to optimize large problems effectively.
    Since CBC is the default solver in PuLP, it is also more accessible from a user’s perspective.
    Thus, for the rest of the experimental section, we exclusively use the CBC solver inside of PuLP.
    The code for our implementation is available in the Python package \texttt{ceREEBerus} \cite{cereeberus}.

\section{Experiments}
\label{sec:experiments}
We run all experiments on a high-performance computing cluster with dual AMD EPYC 9654 processors, 192 CPU cores, and 768GB RAM, using 10–32 cores per run.
\subsection{Experiments with Line and Torus Mappers}
\label{ssec:line_torus_sec}

\begin{figure}
    \centering
    \includegraphics[width=0.45\linewidth, align = c]{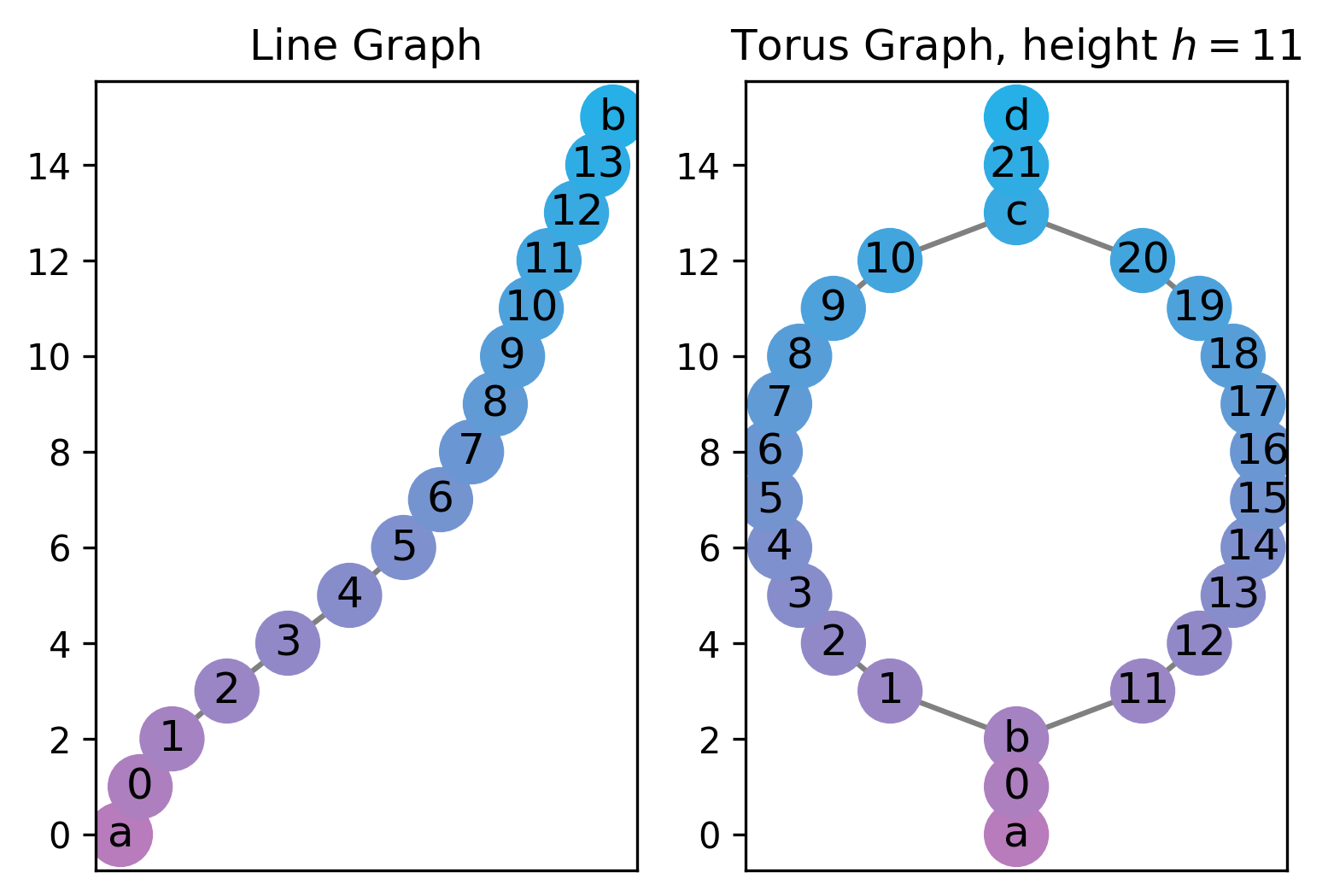}
    \includegraphics[width=0.54\linewidth, align= c]{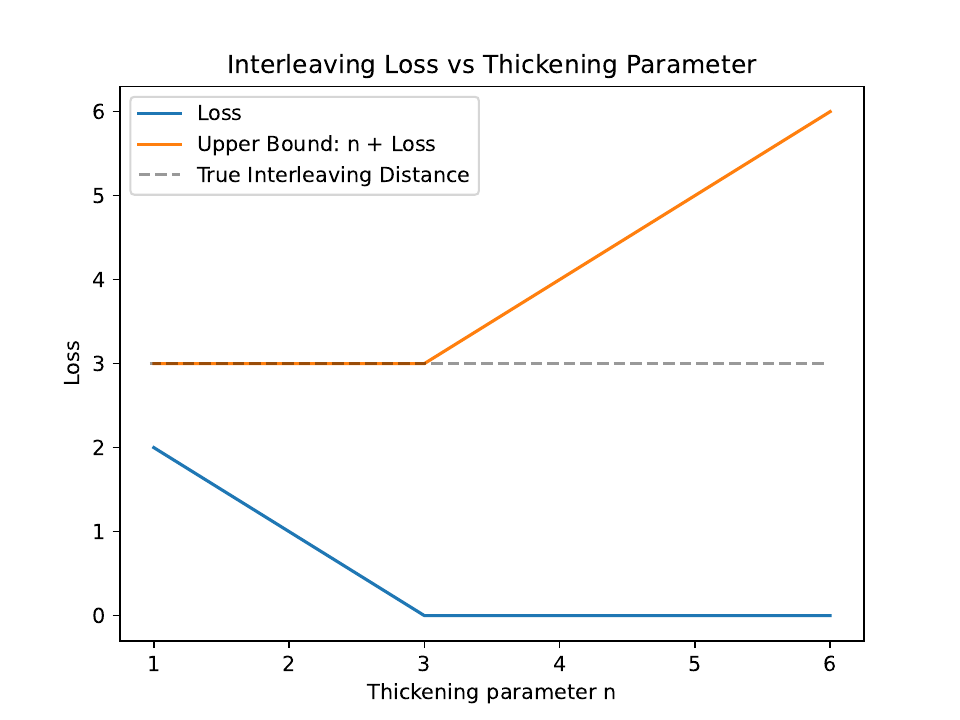}
    \caption{(Left) Examples of line and torus mappers, the latter featuring a loop of height 11. (Right) The loss function and resulting upper bound computed for this pair of graphs over varying values of $n$. In both figures, letters denote critical vertices and numbers denote regular vertices.}
    \label{fig:loss_vs_n}
\end{figure}

We validate our loss-function optimization by testing it on two simple mapper graphs whose interleaving distances can be computed by hand. 
The first is the \emph{line mapper graph} $L$, consisting of a single vertex at each integer function value and edges connecting sequential vertices. 
The second is the \emph{torus mapper} $T_h$, which is a mapper graph of a height function on an upright torus and contains a loop of height $h$.
Both mappers are defined over the same range of function values; see the left panel of \cref{fig:loss_vs_n} for an example with $h = 11$. 
It is straightforward to verify that 
$d_I(L, T_h) = \lceil h/4 \rceil$; for $h = 11$, this gives $d_I(L, T_{11}) = 3$. 

\subsubsection{Loss Optimization vs Theoretical Computations}

We compare the results of the ILP with the hand-computed interleaving distance on the line and torus mappers. In the right panel of \cref{fig:loss_vs_n}, we see that the loss computed for varying choices of $n$ achieves the true interleaving distance for all $n$ up to and including the actual interleaving value. 
Thus, in this particular example, a search over $n$ would not be required. However, this behavior cannot be assumed \textit{a priori}; see \cref{ex:binarySearchExample} for a case where the search is indeed necessary.
Note that we obtain the interleaving maps as a byproduct of this computation, see \cref{fig:optimizedmap} for an example of the output optimal interleavings.
We compute the distances across varying loop height ($h$).
For each mapper pair, we perform 50 iterations of loss optimization and in every case the ILP returns the optimal solution with the same optimized loss value, empirically demonstrating the stability of our optimization pipeline. 

\begin{figure}
\centering
        \includegraphics[width=.8\linewidth, align = c]{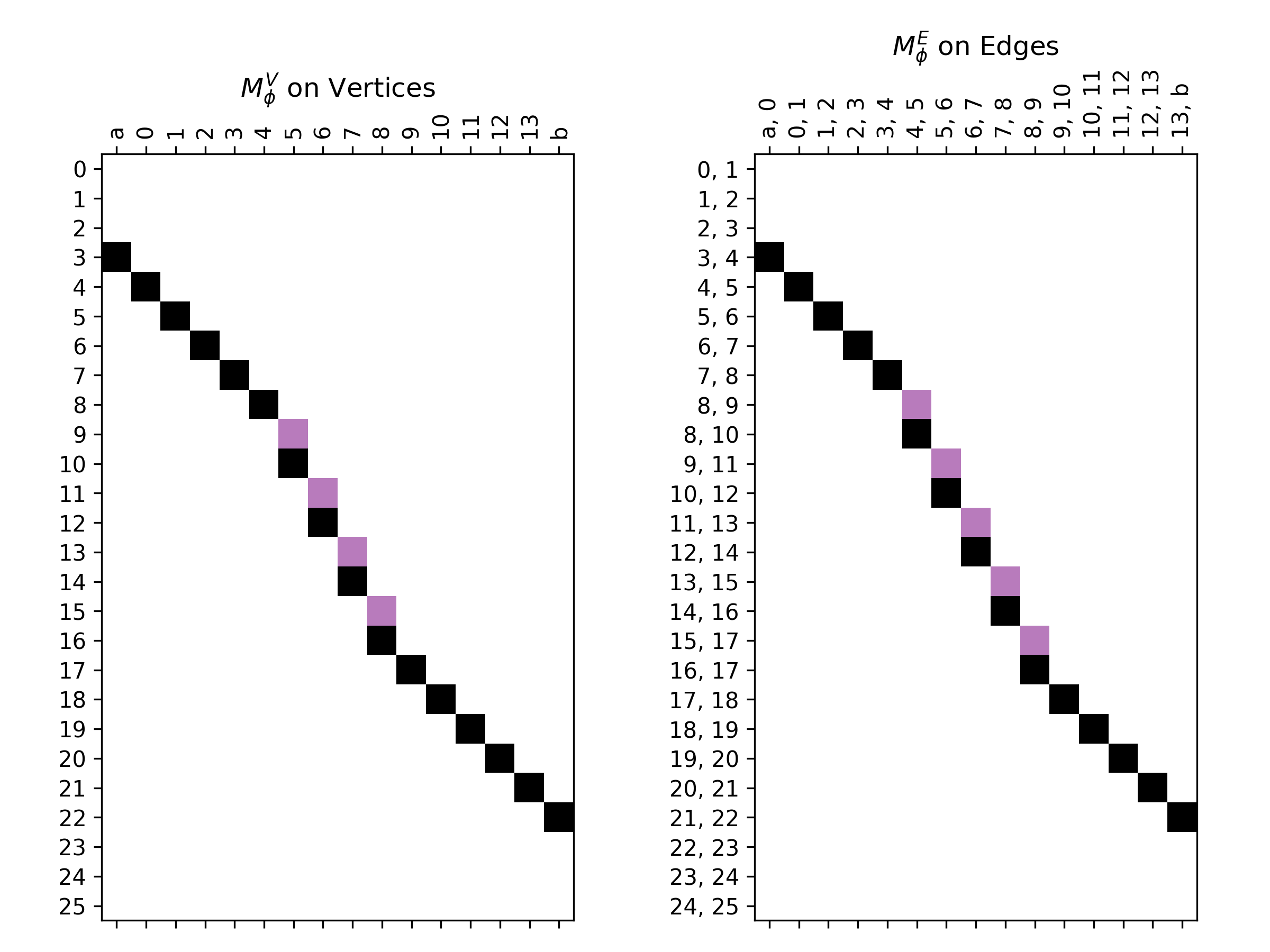}
        \includegraphics[width = .07\textwidth, align = c]{colorbar-PurBl-binary.png}
\caption{The optimized map from the line mapper to the 1-smoothed torus from the example in \cref{fig:loss_vs_n} (left)).}
\label{fig:optimizedmap}
\end{figure}

\subsubsection{Comparing Loss Optimization with \cref{prob:BinaryDecision} and \cref{prob:LossComputation}}
\label{ssec: p2_p3_benchmark}
By \cref{cor:p2_p3_equiv}, the optimization problems \cref{prob:BinaryDecision} and 
\cref{prob:LossComputation} are equivalent and either may be used 
to compute an upper bound on the interleaving distance.
We compare the runtime of \cref{prob:BinaryDecision} and \cref{prob:LossComputation} on the same torus and line mapper pairs described above, varying the loop height $(h)$. 
For each $h$, we measure both the total runtime, which includes both pre-computing the necessary matrices and the ILP optimization, and the ILP optimization time alone, averaged over 10 repetitions.

\begin{figure}
    \centering
    \includegraphics[width=0.7\linewidth]{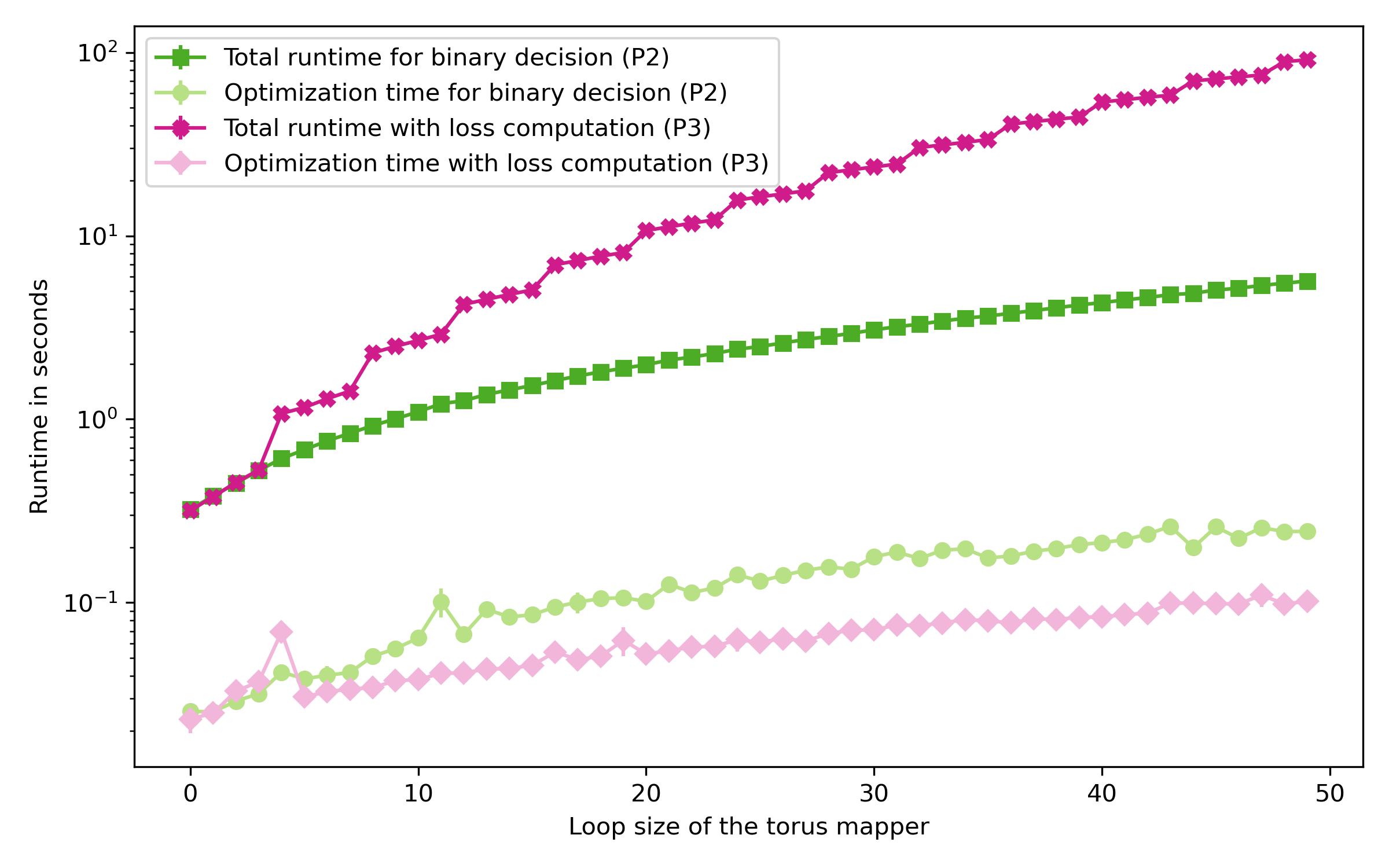}
    \caption{Total and ILP optimization runtime (log scale) for \cref{prob:BinaryDecision} (binary decision, green) and \cref{prob:LossComputation} (loss computation, pink), varying the loop height $h$ of the torus mapper, averaged over 10 repetitions. The additional overhead in \cref{prob:LossComputation} is due to the pre-computation of the distance matrix D.}
    \label{fig:p2_p3_benchmark}
\end{figure}

As shown in \cref{fig:p2_p3_benchmark}, \cref{prob:BinaryDecision} is consistently faster than \cref{prob:LossComputation} in total runtime, except for very small mappers, where there's no difference in runtime. 
However, as the mappers grow in size, the gap in runtime widens. 
On the other hand, when we compare ILP optimization time alone, \cref{prob:LossComputation} is marginally faster than \cref{prob:BinaryDecision}.
This suggests that the overhead in \cref{prob:LossComputation} mostly comes from the pre-processing cost of computing the distance matrix $D$, which requires performing shortest-path computations on the thickened graph for each function value block at every iteration.

We attribute the slightly faster ILP solve time in \cref{prob:LossComputation} to the difference in the problem formulation.
\cref{prob:BinaryDecision} is a pure feasibility problem with hard equality constraints, which demands the solver to certify infeasibility when no solution exists.
\cref{prob:LossComputation}, in contrast, tries to minimize the loss function, which always has a finite solution.
In all experiments, both \cref{prob:BinaryDecision} and 
\cref{prob:LossComputation} return the same interleaving distance 
bound, consistent with \cref{cor:p2_p3_equiv}.
These results suggest that while the two formulations are theoretically equivalent, \cref{prob:BinaryDecision} is currently preferable in practice when computational efficiency is the priority.

\subsection{Classification Experiments}
\label{ssec:ClassificationBrief}

We ran two proof-of-concept classification experiments using this optimized loss on the mapper interleaving distance. 
The details are provided as follows.
\subsubsection{Classification of Binary Images}
\label{ssec:ClassificationMPEG7}

We conducted a classification experiment on the MPEG-7 dataset \cite{jeannin1999}, specifically MPEG-7 CE
Shape-1 Part B, which is a widely used benchmark for shape analysis and classification with binary
images.
The results are shown in \cref{fig:image_mds}.
Our methodology involves computing mapper graphs from the dataset and analyzing their interleaving distance for classification tasks.

In this experiment, we choose six image classes, \emph{apples, cups, forks, hearts, horseshoes} and \emph{sea snakes}, and select 10 images from each class.
See \cref{fig:image_mds}  (left) for examples in each category.
We convert these images to mapper graphs using {\it KeplerMapper} \cite{van2019kepler}, where the lens function is given by the $y$-coordinate of each pixel. For every point cloud, we use a cover of 10 elements with a $30\%$ overlap, and set the epsilon value for DBSCAN to 3. These parameters are chosen to generate a mapper graph which best represents the input image.
For each resulting mapper graph, we choose parameter values to ensure that there is only one connected component in the graph.
For each node, we assign the function value, or \emph{height}, of that node to be the midpoint of the cover element it belongs to, and normalize so that all node heights are in the range $[0, 20]$. 
This normalization ensures consistency in mapper size, simplifying loss computation.

Once the mapper graphs are obtained, we compute the optimal upper bound of the interleaving distance between them using the optimization strategy discussed earlier, which serves as a measure of similarity between mapper graphs. 
To visualize the distance performance, we construct a {\it multi-dimensional scaling} (MDS) \cite{borg2007modern} plot (\cref{fig:image_mds}, right) using the computed distance matrix (MDS stress is 0.14, which ensures good goodness-of-fit). 
MDS is chosen for visualization as it accepts an arbitrary distance matrix as input. 

As we see in \cref{fig:image_mds} (right), the optimal upper bounds on distance successfully cluster similar shapes together. 
This behavior reflects what the interleaving distance is designed to capture: the distance is small only when mapper graphs agree in both graph structure \emph{and} lens function. 
Among all categories, {\it apple} exhibits the most well-defined grouping. 
Additionally, some of the other shapes are well separated, such as {\it apple} and {\it fork} or {\it seasnake}.
However, certain shapes are not differentiated by this distance. 
For instance, the {\it cup} category in the MPEG-7 dataset includes two types of cups: one with joined handles, forming a loop, and another without joined handles. 
This structural difference is reflected in the mapper representation, causing a few to appear far apart despite belonging to the same category.  
This is to be expected since the mapper interleaving distance is built to find similar topological structure in the resulting graph, and does not take the geometric similarity into account. 
Likewise, the horseshoes and sea snakes, despite being similar structurally, have a large distance between them since the mapper graph lens function takes the embedding direction into account; the horseshoe with opening of the U-shape pointing down and the sea-snake with the opening point up are considered very different under this distance. 

We also use the pairwise distances to classify the images into their correct classes using $k$-nearest neighbors ($k$-NN) \cite{cover1967nearest}, based on the optimized loss computed on the mappers.
To determine the optimal number of neighbors $(k)$ and avoid overfitting or underfitting, we first perform a $5$-fold cross-validation on the data, testing values of $k$ from $1$ to $30$.
The optimal $k$ was selected to be $1$, which resulted in the highest and most stable accuracy of $91.67\%$, as shown in \cref{fig:image_knn} (left).
\begin{figure}
    \centering
    \includegraphics[width=0.46\textwidth]{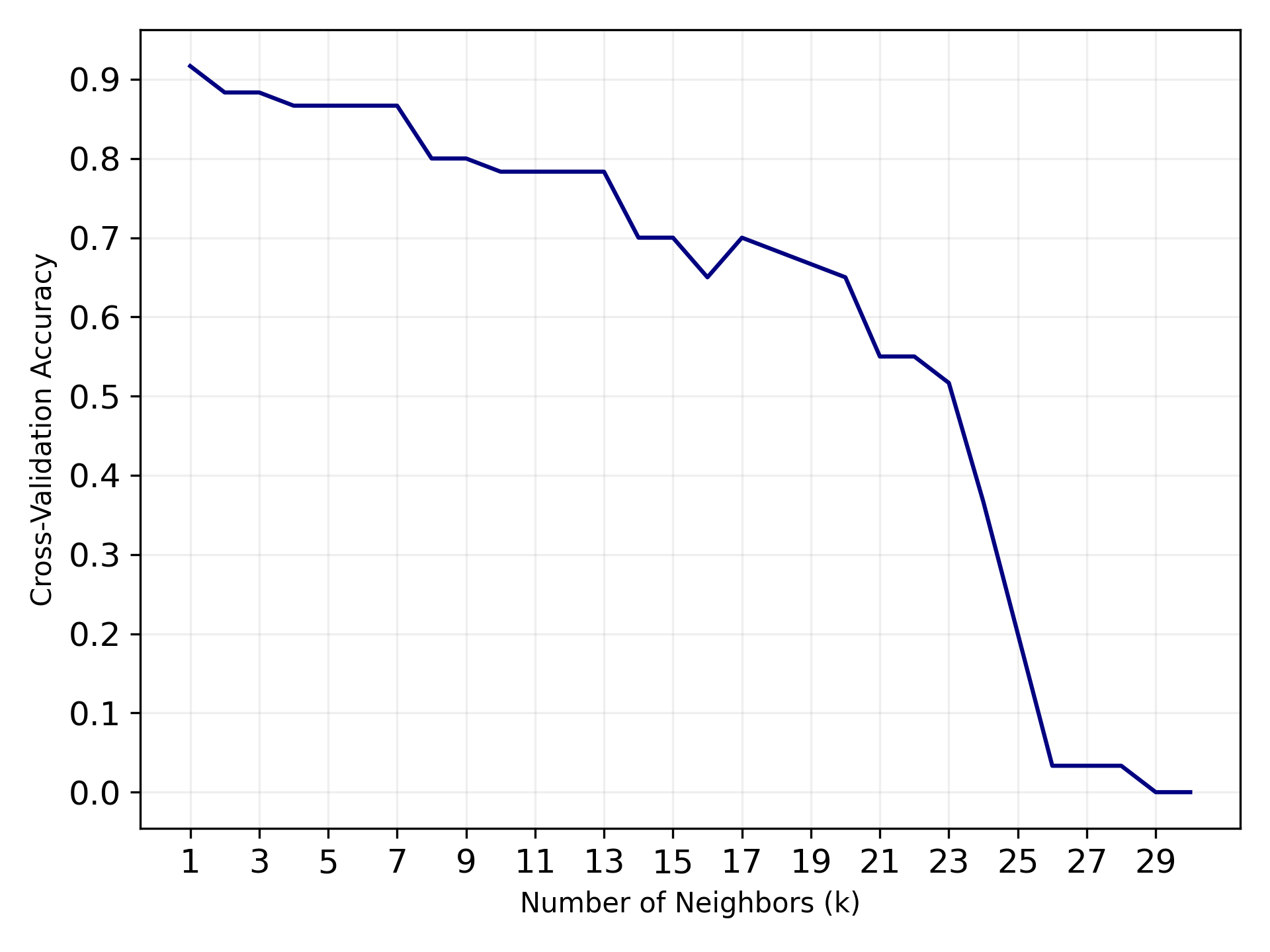}
    \includegraphics[width=0.45\textwidth]{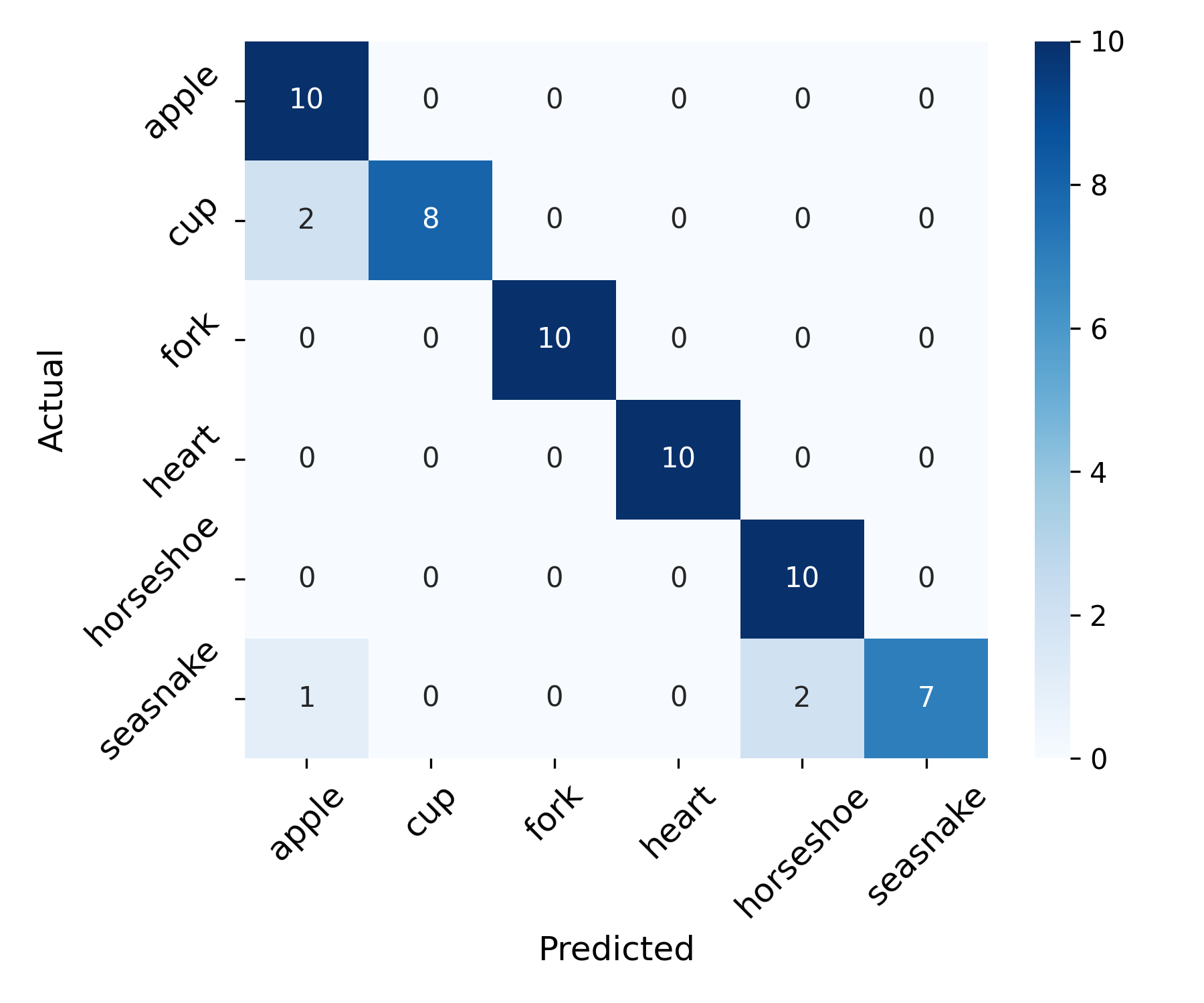}
    \caption{(Left) Accuracy of KNN classifier for different $k$ values during 5-fold cross-validation. The optimal $k = 1$ achieved $91.67\%$ accuracy.\\
    (Right) Confusion matrix showing how the actual classes match with the predicted class. Each $(i,j)$-th cell indicate number of times an image from category $i$ was classified as category $j$.}
    \label{fig:image_knn}
\end{figure}
For a more robust evaluation, we also implement {\it leave-one-out cross-validation} (LOO-CV), where one data point is used as a test case while the rest are used for training the model.
To visualize the classification and misclassification, we display a {\it confusion matrix}, which highlights the categories that are being misclassified, see \cref{fig:image_knn} (right).
The LOO-CV results confirm the model's accuracy of $91.67\%$, indicating that while the model is generally effective, there are still some misclassifications.
It is important to note that this is a preliminary experiment with only 10 images per category (totaling 60 images), and expanding to the full MPEG-7 dataset might improve accuracy.
As we see in the confusion matrix in \cref{fig:image_knn} (right), the categories that suffer from misclassification are {\it cup} and {\it seasnake}. 
As discussed before, the some of the cups have joined handles and some do not, which contributes to the misclassification.
Additionally, the sea snake images are highly varied in how they are oriented, leading to misclassification due to differences in their embeddings, which results in generation of very different mapper graphs.

These results suggest that the computed distance bound indeed captures meaningful topological differences in the data, indicating its potential usefulness in further real-world applications.

\begin{figure}
    \centering
    \includegraphics[width=0.42\linewidth, align = c]{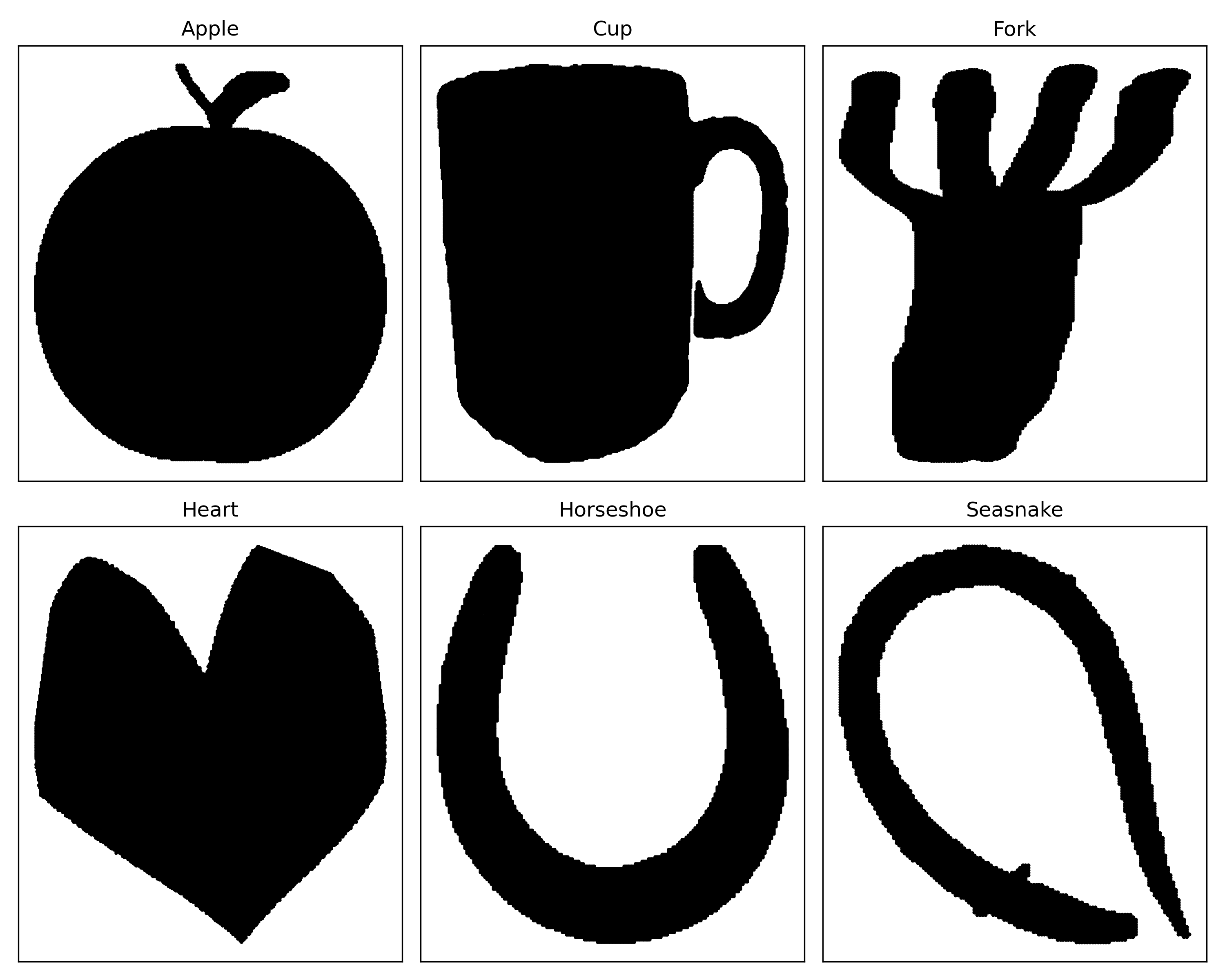}
    \includegraphics[width=0.52\linewidth, align = c]{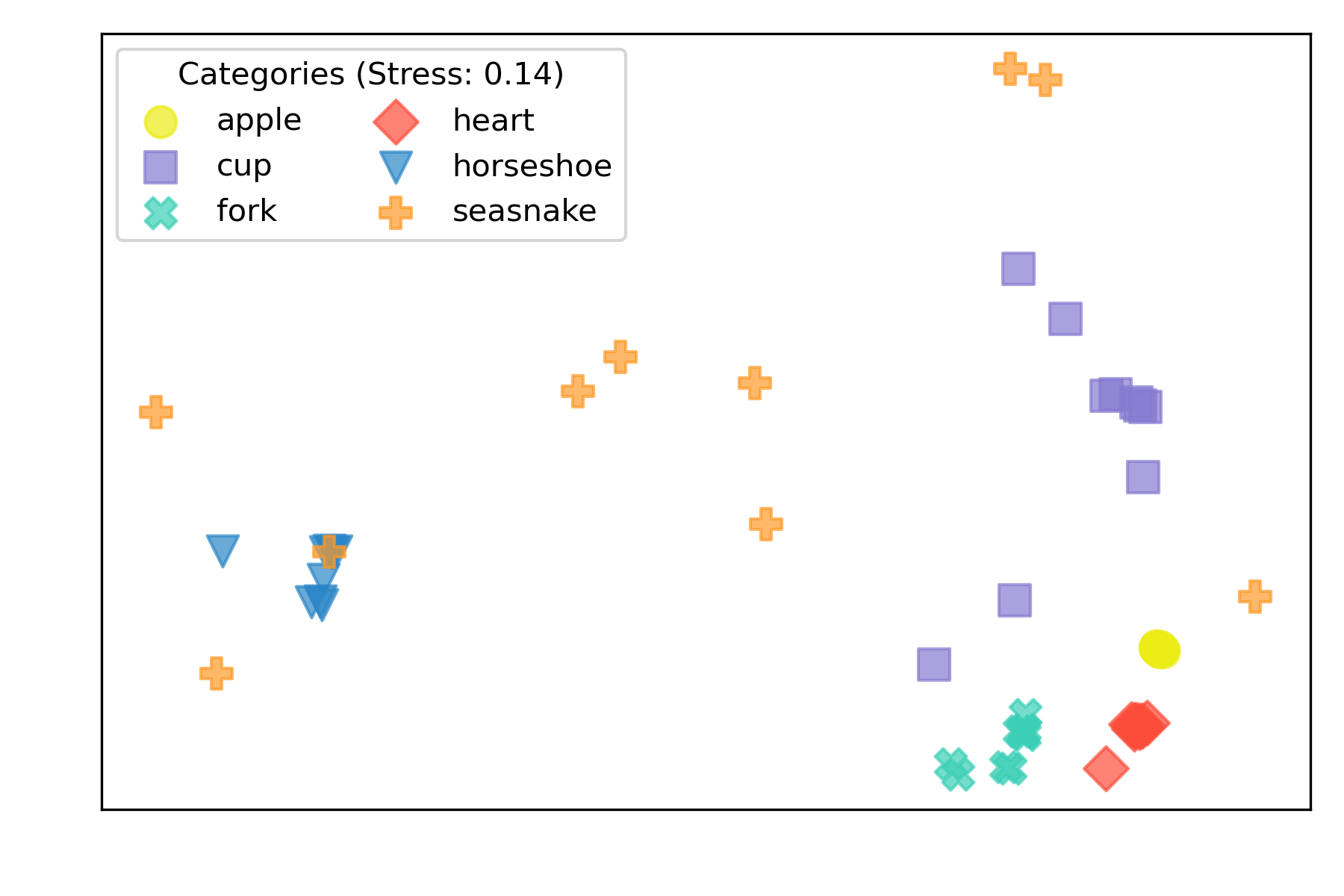}
    \caption{(Left) Example images for six categories of images from the MPEG-7 dataset.
    (Right) MDS plot computed on the pairwise optimal upper bound for the interleaving distance.}
    \label{fig:image_mds}
\end{figure}

\subsubsection{Classification of Alphabet Letters}
\label{ssec:ClassificationLetters}

Alongside the binary images, we also conduct a classification experiment on synthetically generated point clouds of six English letters: \textit{A, B, D, I, M, R}.

For each letter, we generate nine 2D point clouds by varying the number of points (500, 700, 1000) and adding Gaussian noise with standard deviation (0, 0.01, 0.02). 
The point clouds are created using geometric representations of each letter.
Different sampling densities and added noise are meant to mimic real-world variability.
\cref{fig:mds_letters} (left) demonstrates how varying the number of points and noise affects the letter \textit{A}.
We follow the same classification pipeline as used for the binary images in \cref{ssec:ClassificationMPEG7}.
Each point cloud is converted into a mapper graph using \textit{KeplerMapper} \cite{van2019kepler}, with the $y$-coordinate of each point as the lens function.
We choose parameter values so that each mapper has a single connected component, and for each node we assign the function value, or \emph{height}, to be the midpoint of the cover element it belongs to, normalized so that all node heights lie in the range $[0, 20]$.
For every mapper, we use a cover of 10 elements with $30\%$ overlap, with DBSCAN (eps=3) as the clustering algorithm.
These parameters were found to best capture the structure of the data.

\begin{figure}[h]
    \centering
    \includegraphics[width=0.45\linewidth, align = c]{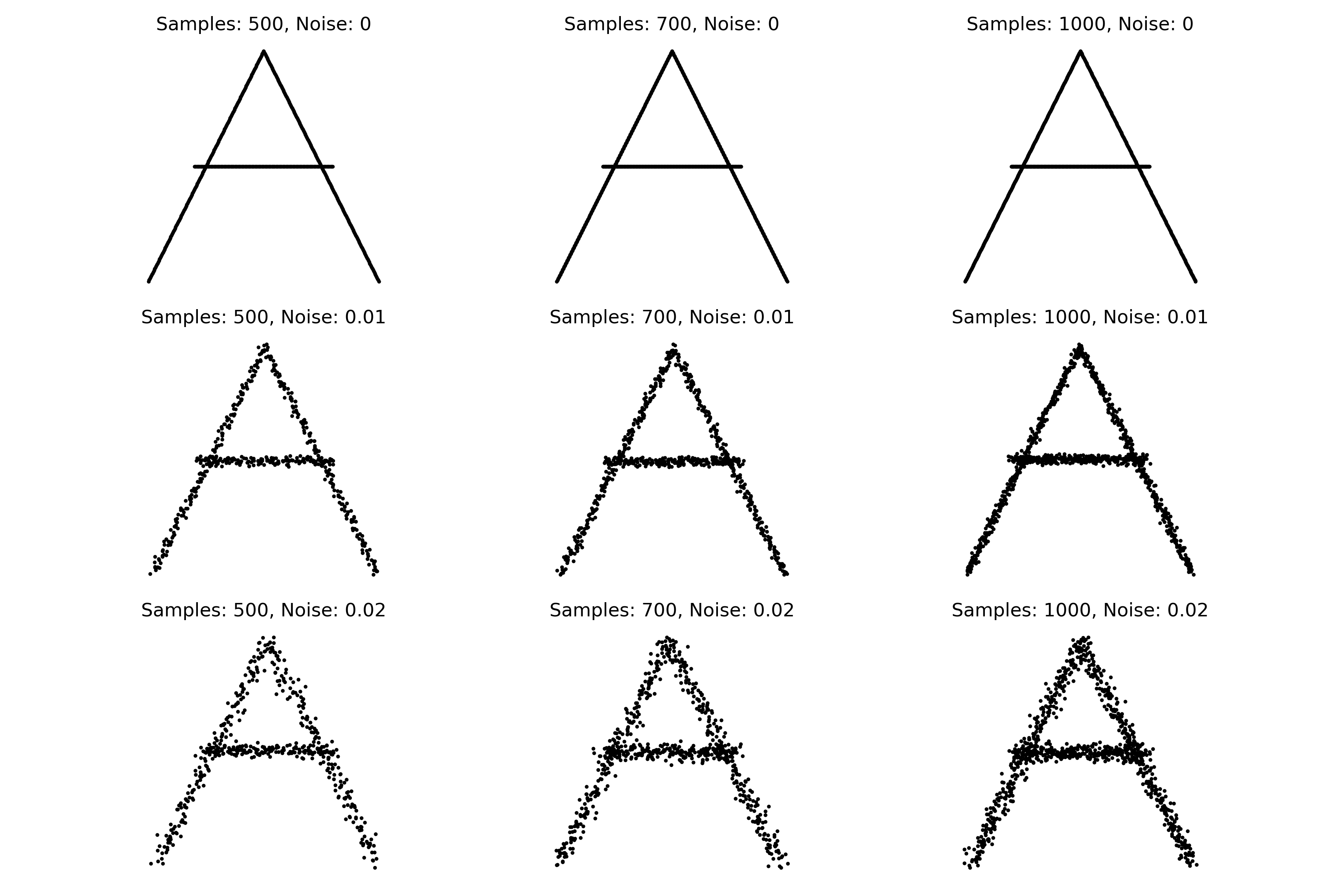}
    \includegraphics[width=0.5\linewidth, align= c]{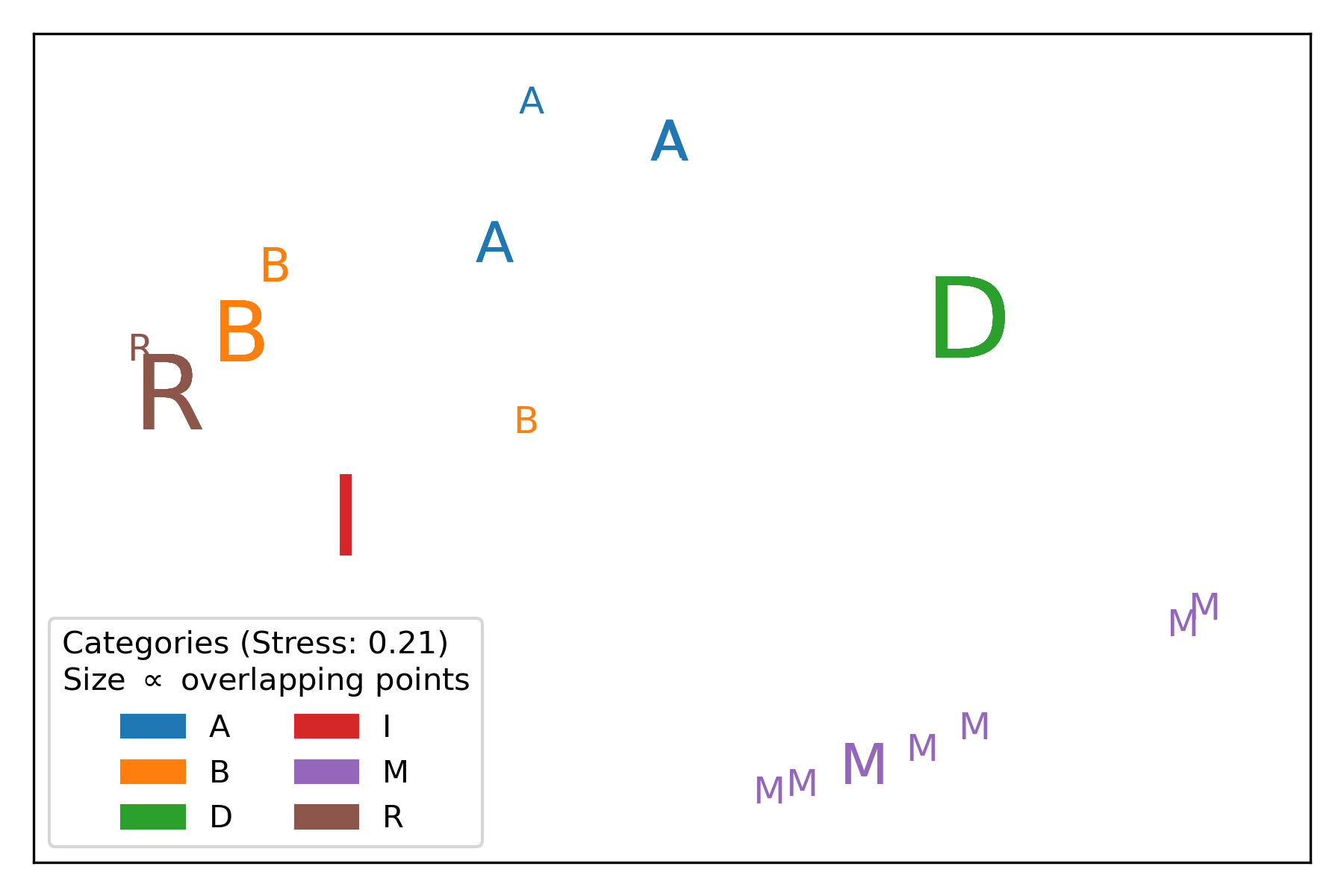}
    \caption{(Left) Point clouds generated for the letter \textit{A} with varying sample density and noise.
    (Right) MDS plot computed on the pairwise optimal upper bound for the interleaving distance.}
    \label{fig:mds_letters}
\end{figure}

We then compute the pairwise optimal upper bound of the interleaving distance and use this as a similarity metric for classification.
To visualize the clustering of different letters, we construct an MDS plot, shown in \cref{fig:image_mds} (right).
The results demonstrate that letters with distinctive shapes, such as \textit{A} and \textit{W}, form clear clusters.
However, some letters such as \textit{R} struggle to cluster distinctly, and distances between certain letters, such as between samples of \textit{R} and \textit{B}, can be small.
This is expected, given the similarity in their mapper graph structures, especially when point density is low and noise is high.

We run a classification task using the pairwise similarity metric to classify letters with a $k$-NN classifier. 
\begin{figure}[h]
    \centering
    \includegraphics[width=0.5\linewidth]{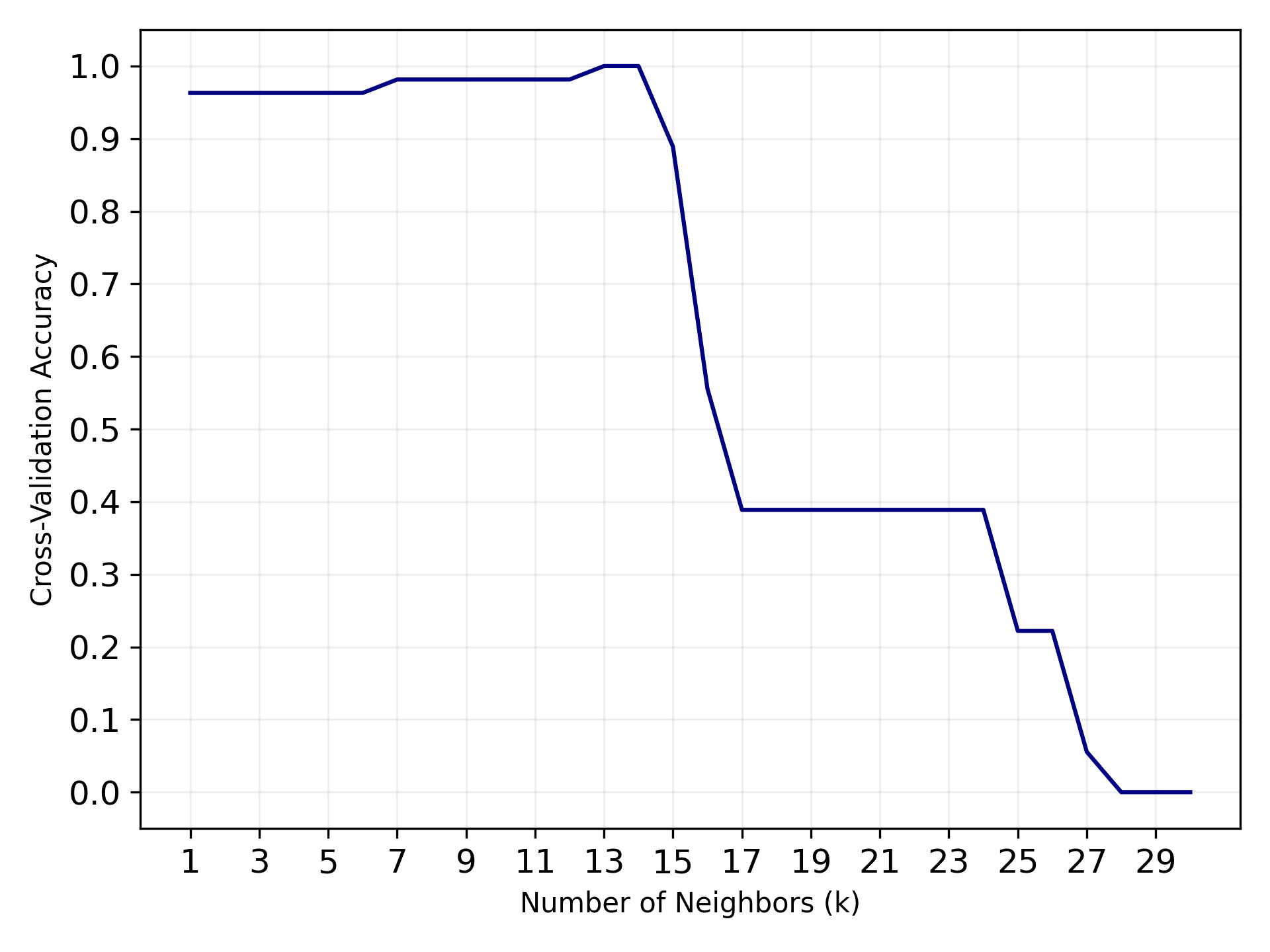}
    \includegraphics[width=0.45\textwidth]{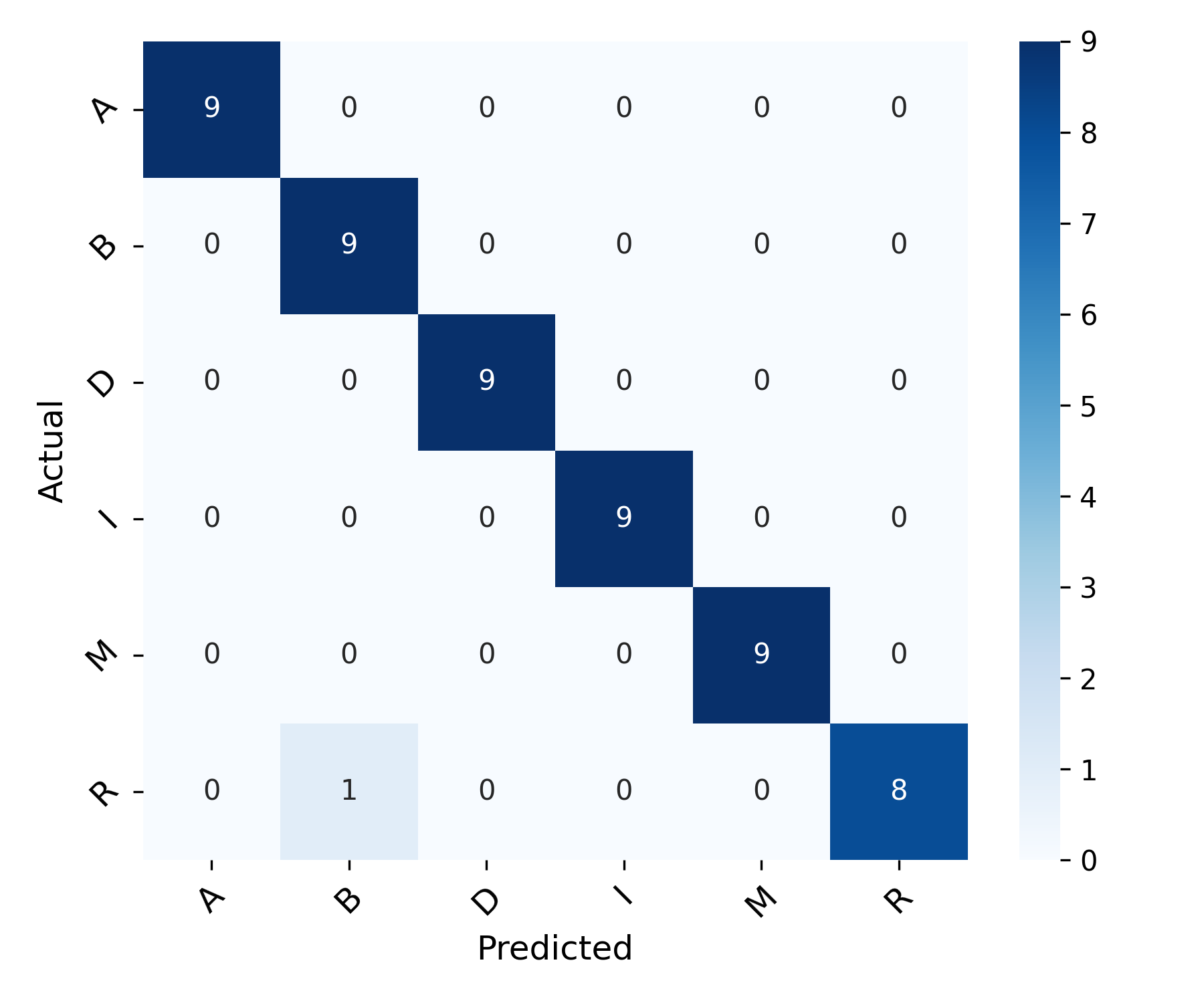}
    \caption{(Left) Accuracy of KNN classifier for different $k$ values during 5-fold cross-validation. The optimal $k = 9$ achieved $98.15\%$ accuracy.\\
    (Right) Confusion matrix showing how the actual classes match with the predicted class. Each $(i,j)$-th cell indicate number of times an image from category $i$ was classified as category $j$.}
    \label{fig:knn_accuracy_letter}
\end{figure}
The optimal value of $k$ is 9 (see \cref{fig:knn_accuracy_letter}), which yields $98.15\%$ accuracy.
We also performed leave-one-out cross-validation (LOO-CV), confirming the accuracy remains at $98.15\%$.
This shows the model is highly effective, though some misclassifications occur.
Like the MPEG-7 dataset, this letter dataset is quite small, with only nine point clouds in each of the six classes.
Expanding the experiment on a more robust dataset could improve the accuracy further.

These results further support the idea that the approximate interleaving distance captures meaningful topological variation in the data.

\section{Discussion}
\label{sec:discussion}

In this paper, we present the \emph{first} algorithm for computing a bound on the interleaving distance for mapper graphs and, to our knowledge, the \emph{first} implementation for exploring interleavings on mapper graphs. 
We do this by providing a detailed matrix formulation of interleavings, and including all code in the open source Python package \texttt{cereeberus} \cite{cereeberus}.
We deliberately avoid claiming that our method “computes’’ the interleaving distance, as our only theoretical guarantee is that the algorithm returns an upper bound. 
Nevertheless, we demonstrate that in several small examples where the true interleaving distance is known, our computed bound matches the exact value. 
We further provide a proof of concept showing that our implementation can be used to compare input images from the MPEG-7 dataset and noisy samples of letters.

This work represents an initial step toward the practical use of the interleaving distance in applications. 
There are many opportunities for optimization and speed improvements in our implementation. 
The most obvious target is the computation of the distance matrix $D$ (\cref{ssec:DistanceMatrices}). 
We observe in \cref{ssec: p2_p3_benchmark} that the loss optimization problem \cref{prob:LossComputation} takes much longer to compute and the difference in time with binary decision problem \cref{prob:BinaryDecision} only increases as the mappers get bigger.
However, the actual ILP optimization time for \cref{prob:LossComputation} is faster than \cref{prob:BinaryDecision}, giving us hope that by improving the computation of $D$, we can use the \cref{prob:LossComputation} for experiments in the future. 
Our current implementation performs a breadth-first search with a union–find data structure to compute distances independently at each level; we suspect that variants of all-pairs shortest paths algorithms could potentially compute these distances in a single sweep rather than level by level.

In future work, we hope that this code and any subsequent improvements will enable new insights and optimizations for mapper constructions. For example, the interleaving distance could serve as a measure of parameter quality by quantifying how close a computed mapper graph is to a ground truth. 
It would also be valuable to benchmark classifier performance using the interleaving distance against that of existing graph distances. 
Once computation times are reduced, a larger-scale experiment on the MPEG-7 dataset, as outlined in \cref{ssec:ClassificationMPEG7}, would be especially compelling.

Finally, because the loss function underlying our formulation is fundamentally categorical, it has the potential to translate to interleaving distances for other topological descriptors, such as merge trees and multiparameter persistence modules; some theoretical steps in this direction have already begun \cite{Olave2026}. 
We conjecture that this ILP-based viewpoint may open the door to practical computation of interleaving distances in settings where prior work has been largely theoretical.

\bibliography{refs-mapper-loss}

\begin{thebibliography}{10}

\bibitem{alvarado2025gmapper}
Enrique Alvarado, Robin Belton, Emily Fischer, Kang-Ju Lee, Sourabh Palande,
  Sarah Percival, and Emilie Purvine.
\newblock {G-Mapper}: Learning a cover in the mapper construction.
\newblock {\em SIAM Journal on Mathematics of Data Science}, 7(2):572--596,
  2025.
\newblock \href {https://doi.org/10.1137/24M1641312}
  {\path{doi:10.1137/24M1641312}}.

\bibitem{alvarado2024}
Enrique~G Alvarado, Robin Belton, Kang-Ju Lee, Sourabh Palande, Sarah Percival,
  Emilie Purvine, and Sarah Tymochko.
\newblock Any graph is a mapper graph, 2024.
\newblock URL: \url{https://arxiv.org/abs/2408.11180}, \href
  {https://arxiv.org/abs/2408.11180} {\path{arXiv:2408.11180}}.

\bibitem{Barmak2011}
Jonathan~A. Barmak.
\newblock {\em Algebraic Topology of Finite Topological Spaces and
  Applications}.
\newblock Springer, Berlin, Heidelberg, 2011.
\newblock \href {https://doi.org/10.1007/978-3-642-22003-6}
  {\path{doi:10.1007/978-3-642-22003-6}}.

\bibitem{Bauer2022}
Ulrich Bauer, H\r{a}vard~Bakke Bjerkevik, and Benedikt Fluhr.
\newblock Quasi-universality of {R}eeb graph distances.
\newblock In Xavier Goaoc and Michael Kerber, editors, {\em 38th International
  Symposium on Computational Geometry (SoCG 2022)}, volume 224 of {\em Leibniz
  International Proceedings in Informatics (LIPIcs)}, pages 14:1--14:18,
  Dagstuhl, Germany, 2022. Schloss Dagstuhl -- Leibniz-Zentrum f{\"u}r
  Informatik.
\newblock \href {https://doi.org/10.4230/LIPIcs.SoCG.2022.14}
  {\path{doi:10.4230/LIPIcs.SoCG.2022.14}}.

\bibitem{Bauer2016}
Ulrich Bauer, Barbara Di~Fabio, and Claudia Landi.
\newblock An edit distance for {R}eeb graphs.
\newblock In {\em Proceedings of the Eurographics 2016 Workshop on 3D Object
  Retrieval}, 3DOR '16, pages 27--34, Goslar, DEU, 2016. Eurographics
  Association.

\bibitem{Bauer2013}
Ulrich Bauer, Xiaoyin Ge, and Yusu Wang.
\newblock Measuring distance between {R}eeb graphs.
\newblock In {\em Proceedings of the Thirtieth Annual Symposium on
  Computational Geometry}, SOCG'14, pages 464--473, New York, NY, USA, 2014.
  Association for Computing Machinery.
\newblock \href {https://doi.org/10.1145/2582112.2582169}
  {\path{doi:10.1145/2582112.2582169}}.

\bibitem{Bauer2020}
Ulrich Bauer, Claudia Landi, and Facundo M{\'e}moli.
\newblock The {R}eeb graph edit distance is universal.
\newblock In Sergio Cabello and Danny~Z. Chen, editors, {\em 36th International
  Symposium on Computational Geometry (SoCG 2020)}, volume 164 of {\em Leibniz
  International Proceedings in Informatics (LIPIcs)}, pages 15:1--15:16,
  Dagstuhl, Germany, 2020. Schloss Dagstuhl--Leibniz-Zentrum f{\"u}r
  Informatik.
\newblock \href {https://doi.org/10.4230/LIPIcs.SoCG.2020.15}
  {\path{doi:10.4230/LIPIcs.SoCG.2020.15}}.

\bibitem{Berkouk2021}
Nicolas Berkouk and Fran\c{c}ois Petit.
\newblock Ephemeral persistence modules and distance comparison.
\newblock {\em Algebraic \& Geometric Topology}, 21(1):247--277, February 2021.
\newblock \href {https://doi.org/10.2140/agt.2021.21.247}
  {\path{doi:10.2140/agt.2021.21.247}}.

\bibitem{Beurskens2023}
Thijs Beurskens, Tim Ophelders, Bettina Speckmann, and Kevin Verbeek.
\newblock Relating interleaving and {F}r\'echet distances via ordered merge
  trees.
\newblock In {\em Proceedings of the 2025 Annual ACM-SIAM Symposium on Discrete
  Algorithms (SODA)}, pages 5027--5050, 2025.
\newblock \href {https://doi.org/10.1137/1.9781611978322.170}
  {\path{doi:10.1137/1.9781611978322.170}}.

\bibitem{bjerkevik2020computing}
H{\aa}vard~Bakke Bjerkevik, Magnus~Bakke Botnan, and Michael Kerber.
\newblock Computing the interleaving distance is {NP}-hard.
\newblock {\em Foundations of Computational Mathematics}, 20:1237--1271, 2020.
\newblock \href {https://doi.org/10.1007/s10208-019-09442-y}
  {\path{doi:10.1007/s10208-019-09442-y}}.

\bibitem{bjerkevik2017computational}
H\r{a}vard~Bakke Bjerkevik and Magnus~Bakke Botnan.
\newblock Computational complexity of the interleaving distance.
\newblock In Bettina Speckmann and Csaba~D. T\'{o}th, editors, {\em 34th
  International Symposium on Computational Geometry (SoCG 2018)}, volume~99 of
  {\em Leibniz International Proceedings in Informatics (LIPIcs)}, pages
  13:1--13:15, Dagstuhl, Germany, 2018. Schloss Dagstuhl -- Leibniz-Zentrum
  f{\"u}r Informatik.
\newblock \href {https://doi.org/10.4230/LIPIcs.SoCG.2018.13}
  {\path{doi:10.4230/LIPIcs.SoCG.2018.13}}.

\bibitem{Blumberg2023}
Andrew Blumberg and Michael Lesnick.
\newblock Universality of the homotopy interleaving distance.
\newblock {\em Transactions of the American Mathematical Society}, September
  2023.
\newblock \href {https://doi.org/10.1090/tran/8738}
  {\path{doi:10.1090/tran/8738}}.

\bibitem{Bollen2021}
Brian Bollen, Erin Chambers, Joshua~A. Levine, and Elizabeth Munch.
\newblock {R}eeb graph metrics from the ground up, 2022.
\newblock \href {https://arxiv.org/abs/2110.05631} {\path{arXiv:2110.05631}}.

\bibitem{borg2007modern}
Ingwer Borg and Patrick~JF Groenen.
\newblock {\em Modern multidimensional scaling: Theory and applications}.
\newblock Springer Science \& Business Media, 2007.
\newblock \href {https://doi.org/10.1007/978-1-4757-2711-1}
  {\path{doi:10.1007/978-1-4757-2711-1}}.

\bibitem{Botnan2020}
Magnus~Bakke Botnan, Justin Curry, and Elizabeth Munch.
\newblock A relative theory of interleavings, 2020.
\newblock \href {https://arxiv.org/abs/2004.14286} {\path{arXiv:2004.14286}}.

\bibitem{Brown2020}
Adam Brown, Omer Bobrowski, Elizabeth Munch, and Bei Wang.
\newblock Probabilistic convergence and stability of random mapper graphs.
\newblock {\em Journal of Applied and Computational Topology}, 5(1):99--140,
  December 2020.
\newblock \href {https://doi.org/10.1007/s41468-020-00063-x}
  {\path{doi:10.1007/s41468-020-00063-x}}.

\bibitem{bubenik2015metrics}
Peter Bubenik, Vin De~Silva, and Jonathan Scott.
\newblock Metrics for generalized persistence modules.
\newblock {\em Foundations of Computational Mathematics}, 15:1501--1531, 2015.
\newblock \href {https://doi.org/10.1007/s10208-014-9229-5}
  {\path{doi:10.1007/s10208-014-9229-5}}.

\bibitem{Carriere2018}
Mathieu Carri{\`e}re, Bertrand Michel, and Steve Oudot.
\newblock Statistical analysis and parameter selection for mapper.
\newblock {\em Journal of Machine Learning Research}, 19(12):1--39, 2018.
\newblock URL: \url{http://jmlr.org/papers/v19/17-291.html}.

\bibitem{Carriere2017}
Mathieu Carri{\`e}re and Steve Oudot.
\newblock Local equivalence and intrinsic metrics between {R}eeb graphs.
\newblock In Boris Aronov and Matthew~J. Katz, editors, {\em 33rd International
  Symposium on Computational Geometry (SoCG 2017)}, volume~77 of {\em Leibniz
  International Proceedings in Informatics (LIPIcs)}, pages 25:1--25:15,
  Dagstuhl, Germany, March 2017. Schloss Dagstuhl--Leibniz-Zentrum fuer
  Informatik.
\newblock \href {https://doi.org/10.4230/LIPIcs.SoCG.2017.25}
  {\path{doi:10.4230/LIPIcs.SoCG.2017.25}}.

\bibitem{Carriere2017b}
Mathieu Carri{\`e}re and Steve Oudot.
\newblock Structure and stability of the one-dimensional mapper.
\newblock {\em Foundations of Computational Mathematics}, 18(6):1333--1396,
  October 2017.
\newblock \href {https://doi.org/10.1007/s10208-017-9370-z}
  {\path{doi:10.1007/s10208-017-9370-z}}.

\bibitem{chalapathi2021adaptive}
Nithin Chalapathi, Youjia Zhou, and Bei Wang.
\newblock Adaptive covers for mapper graphs using information criteria.
\newblock In {\em 2021 IEEE International Conference on Big Data (Big Data)},
  pages 3789--3800, 2021.
\newblock \href {https://doi.org/10.1109/BigData52589.2021.9671324}
  {\path{doi:10.1109/BigData52589.2021.9671324}}.

\bibitem{Chambers2025}
Erin~W. Chambers and Guangyu Meng.
\newblock A stable and theoretically grounded gromov-wasserstein distance for
  reeb graph comparison using persistence images.
\newblock July 2025.
\newblock \href {https://arxiv.org/abs/2507.01171} {\path{arXiv:2507.01171}},
  \href {https://doi.org/10.48550/ARXIV.2507.01171}
  {\path{doi:10.48550/ARXIV.2507.01171}}.

\bibitem{Chambers2021}
Erin~Wolf Chambers, Elizabeth Munch, and Tim Ophelders.
\newblock A family of metrics from the truncated smoothing of {R}eeb graphs.
\newblock In Kevin Buchin and \'{E}ric Colin~de Verdi\`{e}re, editors, {\em
  37th International Symposium on Computational Geometry (SoCG 2021)}, volume
  189 of {\em Leibniz International Proceedings in Informatics (LIPIcs)}, pages
  22:1--22:17, Dagstuhl, Germany, 2021. Schloss Dagstuhl -- Leibniz-Zentrum
  f{\"u}r Informatik.
\newblock \href {https://doi.org/10.4230/LIPIcs.SoCG.2021.22}
  {\path{doi:10.4230/LIPIcs.SoCG.2021.22}}.

\bibitem{chambers2023bounding}
Erin~Wolf Chambers, Elizabeth Munch, Sarah Percival, and Bei Wang.
\newblock Bounding the interleaving distance for mapper graphs with a loss
  function.
\newblock {\em Journal of Applied and Computational Topology}, 9(3), July 2025.
\newblock \href {https://doi.org/10.1007/s41468-025-00215-x}
  {\path{doi:10.1007/s41468-025-00215-x}}.

\bibitem{Chazal2009}
Fr{\'e}d{\'e}ric Chazal, David Cohen-Steiner, Marc Glisse, Leonidas~J. Guibas,
  and Steve~Y. Oudot.
\newblock Proximity of persistence modules and their diagrams.
\newblock In {\em Proceedings of the 25th annual symposium on Computational
  geometry}, SoCG '09, pages 237--246, New York, NY, USA, 2009. ACM.
\newblock \href {https://doi.org/10.1145/1542362.1542407}
  {\path{doi:10.1145/1542362.1542407}}.

\bibitem{cover1967nearest}
Thomas Cover and Peter Hart.
\newblock Nearest neighbor pattern classification.
\newblock {\em IEEE transactions on information theory}, 13(1):21--27, 1967.
\newblock \href {https://doi.org/10.1109/TIT.1967.1053964}
  {\path{doi:10.1109/TIT.1967.1053964}}.

\bibitem{Cruz2019}
Joshua Cruz.
\newblock Metric limits in categories with a flow, 2019.
\newblock \href {https://arxiv.org/abs/1901.04828} {\path{arXiv:1901.04828}}.

\bibitem{curry2013sheaves}
Justin Curry.
\newblock {\em Sheaves, Cosheaves and Applications}.
\newblock PhD thesis, University of Pennsylvania, 2014.
\newblock \href {https://arxiv.org/abs/1303.3255} {\path{arXiv:1303.3255}}.

\bibitem{Curry2022}
Justin Curry, Haibin Hang, Washington Mio, Tom Needham, and Osman~Berat Okutan.
\newblock Decorated merge trees for persistent topology.
\newblock {\em Journal of Applied and Computational Topology}, 6(3):371--428,
  March 2022.
\newblock \href {https://doi.org/10.1007/s41468-022-00089-3}
  {\path{doi:10.1007/s41468-022-00089-3}}.

\bibitem{deSilva2016}
Vin de~Silva, Elizabeth Munch, and Amit Patel.
\newblock Categorified {R}eeb graphs.
\newblock {\em Discrete \& Computational Geometry}, pages 1--53, 2016.
\newblock \href {https://doi.org/10.1007/s00454-016-9763-9}
  {\path{doi:10.1007/s00454-016-9763-9}}.

\bibitem{deSilva2018}
Vin {de Silva}, Elizabeth Munch, and Anastasios Stefanou.
\newblock Theory of interleavings on categories with a flow.
\newblock {\em Theory and Applications of Categories}, 33(21):583--607, 2018.
\newblock URL: \url{http://www.tac.mta.ca/tac/volumes/33/21/33-21.pdf}.

\bibitem{FarahbakhshTouli2019}
Elena {Farahbakhsh Touli} and Yusu Wang.
\newblock {FPT}-algorithms for computing {G}romov-{H}ausdorff and interleaving
  distances between trees.
\newblock In Michael~A. Bender, Ola Svensson, and Grzegorz Herman, editors,
  {\em 27th Annual European Symposium on Algorithms (ESA 2019)}, volume 144 of
  {\em Leibniz International Proceedings in Informatics (LIPIcs)}, pages
  83:1--83:14, Dagstuhl, Germany, 2019. Schloss Dagstuhl--Leibniz-Zentrum fuer
  Informatik.
\newblock \href {https://doi.org/10.4230/LIPIcs.ESA.2019.83}
  {\path{doi:10.4230/LIPIcs.ESA.2019.83}}.

\bibitem{glpk}
GLPK GNU Project Free Software~Foundation (FSF).
\newblock URL: \url{https://www.gnu.org/software/glpk/}.

\bibitem{Gao2010}
Xinbo Gao, Bing Xiao, Dacheng Tao, and Xuelong Li.
\newblock A survey of graph edit distance.
\newblock {\em Pattern Analysis and applications}, 13(1):113--129, 2010.
\newblock \href {https://doi.org/10.1007/s10044-008-0141-y}
  {\path{doi:10.1007/s10044-008-0141-y}}.

\bibitem{Garey1979}
Michael~R. Garey and David~S. Johnson.
\newblock {\em Computers and Intractibility: A Guide to the Theory of
  NP-Completeness}.
\newblock W.H. Freeman and Co., 1979.

\bibitem{Gasparovic2024}
Ellen Gasparovic, Elizabeth Munch, Steve Oudot, Katharine Turner, Bei Wang, and
  Yusu Wang.
\newblock Intrinsic interleaving distance for merge trees.
\newblock {\em La Matematica}, December 2024.
\newblock \href {https://doi.org/10.1007/s44007-024-00143-9}
  {\path{doi:10.1007/s44007-024-00143-9}}.

\bibitem{losspapergithubexperiments}
Ishika Ghosh.
\newblock \url{https://github.com/ishikaghosh2201/Mapper-Loss-Computation/},
  2026.

\bibitem{gurobi}
{Gurobi Optimization, LLC}.
\newblock Gurobi optimizer reference manual, 2024.
\newblock URL: \url{https://www.gurobi.com}.

\bibitem{hart2009python}
William~E Hart.
\newblock Python optimization modeling objects ({P}yomo).
\newblock In {\em Operations Research and Cyber-Infrastructure}, pages 3--19.
  Springer, 2009.
\newblock \href {https://doi.org/10.1007/978-0-387-88843-9_1}
  {\path{doi:10.1007/978-0-387-88843-9_1}}.

\bibitem{jeannin1999}
Sylvie Jeannin and Miroslaw Bober.
\newblock Description of core experiments for {MPEG}-7 motion/shape.
\newblock {\em MPEG-7, ISO/IEC/JTC1/SC29/WG11/MPEG99 N}, 2690, 1999.

\bibitem{Kim2023}
Woojin Kim, Facundo M{\'e}moli, and Anastasios Stefanou.
\newblock Interleaving by parts: Join decompositions of interleavings and
  join-assemblage of geodesics.
\newblock {\em Order}, 41(2):497--537, September 2023.
\newblock \href {https://doi.org/10.1007/s11083-023-09643-9}
  {\path{doi:10.1007/s11083-023-09643-9}}.

\bibitem{Kim2023a}
Woojin Kim and Facundo Mémoli.
\newblock Extracting persistent clusters in dynamic data via möbius inversion.
\newblock {\em Discrete \& Computational Geometry}, 71(4):1276--1342, October
  2023.
\newblock \href {https://arxiv.org/abs/1712.04064} {\path{arXiv:1712.04064}},
  \href {https://doi.org/10.1007/s00454-023-00590-1}
  {\path{doi:10.1007/s00454-023-00590-1}}.

\bibitem{Lan2024}
Fangfei Lan, Salman Parsa, and Bei Wang.
\newblock Labeled interleaving distance for {R}eeb graphs.
\newblock {\em Journal of Applied and Computational Topology}, 8(8):2367--2399,
  2024.
\newblock \href {https://doi.org/10.1007/s41468-024-00193-6}
  {\path{doi:10.1007/s41468-024-00193-6}}.

\bibitem{Lesnick2015}
Michael Lesnick.
\newblock The theory of the interleaving distance on multidimensional
  persistence modules.
\newblock {\em Foundations of Computational Mathematics}, 15(3):613--650, 2015.
\newblock \href {https://doi.org/10.1007/s10208-015-9255-y}
  {\path{doi:10.1007/s10208-015-9255-y}}.

\bibitem{cbc}
Robin Lougee-Heimer.
\newblock The {C}ommon {O}ptimization {IN}terface for {O}perations {R}esearch:
  Promoting open-source software in the operations research community.
\newblock {\em IBM Journal of Research and Development}, 47(1):57--66, 2003.
\newblock \href {https://doi.org/10.1147/rd.471.0057}
  {\path{doi:10.1147/rd.471.0057}}.

\bibitem{Meehan2017}
Killian Meehan and David Meyer.
\newblock Interleaving distance as a limit, 2017.
\newblock \href {https://arxiv.org/abs/1710.11489} {\path{arXiv:1710.11489}}.

\bibitem{mitchell2011pulp}
Stuart Mitchell, Michael O'Sullivan, and Iain Dunning.
\newblock {P}u{LP}: a linear programming toolkit for {P}ython.
\newblock {\em The University of Auckland, Auckland, New Zealand}, 65:25, 2011.

\bibitem{Morozov2013}
Dmitriy Morozov, Kenes Beketayev, and Gunther Weber.
\newblock Interleaving distance between merge trees.
\newblock In {\em Proceedings of TopoInVis}, 2013.

\bibitem{cereeberus}
Elizabeth Munch, Danielle Barnes, Ishika Ghosh, and Elena~Xinyi Wang.
\newblock {ceREEBerus}: Reeb graph computations in python.
\newblock \url{https://munchlab.github.io/ceREEBerus/}, 2026.

\bibitem{Munch2019}
Elizabeth Munch and Anastasios Stefanou.
\newblock The $\ell_\infty$-cophenetic metric for phylogenetic trees as an
  interleaving distance.
\newblock In {\em Association for Women in Mathematics Series}, pages 109--127.
  Springer International Publishing, 2019.
\newblock \href {https://doi.org/10.1007/978-3-030-11566-1_5}
  {\path{doi:10.1007/978-3-030-11566-1_5}}.

\bibitem{Munch2016}
Elizabeth Munch and Bei Wang.
\newblock Convergence between categorical representations of {R}eeb space and
  {M}apper.
\newblock In S{\'a}ndor Fekete and Anna Lubiw, editors, {\em 32nd International
  Symposium on Computational Geometry (SoCG 2016)}, volume~51 of {\em Leibniz
  International Proceedings in Informatics (LIPIcs)}, pages 53:1--53:16,
  Dagstuhl, Germany, 2016. Schloss Dagstuhl--Leibniz-Zentrum fuer Informatik.
\newblock \href {https://doi.org/http://dx.doi.org/10.4230/LIPIcs.SoCG.2016.53}
  {\path{doi:http://dx.doi.org/10.4230/LIPIcs.SoCG.2016.53}}.

\bibitem{nlab:unnatural_transformation}
{nLab authors}.
\newblock Unnatural transformation.
\newblock \url{https://ncatlab.org/nlab/show/unnatural+transformation},
  November 2023.
\newblock
  \href{https://ncatlab.org/nlab/revision/unnatural+transformation/1}{Revision
  1}.

\bibitem{Olave2026}
Astrid~A. Olave and Elizabeth Munch.
\newblock Bounding the interleaving distance on concrete categories using a
  loss function, 2026.
\newblock URL: \url{https://arxiv.org/abs/2601.09034}, \href
  {https://arxiv.org/abs/2601.09034} {\path{arXiv:2601.09034}}.

\bibitem{Oulhaj2025}
Ziyad Oulhaj, Mathieu Carrière, and Bertrand Michel.
\newblock Gromov-wasserstein bound between reeb and mapper graphs.
\newblock June 2025.
\newblock \href {https://arxiv.org/abs/2506.02810} {\path{arXiv:2506.02810}},
  \href {https://doi.org/10.48550/ARXIV.2506.02810}
  {\path{doi:10.48550/ARXIV.2506.02810}}.

\bibitem{Pegoraro2021}
Matteo Pegoraro.
\newblock A graph-matching formulation of the interleaving distance between
  merge trees.
\newblock {\em AIMS Mathematics}, 10(6):13025--13081, 2025.
\newblock \href {https://doi.org/10.3934/math.2025586}
  {\path{doi:10.3934/math.2025586}}.

\bibitem{Purvine2023}
Emilie Purvine, Davis Brown, Brett Jefferson, Cliff Joslyn, Brenda Praggastis,
  Archit Rathore, Madelyn Shapiro, Bei Wang, and Youjia Zhou.
\newblock Experimental observations of the topology of convolutional neural
  network activations.
\newblock In {\em Proceedings of the Thirty-Seventh AAAI Conference on
  Artificial Intelligence and Thirty-Fifth Conference on Innovative
  Applications of Artificial Intelligence and Thirteenth Symposium on
  Educational Advances in Artificial Intelligence}, AAAI'23/IAAI'23/EAAI'23.
  AAAI Press, 2023.
\newblock \href {https://doi.org/10.1609/aaai.v37i8.26134}
  {\path{doi:10.1609/aaai.v37i8.26134}}.

\bibitem{Rathore2023}
Archit Rathore, Yichu Zhou, Vivek Srikumar, and Bei Wang.
\newblock Topobert: Exploring the topology of fine-tuned word representations.
\newblock {\em Information Visualization}, 22:147387162311686, 05 2023.
\newblock \href {https://doi.org/10.1177/14738716231168671}
  {\path{doi:10.1177/14738716231168671}}.

\bibitem{Reeb}
Georges Reeb.
\newblock Sur les points singuliers d'une forme de {P}faff compl{\`e}tement
  int{\'e}grable ou d'une fonction num{\'e}rique.
\newblock {\em Comptes Rendus de l'Acad{\'e}mie des Sciences}, 222:847--849,
  1946.

\bibitem{riehl2017category}
Emily Riehl.
\newblock {\em Category theory in context}.
\newblock Courier Dover Publications, 2017.

\bibitem{Rizvi2017}
Abbas~H Rizvi, Pablo~G Camara, Elena~K Kandror, Thomas~J Roberts, Ira Schieren,
  Tom Maniatis, and Raul Rabadan.
\newblock Single-cell topological {RNA}-seq analysis reveals insights into
  cellular differentiation and development.
\newblock {\em Nature Biotechnology}, 35(6):551--560, May 2017.
\newblock \href {https://doi.org/10.1038/nbt.3854}
  {\path{doi:10.1038/nbt.3854}}.

\bibitem{Robinson2020}
Michael Robinson.
\newblock Assignments to sheaves of pseudometric spaces.
\newblock {\em {Compositionality}}, 2, June 2020.
\newblock \href {https://doi.org/10.32408/compositionality-2-2}
  {\path{doi:10.32408/compositionality-2-2}}.

\bibitem{Saggar2018}
Manish Saggar, Olaf Sporns, Javier Gonzalez-Castillo, Peter~A. Bandettini,
  Gunnar Carlsson, Gary Glover, and Allan~L. Reiss.
\newblock Towards a new approach to reveal dynamical organization of the brain
  using topological data analysis.
\newblock {\em Nature Communications}, 9(1), April 2018.
\newblock \href {https://doi.org/10.1038/s41467-018-03664-4}
  {\path{doi:10.1038/s41467-018-03664-4}}.

\bibitem{Scoccola2020}
Luis Scoccola.
\newblock {\em Locally persistent categories andmetric properties of
  interleaving distances}.
\newblock PhD thesis, Western University, 2020.

\bibitem{Singh2007}
Gurjeet Singh, Facundo M\'emoli, and Gunnar Carlsson.
\newblock Topological methods for the analysis of high dimensional data sets
  and {3D} object recognition.
\newblock In {\em Eurographics Symposium on Point-Based Graphics}, 2007.
\newblock \href {https://doi.org/10.2312/SPBG/SPBG07/091-100}
  {\path{doi:10.2312/SPBG/SPBG07/091-100}}.

\bibitem{van2019kepler}
Hendrik~Jacob Van~Veen, Nathaniel Saul, David Eargle, and Sam~W Mangham.
\newblock Kepler mapper: A flexible python implementation of the mapper
  algorithm.
\newblock {\em Journal of Open Source Software}, 4(42):1315, 2019.
\newblock \href {https://doi.org/10.21105/joss.01315}
  {\path{doi:10.21105/joss.01315}}.

\bibitem{Yan2021a}
Lin Yan, Talha~Bin Masood, Raghavendra Sridharamurthy, Farhan Rasheed, Vijay
  Natarajan, Ingrid Hotz, and Bei Wang.
\newblock Scalar field comparison with topological descriptors: Properties and
  applications for scientific visualization.
\newblock {\em Computer Graphics Forum}, 40(3):599--633, jun 2021.
\newblock \href {https://doi.org/10.1111/cgf.14331}
  {\path{doi:10.1111/cgf.14331}}.

\bibitem{Yan2019}
Lin Yan, Yusu Wang, Elizabeth Munch, Ellen Gasparovic, and Bei Wang.
\newblock A structural average of labeled merge trees for uncertainty
  visualization.
\newblock {\em IEEE Transactions on Visualization and Computer Graphics},
  26(1):832--842, 2020.
\newblock \href {https://doi.org/10.1109/TVCG.2019.2934242}
  {\path{doi:10.1109/TVCG.2019.2934242}}.

\bibitem{zhou2023comparing}
Youjia Zhou, Helen Jenne, Davis Brown, Madelyn Shapiro, Brett Jefferson, Cliff
  Joslyn, Gregory Henselman-Petrusek, Brenda Praggastis, Emilie Purvine, and
  Bei Wang.
\newblock Comparing mapper graphs of artificial neuron activations.
\newblock In {\em 2023 Topological Data Analysis and Visualization
  (TopoInVis)}, pages 41--50, 2023.
\newblock \href {https://doi.org/10.1109/TopoInVis60193.2023.00011}
  {\path{doi:10.1109/TopoInVis60193.2023.00011}}.

\end{thebibliography}

\appendix
\appendix

\end{document}